\documentclass{article}
\usepackage{kr}

\usepackage{times}
\usepackage{soul}
\usepackage{url}
\usepackage[hidelinks]{hyperref}
\usepackage[utf8]{inputenc}
\usepackage[small]{caption}
\usepackage{graphicx}
\usepackage{amsmath}
\usepackage{amsthm}
\usepackage{booktabs}
\usepackage{algorithm}
\usepackage{algorithmic}
\usepackage{paralist}
\usepackage[T1]{fontenc}
\usepackage{amsfonts}
\usepackage{amssymb}
\usepackage{stmaryrd}
\usepackage{wasysym}
\usepackage{bbm}
\usepackage{rotating}
\usepackage{pifont}
\usepackage{array}
\usepackage{multirow}
\usepackage{tikz}
\usepackage{color}
\usepackage{xcolor}
\usepackage[inline]{enumitem}
\usepackage{xspace}
\usepackage{thm-restate}
\usepackage{phonetic}
\usepackage{hyperref}
\usepackage{todonotes}
\usepackage{thmtools}

\usepackage{latexsym}

\newtheorem{example}{Example}
\newtheorem{theorem}{Theorem}

\newtheorem{proposition}[theorem]{Proposition}

\newtheorem{lemma}[theorem]{Lemma}

\newtheorem{claim}{Claim}

\newtheorem{remark}{Remark}

\newif\ifappendix
\appendixtrue

\newif\ifshort
\shortfalse

\newcommand{\nb}[1]{\hbox to 0pt{\textcolor{red}{\bf!}}%
\marginpar[\parbox{20mm}{\scriptsize\textcolor{red}{\raggedleft #1}}]%
{\parbox{12mm}{\scriptsize\textcolor{red}{\raggedright #1}}}}

\newcommand{\avec}[1]{\boldsymbol{#1}}
\newcommand{\extN}{\mathbb{N}_\infty}

\newcommand{\p}{\varphi}

\newcommand{\eset}{\emptyset}

\newcommand{\Int}{\ensuremath{\mathcal{I}}\xspace}
\newcommand{\I}{\ensuremath{\Imc}}

\newcommand{\lp}{{\langle}}
\newcommand{\rp}{{\rangle}}

\newcommand{\red}{^{\ast}}

\newcommand{\sttr}[2]{\ensuremath{{\pi}_{#1}(#2)}\xspace}

\newcommand{\md}{\ensuremath{\mathit{md}}}

\newcommand{\rpath}{\mathsf{rp}}
\newcommand{\rpathp}{\overline{\rpath}}
\newcommand{\rpathw}{\mathsf{rw}}

\newcommand{\sig}[1]{\ensuremath{\Sigma_{#1}}\xspace}

\newcommand{\sub}[1]{\ensuremath{\mathsf{sub}(#1)}\xspace}

\newcommand{\Var}{\ensuremath{\mathsf{Var}}\xspace}

\newcommand{\con}[1]{\ensuremath{\mathsf{con}(#1)}\xspace}

\newcommand{\assign}{\ensuremath{\sigma}\xspace}

\newcommand{\tvalue}[3]{\ensuremath{{#1}^{\Mmf, {#2}}_{#3}}\xspace}

\newcommand{\Next}{{\ensuremath{\raisebox{0.25ex}{\text{\scriptsize{$\bigcirc$}}}}}}

\newcommand{\monodic}{\ensuremath{\mathop{\ooalign{$\Box$ \cr \kern0.57ex \raisebox{0.2ex}{\scalebox{0.55}{$1$}}}\rule{0pt}{1.5ex} \kern-0.7ex}}\xspace}

\newcommand{\Boxuni}{\ensuremath{\mathop{\ooalign{$\Box$ \cr \kern0.57ex \raisebox{0.2ex}{\scalebox{0.55}{$u$}}}\rule{0pt}{1.5ex} \kern-0.7ex}}\xspace}

\DeclareRobustCommand{\zerodic}{\ensuremath{\mathop{\ooalign{$\Box$ \cr \kern0.57ex \raisebox{0.2ex}{\scalebox{0.55}{$0$}}}\rule{0pt}{1.5ex} \kern-0.7ex}}\xspace}

\newcommand{\NC}{\ensuremath{{\sf N_C}}\xspace}
\newcommand{\NI}{\ensuremath{{\sf N_I}}\xspace}

\newcommand{\NR}{\ensuremath{{\sf N_R}}\xspace}
\newcommand{\NPr}{\ensuremath{\sf{N}_{\sf{P}}\xspace}}

\newcommand{\DL}{\ensuremath{\mathcal{DL}}\xspace}

\newcommand{\ALC}{\ensuremath{\smash{\mathcal{ALC}}}\xspace}
\newcommand{\ALCO}{\ensuremath{\smash{\mathcal{ALCO}}}\xspace}

\newcommand{\ALCOu}{\ensuremath{\smash{\mathcal{ALCO}_{\!u}}}\xspace}

\newcommand{\ALCOud}{\ensuremath{\smash{\mathcal{ALCO}_{\!u}^{\defdes}}}\xspace}

\newcommand{\ALCOd}{\ensuremath{\smash{\mathcal{ALCO}^{\defdes}}}\xspace}

\newcommand{\ELO}{\ensuremath{\smash{\mathcal{ELO}}}\xspace}

\newcommand{\ELOu}{\ensuremath{\smash{\mathcal{ELO}_{\!u}}}\xspace}
\newcommand{\ELOd}{\ensuremath{\smash{\mathcal{ELO}^{\defdes}}}\xspace}
\newcommand{\ELOud}{\ensuremath{\smash{\mathcal{ELO}_{\!u}^{\defdes}}}\xspace}

\newcommand{\quasimod}{\ensuremath{\Qmf}\xspace}

\newcommand{\runs}{\ensuremath{\Rmf}\xspace}

\newcommand{\funcand}{\ensuremath{\boldsymbol{q}}\xspace}

\newcommand{\settp}{\ensuremath{\boldsymbol{T}}\xspace}

\newcommand{\contp}{\ensuremath{{\boldsymbol{t}}}\xspace}

\newcommand{\ML}{\ensuremath{\smash{\mathcal{ML}}}\xspace}

\newcommand{\QMLdl}{\ensuremath{\smash{\mathcal{QML}^{\defdes}_{\lambda}}}\xspace}

\newcommand{\MLALCO}{\ensuremath{\smash{\ML^{n}_{\mathcal{ALCO}}}}\xspace}

\newcommand{\MLALCud}{\ensuremath{\smash{\ML^{n}_{\mathcal{ALC}_{\!u}^{\defdes}}}}\xspace}
\newcommand{\MLALCOd}{\ensuremath{\smash{\ML^{n}_{\mathcal{ALCO}^{\defdes}}}}\xspace}
\newcommand{\MLALCOu}{\ensuremath{\smash{\ML^{n}_{\ALCOu}}}\xspace}
\newcommand{\MLALCOud}{\ensuremath{\smash{\ML^{n}_{\ALCOud}}}\xspace}

\newcommand{\K}{\ensuremath{\mathbf{K}}}

\newcommand{\GL}{\ensuremath{\mathbf{GL}}}
\newcommand{\Grz}{\ensuremath{\mathbf{Grz}}}

\newcommand{\Kfn}{\ensuremath{\K\!\boldsymbol{f}^{\ast n}}}

\newcommand{\Sfive}{\ensuremath{\mathbf{S5}}}

\newcommand{\TL}{\ensuremath{\smash{\mathcal{TL}}}\xspace}

\newcommand{\TLALCO}{\ensuremath{\smash{\mathcal{TL}_{\mathcal{ALCO}}}}\xspace}
\newcommand{\TLnALCO}{\ensuremath{\smash{\mathcal{TL}^{\textstyle\circ}_{\mathcal{ALCO}}}}\xspace}

\newcommand{\TLALCOd}{\ensuremath{\smash{\mathcal{TL}_{\mathcal{ALCO}^{\defdes}}}}\xspace}
\newcommand{\TLALCOud}{\ensuremath{\smash{\mathcal{TL}_{\mathcal{ALCO}_{\!u}^{\defdes}}}}\xspace}
\newcommand{\TLALCOu}{\ensuremath{\smash{\mathcal{TL}_{\mathcal{ALCO}_{\!u}}}}\xspace}
\newcommand{\TLnALCOu}{\ensuremath{\smash{\mathcal{TL}^{\textstyle\circ}_{\mathcal{ALCO}_{\!u}}}}\xspace}

\newcommand{\LTL}{\ensuremath{\mathbf{LTL}}}
\newcommand{\LTLfo}{\ensuremath{\LTLf^{\textstyle\circ}}}
\newcommand{\LTLo}{\ensuremath{\LTL^{\textstyle\circ}}}
\newcommand{\LTLd}{\ensuremath{\LTL^\Diamond}}
\newcommand{\LTLf}{\ensuremath{\LTL\!{\text{\ensuremath{\boldsymbol{f}}}}}}
\newcommand{\LTLfd}{\ensuremath{\LTLf^\Diamond}}

\newcommand{\LTLALCOud}{\ensuremath{\smash{\LTL_{\mathcal{ALCO}_{\!u}^{\defdes}}}}\xspace}

\newcommand{\LTLfALCOud}{\ensuremath{\smash{\LTLf_{\!\mathcal{ALCO}_{\!u}^{\defdes}}}}\xspace}

\newcommand{\LTLfdALCOu}{\ensuremath{\smash{\LTLf^{\Diamond}_{\!\ALCOu}}}\xspace}

\newcommand{\Ex}{\ensuremath{\mathsf{Ex}}\xspace}

\newcommand{\exrel}{\ensuremath{\downarrow\Ex}\xspace}

\newcommand{\defdes}{\ensuremath{\smash{\iota}}\xspace}

\newcommand{\PTime}{\textsc{PTime}}

\newcommand{\ExpTime}{\textsc{ExpTime}}
\newcommand{\NExpTime}{\textsc{NExpTime}}

\newcommand{\Fmf}{\ensuremath{\mathfrak{F}}\xspace}

\newcommand{\Mmf}{\ensuremath{\mathfrak{M}}\xspace}

\newcommand{\Qmf}{\ensuremath{\mathfrak{Q}}\xspace}
\newcommand{\Rmf}{\ensuremath{\mathfrak{R}}\xspace}

\newcommand{\Tmf}{\ensuremath{\mathfrak{T}}\xspace}

\newcommand{\Cmc}{\ensuremath{\mathcal{C}}\xspace}

\newcommand{\Imc}{\ensuremath{\mathcal{I}}\xspace}

\newcommand{\Omc}{\ensuremath{\mathcal{O}}\xspace}

\newcommand{\Smc}{\ensuremath{\mathcal{S}}\xspace}

\newcommand{\Nbl}{\ensuremath{\mathbb{N}}\xspace}

\title{Non-Rigid Designators in Modal and Temporal Free Description Logics}

\author{%
Alessandro Artale$^1$\and
Roman Kontchakov$^2$\and
Andrea Mazzullo$^{1,3}$\and
Frank Wolter$^4$ \\
\affiliations
$^1$Free University of Bozen-Bolzano\\
$^2$Birkbeck, University of London\\
$^3$University of Trento\\
$^4$University of Liverpool \\
\emails
artale@inf.unibz.it,
roman@dcs.bbk.ac.uk,
andrea.mazzullo@unitn.it,
wolter@liverpool.ac.uk
}

\allowdisplaybreaks

\begin{document}

\maketitle

\begin{abstract}
Definite descriptions, such as
`the General Chair of KR 2024,\!'
are a semantically transparent device for object identification in knowledge representation. In first-order modal logic, definite descriptions have been widely investigated for their non-rigidity, which allows them to designate different objects (or  none at all) at different states. 
We propose expressive modal description logics with non-rigid definite descriptions and names, and investigate decidability and complexity of the satisfiability problem. We first systematically link satisfiability for the one-variable fragment of first-order modal logic with counting to our modal description logics. Then, we prove a promising \NExpTime{}-completeness result
for concept satisfiability for the fundamental epistemic multi-agent logic $\Sfive^{n}$ 
and its neighbours,
and show that some expressive logics that are undecidable with constant domain become decidable (but Ackermann-hard) with expanding domains. Finally, we conduct a fine-grained analysis of decidability of temporal logics.

\end{abstract}




\maketitle



%
%

\section{Introduction}
\label{sec:intro}

\emph{Definite descriptions},
like
`the General Chair of KR 2024,\!'
are expressions of the form `the unique $x$ such that $\p$.\!' Together with \emph{individual names}
such as
`Pierre,\!'
they are used as \emph{referring expressions} to identify objects in a given domain~\cite{ReiDal00,BorEtAl16,BorEtAl17}.

Definite descriptions and individual names can also \emph{fail to designate} any object at all, as in the case of the definite description
`the KR Conference held after KR 2018 and before KR 2020,\!'
or the individual name
`KR 2019.\!'
In order to admit these
as genuine terms of the language, while allowing for their possible lack of referents, formalisms based on \emph{free logic} semantics have been developed~\cite{Ben02,Leh02,Ind21,IndZaw21}.
In contrast, classical logic approaches assume that individual names always designate, and that definite descriptions can be paraphrased in terms of
existence and uniqueness conditions, an approach dating back to Russell~(\citeyear{Rus05}).
Recently,
definite descriptions have been introduced to enrich standard description logics (DLs) with nominals, $\ALCO$ and $\ELO$~\cite{NeuEtAl20,ArtEtAl20b,ArtEtAl21a}.

In particular, DLs $\ALCOud$ and $\ELOud$~\cite{ArtEtAl20b,ArtEtAl21a} include the universal role, $u$, as well as nominals and definite descriptions of the form $\{ \defdes C \}$ (`the unique object in $C$')
as basic concept constructs,
while also employing a free logic semantics that allows non-designating terms.
Hence, for instance, using nominal $\{ \mathsf{kr24} \}$ to designate KR 2024, we can refer to the General Chair of KR 2024 by means of the definite description $\{ \defdes \exists \mathsf{isGenChair}. \{ \mathsf{kr24} \} \}$.
Then, in $\ALCOud$, we can say that Pierre is the General Chair of KR 2024 with the following concept:
\begin{equation*}
\exists u. ( \{ \mathsf{pierre} \} \sqcap \{ \defdes \exists \mathsf{isGenChair}. \{ \mathsf{kr24} \} \} ).
\end{equation*}
Indeed, if this concept is satisfiable (in other words, has a non-empty extension), then there are objects $p$ and $k$ designated by $\mathsf{pierre}$ and $\mathsf{kr24}$, respectively,  and the pair $(p, k)$ belongs to the interpretation of role $\mathsf{isGenChair}$, which contains no other $(o,k)$, thus making $p$ the only object connected to $k$ by  $\mathsf{isGenChair}$.
The free DLs $\ALCOud$ and $\ELOud$ have, respectively, $\ExpTime$-complete ontology satisfiability and $\PTime$-complete entailment problems, thus matching the complexity of the classical
DL counterparts.

Names and descriptions display interesting behaviours also in \emph{modal} (epistemic, temporal) settings. 
These are indeed \emph{referentially opaque} contexts, where the \emph{intension} (i.e., the meaning) of a term might not coincide with its \emph{extension} (that is, its referent)~\cite{Fit04}.
Here, referring expressions can behave as \emph{non-rigid designators}, meaning that they can designate different individuals across different states (epistemic alternatives, instants of time, etc.).

For example, in an epistemic scenario, even if everybody is aware that Pierre is the General Chair of KR 2024, not everyone thereby knows that he is also the General Chair of the (only, so far, and excluding virtual location) KR Conference held in Southeast Asia, despite the fact that `the KR Conference held in Southeast Asia' and `KR 2024' refer (to this day) to the same object. Indeed, `the KR Conference held in Southeast Asia' can be conceived to designate another event by someone unaware of its actual reference to KR 2024.
Similarly, in a temporal setting,
`the General Chair of KR'
refers to
Pierre
in 2024,
but
will designate someone else over the years.
So, for instance, if we assume that Pierre works and will continue working in Europe, 
we can conclude that the General Chair of KR currently works in Europe, but we should not infer that it will always remain the case.

%
Due to this fundamental and challenging interplay between designation and modalities,
non-rigid descriptions and names
have been widely investigated in first-order modal and temporal logics
as individual concepts or flexible terms
capable of taking different values across states~\cite{Coc84,Gar01,BraGhi07,KroMer08,FitMen12,CorOrl13,Ind20,Orl21,IndZaw23}.
%


The aim of this contribution is to introduce modal DLs that have sufficient expressive power to represent the phenomena discussed above, explore their relationship to standard modal DLs without definite descriptions and non-designating names, and investigate decidability and complexity of reasoning in these formalisms.

In detail, we propose the language $\MLALCOud$, which is a modalised extension of the free DL $\ALCOud$.
%
For instance, under an epistemic reading of the modal operator $\Box$, we can express with an $\MLALCOud$ concept that \emph{of} Pierre it is \emph{known} that he is the General Chair of KR 2024:
\[
\exists u. ( \{ \mathsf{pierre} \} \sqcap \Box \{ \defdes \exists \mathsf{isGenChair}. \{ \mathsf{kr24} \} \} ),
\]
while \emph{of} Pierre it is \emph{not} known that he is the General Chair of the KR Conference held in Southeast Asia:
\begin{multline*}
	\exists u. ( \{ \mathsf{pierre} \} \sqcap \lnot \Box  \{ \defdes \exists \mathsf{isGenChair}. \\
								 \{ \defdes ( \mathsf{KRConf} \sqcap \exists
\mathsf{hasLoc}.\mathsf{SEAsiaLoc}) \} \} ).
\end{multline*}

In a temporal setting, where $\Box$
is read as `always' and its dual $\Diamond$ as `at some point in the future,\!' KR 2024 can be made a rigid designator, which refers to the same object at all time instants. This can be achieved, for example, by means of an ontology that holds globally, at all time instants, and consists of the following concept inclusion (CI):
\[
	\{ \mathsf{kr24} \} \sqsubseteq \Box  \{ \mathsf{kr24} \}.
\]
In contrast, we can use the nominal $\{ \mathsf{kr} \}$
to refer to the current edition of the KR Conference,
for instance 
stating that KR 2024 is the current KR Conference with the concept $\exists u. (\{\mathsf{kr24}\} \sqcap \{ \mathsf{kr} \})$.
Moreover, by reading $\Next$ as `next year',
we can exploit its lack of rigidity to express, e.g., that there will be \emph{other}
KR Conferences in the future, with the CIs:
	\[
	\top \sqsubseteq \Diamond  \exists u. \{\mathsf{kr}\}, \qquad \{\mathsf{kr}\} \sqsubseteq \lnot \Next \{ \mathsf{kr} \}.
	\]


Compared
to first-order modal logics
with non-rigid designators~\cite{StaTho68,FitMen12},
both the definition of the language and the scope distinctions
for
modal operators are simplified in modalised DLs,
as first-order variables are replaced by a class-based DL syntax  leaving quantification
and predicate abstraction
implicit.
We, however, show that 
our modal DLs can be translated to a natural fragment of first-order modal logic with definite descriptions and predicate abstraction.

In this work, we consider
interpretations with both constant domains (in which the first-order domain is fixed across all worlds) and expanding domains (in which it
can grow when moving 
to accessible worlds).
%
We
first
establish, for any class of Kripke frames and both constant and expanding domains, polytime  satisfiability-preserving reductions (with and without ontology)  to the language $\MLALCOu$ without definite descriptions. We will show that, in addition, we can assume that each nominal designates in every world, but importantly is still non-rigid.

We
then
study
the satisfiability problem for various fundamental modal logics with epistemic and temporal interpretations. While for first-order modal logics with
rigid designators and no counting the restriction to \emph{monodic} formulas (in which modal operators are  applied only to formulas with at most one
free variable) very often ensures decidability, this is no longer the case if non-rigid designators
and/or
some counting are admitted~\cite{GabEtAl03}. For our modal DLs, this implies that the standard recipe for designing decidable languages --- apply modal operators only to concepts --- does not always work anymore.
Here,
we explore in detail when this recipe still works, and when it does not. 

First, we
closely link the two main sources of bad computational behaviour, non-rigid designators and counting, enabling us to use the results and machinery introduced
for first-order modal logics with counting~\cite{DBLP:conf/aiml/HampsonK12,HamKur15,DBLP:conf/aiml/Hampson16}.

Next, we prove that, rather surprisingly, for some fundamental modal epistemic logic, non-rigid designators come for free: concept satisfiability for the modal logics of all Kripke frames with $n$ accessibility relation, $\K^{n}$, and of all Kripke frames with $n$ equivalence relations, $\Sfive^{n}$, is in \NExpTime{} and thus no harder than without names at all. This holds under both  constant and expanding domains, and the proof is by showing the exponential finite model property. This answers an open problem discussed in~\cite{DBLP:conf/aiml/Hampson16}. With ontologies, however, concept satisfiability becomes undecidable under constant domains. While for $\Sfive^{n}$ (because of symmetric accessibility relations) constant domains coincide with expanding domains, for $\K^{n}$ decidability under expanding domains remains open with ontologies. As a fundamental example of an expressive modal logic, we investigate the extension $\K^{\ast n}$ of $\K^{n}$ with a modal operator interpreted by the transitive closure of the union of the $n$ accessibility relations, which can be interpreted as common knowledge~\cite{DBLP:books/mit/FHMV1995} but also as a fragment of propositional dynamic logic, PDL~\cite{DBLP:journals/sigact/HarelKT01}. In this case concept satisfiability is undecidable under expanding and constant domains, but becomes decidable, though Ackermann-hard, for the corresponding logic, $\Kfn$, on finite acyclic models with expanding domains. This answers an open problem posed in~\cite{WolZak01}.
 Note that Ackermann-hardness means that  the time required to establish (un)satisfiability is not bounded by any primitive recursive (or computable) function.
We refer the reader to Table~\ref{table:complexity} for an overview of our results. Recall that with rigid designators all these problems are known to be in 
2\NExpTime~\cite{WolZak01,GabEtAl03}.

Finally, in the temporal setting, we show that undecidability is a widespread phenomenon: concept satisfiability under global ontology with constant domain is undecidable for all our $\ALCO$-based fragments; the same applies to concept satisfiability in the languages with the universal modality and the temporal $\Diamond$ operator. 
Reasoning becomes decidable only when considering concept satisfiability (under global ontology) with expanding domains over finite flows of time (though the problem is Ackermann-hard),  or in fragments with  the  $\Next$ operator only, 
for which we prove \ExpTime{}-membership of  concept satisfiability (without ontologies).


\begin{table}[t] 
\centering\tabcolsep=4pt%
\newcommand{\refr}[1]{~{\tiny[#1]}}%
	\begin{tabular}{lccc}\toprule
		\mbox{}\hfil modal  & \multirow{2}{*}{domain} & \multirow{2}{*}{concept sat.} &  concept sat.\\  
		\mbox{}\hfil logic $L$ & & & under global ont.\\
		\midrule
		\multirow{2}{*}{$\K^n$, $n \geq 1$} & const. & \textsc{NExp}-c.\refr{T~\ref{thm:modconcdec}} & undecidable\refr{T~\ref{thm:k:global:sat}}\\ & exp. &  \textsc{NExp}-c.\refr{T~\ref{thm:modconcdec}}   & ? \\[2pt]
		$\Sfive$ & ---& \multicolumn{2}{c}{\textsc{NExp}-complete\refr{T~\ref{thm:modconcdec}}} \\
		$\Sfive^{n}$, $n\geq 2$ & --- & \textsc{NExp}-c.\refr{T~\ref{thm:modconcdec}} & undecidable\refr{T~\ref{thm:k:global:sat}} \\[2pt]
		\multirow{2}{*}{$\K^{\ast n}$, $n\geq 1$} & const. & \multicolumn{2}{c}{$\Sigma^{1}_{1}$-complete\refr{L~\ref{lemma:ltl-to-kfn} + T~\ref{th:temp1:const}}} \\ 
		& exp. & \multicolumn{2}{c}{undecidable\refr{L~\ref{lemma:ltl-to-kfn} + T~\ref{th:temp1:exp}.1}}  \\[2pt]
		\multirow{2}{*}{$\Kfn$, $n\geq 1$} & const. & \multicolumn{2}{c}{undecidable\refr{L~\ref{lemma:ltl-to-kfn} + T~\ref{th:temp1:const}}} \\
		  &  exp. & \multicolumn{2}{c}{\begin{tabular}{c}decidable\refr{T~\ref{thm:moddickson}},\\[-2pt]Ackermann-hard\refr{L~\ref{lemma:ltl-to-kfn} + T~\ref{th:temp1:exp}.2}\end{tabular}}\\\bottomrule
	\end{tabular}
	\caption{Concept satisfiability (under global ontology) for $L_{\ALCOud}$}\label{table:complexity}
\end{table}

It is to be emphasised that the non-rigidity of symbol interpretation by itself is \emph{not} the cause for the satisfiability problem to become harder. For instance, 
rigid roles are known to often cause 
an increase in the hardness of
the satisfiability problem compared with the case of non-rigid roles only~\cite{GabEtAl03}. What makes non-rigid designators computationally much harder than rigid designators is their ability to count in an unbounded way across worlds.

Full proofs and additional material are available in the appendix.

\paragraph{Related Work}

Other than in non-modal DLs~\cite{NeuEtAl20,ArtEtAl20b,ArtEtAl21a}, and  in the already mentioned \emph{first-order} modal and temporal settings,
definite descriptions have been recently investigated in the context of \emph{propositional} hybrid logics with nominals and the $@$ operator~\cite{WalZaw23}. Here, the additional $\defdes$ operator allows one to refer to the (one and only) state of a model that satisfies a certain condition.

Non-rigid designators have received, to the best of our knowledge, little attention in modal DLs,
despite the extensive body of research both on temporal~\cite{WolZak98,ArtFra05,LutEtAl08,ArtEtAl14} and epistemic~\cite{DonEtAl98,ArtEtAl07a,CalEtAl08,ConLen20} extensions.
As a notable exception, \citeauthor{MehRud11}~(\citeyear{MehRud11}) investigate non-rigid individual names  in an epistemic DL context,  
where
abstract individual names are interpreted on an infinite common domain, but without definite descriptions.
\section{Preliminaries}
\label{subsec:ml-nonrig}

The
$\MLALCOud$
language
is a modalised extension of the free description logic (DL)
$\ALCOud$~\cite{ArtEtAl20b,ArtEtAl21a}.
%
Let \NC, \NR and \NI be countably infinite and pairwise disjoint sets
of \emph{concept names}, \emph{role names} and \emph{individual names}, respectively, and let $I = \{ 1, \ldots, n \}$ be a finite set of \emph{modalities}.
$\MLALCOud$ \emph{terms} and
\emph{concepts} are defined by the following grammar:
\begin{gather*}
  \tau ::= a \mid \defdes C,
  \\
  C ::= A \mid \{ \tau \} \mid \lnot
  C \mid (C \sqcap C)
\mid \exists r.C
\mid \exists u.C
\mid \Diamond_{i} C,
\end{gather*}
where $a \in \NI$, $A \in \NC$, $r \in \NR$,
 $u\notin \NR$ is the
\emph{universal role}, and $\Diamond_{i}$, with $i \in I$, is a \emph{diamond} operator.
A term of the form $\defdes C$ is called a \emph{definite
  description} and
  $C$ its \emph{body}; a concept $\{ \tau \}$ is called a \emph{\textup{(}term\textup{)}
  nominal}. All the usual syntactic abbreviations are assumed:
$\bot = A \sqcap \lnot A$, $\top = \lnot \bot$,
$C \sqcup D = \lnot (\lnot C \sqcap \lnot D)$,
$C \Rightarrow D = \lnot C \sqcup D$,
$C \Leftrightarrow D = (C \Rightarrow D) \sqcap (D \Rightarrow C)$,
$\forall s. C = \lnot \exists s. \lnot C$, for
$s \in \NR \cup \{ u \}$, and \emph{box} operator $\Box_{i} C = \lnot \Diamond_{i} \lnot C$.
%
A \emph{concept inclusion} (\emph{CI}) is of the form $C \sqsubseteq D$,
for  concepts $C, D$.  We  use $C \equiv D$ to abbreviate $C \sqsubseteq D$ and $D \sqsubseteq C$.
An \emph{ontology} $\Omc$ is a finite set of CIs.

Fragments
$\ML^n_{\ALCOu}$, $\ML^n_{\ALCOd}$, and $\ML^n_{\ALCO}$ of  the full language are defined by restricting the
available DL constructors: they do not contain, respectively, definite descriptions, the universal role, and both 
definite descriptions and the universal role.

Given a concept $C$, the
set of \emph{subconcepts} of $C$, denoted by $\sub{C}$, is defined as usual
(see Appendix~\ref{sec:introprel}):
we only note that 
$\sub{\{\defdes C\}}$ contains $C$ along with its own subconcepts.
%
The \emph{signature} of $C$, denoted by~$\Sigma_{C}$, is the set of concept, role and individual names in~$C$. The signature of a CI or an ontology is defined similarly.
The set of \emph{connectives} of an ontology $\Omc$ is the constructors from the following list that occur in $\Omc$: $\defdes$, $\{ \cdot \}$, $\lnot$, $\sqcap$, $\exists$ with roles in $\NR$, $\exists u$, and $\Diamond_{i}$ with $i\in I$.
The \emph{modal depth} of terms and concepts is the maximum number of nested modal operators: $\md(A) = 0$, $\md(\defdes C) = \md(C)$ and $md(\Diamond_{i} C)  = \md(C) + 1$, for example.
%
The \emph{modal depth} of a CI or an ontology is the maximum modal depth of their concepts.

A \emph{frame} is a pair $\Fmf = (W, \{ R_{i} \}_{i \in I})$, where $W$ is a non-empty set of \emph{worlds} (or \emph{states}) and each
\mbox{$R_{i} \subseteq W \times W$}, for $i\in I$, is a binary \emph{accessibility relation} on $W$.
%
A \emph{partial interpretation with expanding domains} based on
a frame $\Fmf = (W, \{ R_{i} \}_{i \in I} )$ is a triple $\Mmf = ( \Fmf, \Delta, \Int)$, where
$\Delta$ is a function associating with every $w \in W$ a non-empty set, $\Delta^{w}$, called the \emph{domain of $w$ in~$\Mmf$}, such that $\Delta^{w} \subseteq \Delta^{v}$, whenever $w R_{i} v$,
for some $i \in I$;
and $\Int$ is a function associating with every $w \in W$ a
  \emph{partial} DL interpretation
  $\Imc_{w} = (\Delta^{w}, \cdot^{\Imc_{w}})$ that maps every
  $A \in \NC$ to a subset $A^{\Imc_{w}}$ of~$\Delta^{w}$, every
  $r\in\NR$ to a subset $r^{\Imc_{w}}$ of
  $\Delta^{w} \times \Delta^{w}$, the universal role $u$ to the set
  $\Delta^{w} \times \Delta^{w}$, and every $a$ in \emph{some
  subset} of $\NI$ to an element $a^{{\Imc_{w}}}$ in
  $\Delta^{w}$.
Hence, every
$\cdot^{\Imc_{w}}$ is a total function on $\NC \cup \NR$ but a
\emph{partial} function on~$\NI$. If $\Imc_{w}$ is defined on $a\in\NI$, then we say
that $a$ \emph{designates at~$w$}.
If every $a\in\NI$ designates at $w\in W$, then $\Imc_{w}$ is called \emph{total}.
We say that $\Mmf = ( \Fmf, \Delta, \Int)$ is a \emph{total} 
  interpretation if every $\Imc_{w}$, 
$w \in W$, is a \emph{total} interpretation.
In the sequel, we refer to partial interpretations as interpretations, and add the adjective `total' explicitly whenever this is the case.


An \emph{interpretation with constant domains} is defined as a special case
, where the function $\Delta$ is such that $\Delta^{w} = \Delta^{v}$, for every $w, v \in W$. With an abuse of notation, we denote the common domain by $\Delta$ and call it the \emph{domain of $\Mmf$}.

Given $\Mmf = (\Fmf, \Delta, \Int)$, with
$\Fmf = (W, \{ R_{i} \}_{i \in I})$, we define the \emph{value}
$\tau^{\Int_{w}}$ of a term $\tau$ in world $w\in W$ as $a^{\Int_{w}}$, for
$\tau = a$, and as follows, for $\tau = \defdes C$:
\begin{gather*}
		(\defdes C)^{\Int_{w}}  =
			\begin{cases}
				d, & \text{if} \ C^{\Int_{w}} = \{ d \}, \ \text{for some} \ d \in \Delta^{w}; \\
				\text{undefined}, & \text{otherwise}.
			\end{cases}
\end{gather*}
A term $\tau$ is said to
\emph{designate at $w$}
if $\tau^{\Int_{w}}$ is defined.
%
The
\emph{extension} $C^{\Int_{w}}$ of a concept $C$ in $w\in W$ is defined
as
follows, where $s \in \NR \cup \{ u \}$:
\begin{align*}
	(\neg C)^{\Int_{w}} & = \Delta^{w} \setminus C^{\Int_{w}}, \\ 
	(C \sqcap D)^{\Int_{w}} & = C^{\Int_{w}} \cap D^{\Int_{w}}, \\
	(\exists s.C)^{\Int_{w}} & = \{d \in \Delta^{w} \mid \exists  e \in C^{\Int_{w}}: (d,e) \in s^{\Int_{w}}\}, \\
  	(\Diamond_{i} C)^{\Imc_{w}} &= \bigl\{ d \in \Delta^{w}  \mid \exists 
  v \in W :
  w R_{i} v \text{ and } d \in C^{\Imc_{v}}
  \bigr\},
  \\
  \{ \tau \}^{\Imc_{w}} &=
			\begin{cases}
				\{ \tau^{\Imc_{w}} \}, & \text{if $\tau$ designates at $w$}, \\
				\, \emptyset, & \text{otherwise.}
			\end{cases}
\end{align*}
%
%
A concept $C$ is \emph{satisfied at $w\in W$ in $\Mmf$} if
$C^{\Imc_{w}} \neq \eset$; $C$ is \emph{satisfied in $\Mmf$} if it is satisfied at some $w\in W$ in $\Mmf$.  
A CI $C \sqsubseteq D$ is \emph{satisfied in~$\Mmf$}, written $\Mmf\models C\sqsubseteq D$, if 
$C^{\Imc_{w}} \subseteq D^{\Imc_{w}}$, for every $w\in W$. An ontology $\Omc$ is \emph{satisfied in $\Mmf$}, written $\Mmf\models \Omc$, if every CI in $\Omc$ is satisfied in $\Mmf$; we also say a concept $C$ is \emph{satisfied in $\Mmf$ under an ontology~$\Omc$}
if $\Mmf\models \Omc$ and $C^{\Imc_w} \neq \emptyset$, for some $w\in W$. 

An ontology $\Omc'$ is called a
\emph{model conservative extension} of an ontology $\Omc$ if every  interpretation that satisfies $\Omc'$ also satisfies $\Omc$, and every interpretation that satisfies $\Omc$ can be turned to satisfy $\Omc'$ by modifying the interpretation of symbols in $\sig{\Omc'} \setminus \sig{\Omc}$, while keeping fixed the interpretation of symbols in $\sig{\Omc}$.
Similarly, a concept $C'$ is said to be a \emph{model conservative extension} of a concept $C$ if every interpretation that satisfies $C'$ also satisfies $C$, and every interpretation satisfying $C$ can be turned into an interpretation that satisfies $C'$, by modifying the interpretation of symbols in $\sig{C'} \setminus \sig{C}$, while keeping fixed the interpretation of symbols in $\sig{C}$.

\begin{remark}[Encoding of assertions]\em
\label{rem:assertion}
Assertions can be introduced as syntactic sugar using the universal role, with $C(\tau)$ and  $r(\tau_1,\tau_2)$ abbreviations for, respectively, concepts
%
\begin{equation*}
   \exists u. (\{ \tau \} \sqcap  C) \text{ and }
   \exists u. (\{ \tau_{1} \} \sqcap \exists r.\{ \tau_{2} \}).
\end{equation*}
The first example in Sec.~\ref{sec:intro} is thus $\mathsf{isGenChair}(\mathsf{pierre}, \mathsf{kr24})$.

To avoid ambiguities, we need to use
square brackets
when applying negation and modal operators to assertions, as in
$\lnot [ C(\tau) ]$
and
$\Diamond_i [ C(\tau) ]$.
%
Observe that, in an assertion of the form $\Diamond_{i} C(\tau)$, the diamond acts as a \emph{de re} operator, since the concept $\Diamond_{i} C$ applies to the object, if any, designated by the term $\tau$ at the current world $w$, and the assertion is false at a world whenever $\tau$ fails to designate at $w$. On the other hand, in an expression of the form
$\Diamond_{i} [ C(\tau) ]$,
the diamond plays the role of a \emph{de dicto} modality, as it refers to the world of evaluation for the whole assertion $C(\tau)$. Using the lambda abstraction notation for first-order modal logic~\cite{FitMen12}, assertion $\Diamond_{i} A(a)$
corresponds to
$\exists x. (\lp \lambda y. x = y \rp (a) \land \Diamond_i A(x))$,
whereas
$\Diamond_i [A(a)]$
stands for
$\Diamond_{i} \exists x. (\lp \lambda y. x = y \rp (a) \land A(x))$; see Appendix~\ref{sec:introprel}
for details on the standard translation.

 \end{remark}
\section{Reasoning Problems and Reductions}
\label{sec:reduction}


Let $\Cmc$ be a class of frames (e.g., frames with $n$ equivalence relations) and $\ML^n_{\DL}$ a language.
We consider the following two main reasoning problems. 
\begin{description}
\item[Concept $\Cmc$-Satisfiability:] Given an  $\ML^n_{\DL}$-concept $C$, is there an interpretation $\Mmf$ based on a frame in $\Cmc$ such that
$C$ is satisfied in $\Mmf$?
\item[Concept $\Cmc$-Satisfiability under Global Ontology:] \hfill Given an  $\ML^n_{\DL}$-concept $C$ and an  $\ML^n_{\DL}$-ontology $\Omc$, is there an interpretation $\Mmf$ based on a frame in $\Cmc$
such that $C$ is satisfied  in $\Mmf$ under $\Omc$?
\end{description}
In the sequel, for the case of concept $\Cmc$-satisfiability under global ontology, we will assume without loss of generality that $C$ is a concept name.
Indeed, we can extend $\Omc$ with CI $A \equiv C$, for a fresh concept name $A$, and consider satisfiability of $A$ under the extended ontology, which is a model conservative extension of $\Omc$.

We begin
with a few observations on polytime
	reductions between the concept satisfiability problems (under global ontology) for two main languages, $\ML^n_{\ALCOu}$ and $\MLALCOud$, and various semantic conditions, including (non-)rigid designators, total and partial interpretations, and expanding and constant domains. We also show how to eliminate the universal role using the global ontology and how to replace definite descriptions with nominals, and the other way round, which in particular means that satisfiability in the full $\MLALCOud$ and its fragment without $\defdes$ has the same computational properties. These observations will be useful in our constructions below, where we also apply them to smaller fragments, if the results carry over. Proofs are available in Appendix~\ref{app:reductions}.

\paragraph{No-RDA Subsumes
RDA}

An interpretation $\Mmf$ satisfies the \emph{rigid designator assumption} \textup{(}\emph{RDA}\textup{)} if every individual name $a \in \NI$ is a \emph{rigid designator} in $\Mmf$, in the sense that,
for every $w, v \in W$ such that $w R_{i} v$, if $a$ designates at $w$, then it designates at~$v$ and $a^{\Int_{w}} = a^{\Int_{v}}$.
%
%
%
For instance, concept 
\begin{equation*}
\exists u.(\{a\} \sqcap \Box C) \sqcap \Diamond \exists u.(\{a\} \sqcap \lnot C),
\end{equation*}
is unsatisfiable in
interpretations 
\emph{with the RDA},
but is satisfiable otherwise, 
as $a$ can designate differently at
different worlds.
Note that
 an individual that fails to designate at all worlds is vacuously
 rigid.
The following proposition shows the
RDA can be enforced in interpretations
by an ontology.

\begin{restatable}{proposition}{rdatononrda}
\label{prop:rdatononrda}
In $\ML^n_{\ALCOu}$ and $\MLALCOud$,
concept $\Cmc$-satisfiability under global ontology with the RDA is poly\-time-reducible to
concept $\Cmc$-satisfiability under global ontology,
with both constant and expanding domains.
\end{restatable}

%
 The proof is based on the observation that
CIs of the form
$\{a\} \sqsubseteq \Box_i\{a\}$,
for $i\in I$, ensure that $a$ can be made  a rigid designator in any interpretation. This provides a reduction for the case of global ontology, where a given $\Omc$ is extended with such CIs for every $a$.
A similar reduction is provided in Appendix~B for concept satisfiability in total interpretations.

In the sequel,
we assume implicitly that interpretations do \emph{not} satisfy the RDA, and explicitly write `with the
RDA' where necessary.

\paragraph{From Total to Partial Satisfiability and Back}

Partial interpretations are a generalisation of the classical, total, interpretations, where all nominals designate at all possible worlds. It turns out that satisfiability in partial and total interpretations  are polytime-reducible to each other.
%

\begin{restatable}{proposition}{redtotaltopartial}
\label{lemma:redtotaltopartial}
In $\ML^n_{\ALCOu}$ and $\MLALCOud$,
total concept $\Cmc$-satisfiability \textup{(}under global ontology\textup{)} is polytime-re\-ducible to 
concept $\Cmc$-satisfiability \textup{(}under global ontology, respectively\textup{)},
 with both  constant and  expanding domains.
\end{restatable}

We sketch the proof for the case of global ontology. Let~$\Omc$ be an ontology. Consider the extension $\Omc'$ of $\Omc$ with CIs
$\top \sqsubseteq \exists u. \{ a \}$,
for all individual names $a$ in~$\Omc$. Clearly, these CIs  ensure that each $a$ designates in every accessible world.
The case of concept satisfiability is shown in Appendix~\ref{app:reductions}.
Next, we provide the converse reduction. 

\begin{restatable}{proposition}{redpartialtototal}
\label{lemma:redpartialtototal}
In $\ML^n_{\ALCOu}$ and $\MLALCOud$,
concept $\Cmc$-satisfiability \textup{(}under global ontology\textup{)} is polytime-reducible to total
concept $\Cmc$-satisfiability \textup{(}under global ontology, respectively\textup{)},
with both constant and expanding domains.
\end{restatable}
We again sketch the proof for the case of global ontology. Let $\Omc$ be an ontology. Consider $\Omc'$ obtained from $\Omc$ by replacing every nominal~$\{ a \}$ in~$\Omc$ with a fresh concept name
$N_{a}$ and by extending the result with all CIs
$N_a \sqsubseteq \{ a \}$. It follows that every $a$ in $\Omc'$ can designate in all worlds, but the corresponding concept $N_a$ may still be interpreted by the empty set in some worlds, thus reflecting the fact that $a$ in $\Omc$ could have failed to designate in those worlds.
The case of concept satisfiability is dealt with in Appendix~\ref{app:reductions}.




\paragraph{Normal Form for Ontologies and Concepts}
Next,
we define normal form that will help us prove further polytime reductions, e.g., Lemma~\ref{lemma:spy-point-reduction} and Proposition~\ref{lemma:redmludtomlu}, and then complexity upper bounds.
Let $\Omc$ be an ontology  and  $C$ a subconcept in $\Omc$. Denote by $\Omc[C/A]$ the result of replacing every occurrence of~$C$ in $\Omc$ with a fresh concept name $A$, called the \emph{surrogate} of $C$. Clearly, $\Omc[C/A]\cup \{ C \equiv A\}$ is a model conservative extension of~$\Omc$. We can systematically apply this procedure to obtain an ontology in \emph{normal form} where connectives are applied only to concept names: e.g., definite descriptions occur only in the form of $\defdes B$, for a concept name~$B$.  If surrogates are introduced for innermost connectives first, then the transformation runs in polytime.

\begin{restatable}{lemma}{ontologynormalform} 
\label{lemma:ontology:normal-form}
For any $\MLALCOud$ ontology $\Omc$, we can construct in polytime an $\MLALCOud$ ontology $\Omc'$ in normal form that is a model conservative extension of $\Omc$. Moreover, $\Omc'$ uses the same set of connectives as $\Omc$.
\end{restatable}

If the language contains the universal role, then a similar construction transforms concepts $D$ into normal form. For a single modality ($n=1$), we can use $\Box_1^{k} \forall u.(C \Leftrightarrow A)$, for all $k \leq \md(D)$.  If $n > 1$, then we need to carefully select sequences of boxes to avoid an exponential blowup. So, for an $\MLALCOud$ concept $D$ and its subconcept $C$, we define the set of \emph{$C$-relevant paths in $D$},
denoted by $\rpath(D, C)$,
as the sequences $(i_1,\dots, i_n)$ of the $\Diamond_{i_j}$ operators under which $C$ occurs in~$D$. For example, for $D = \Diamond_1\neg A \sqcap \Diamond_2\Diamond_3 A$, we have $\rpath(D, A) =  \{(1), (2, 3) \}$ and $\rpath(D, \neg A) = \{(1)\}$.
 %
Note that the maximum length of a path in $\rpath(D, C)$ is $\md(D)$.
We also define the `$\Box$-modality' for each path: 
for a concept $E$, we recursively define
\begin{equation*}
	\Box^{\epsilon} E  = E  \text{ and } \Box^{i \cdot \pi} E =  \Box_{i} \Box^{\pi} E, \text{ for any path } \pi.
\end{equation*}
As before,  the \emph{surrogate} of $C$ is  a fresh concept name $A$, and $D[C/A]$ denotes the result of replacing~$C$ with $A$ in $D$.


\begin{restatable}{lemma}{normalformconc}
\label{lem:normalformconc}
Let $D$ be an $\MLALCOud$ concept and $C$ its subconcept. 
Denote by $D'$  the conjunction of $D[C/A]$ and
\begin{equation}
\label{eq:normalformconj}
\Box^\pi \forall u.(C \Leftrightarrow A), \text{ for all } \pi \in \rpath(D, C).
\end{equation}
Then $D'$ is a model conservative extension of $D$. Moreover,
$\rpath(D', A) = \rpath(D, C)$ and $\rpath(D', E) = \rpath(D, E)$, for any subconcept $E$ of $C$.
\end{restatable}

For any concept $D$, by repeatedly replacing non-atomic subconcepts with their surrogates, one can obtain a concept $D^{\ast}$ in \emph{normal form}, which is a conjunction of a concept name and concepts of the form~\eqref{eq:normalformconj}. By Lemma~\ref{lem:normalformconc}, $D^{\ast}$ is a model conservative extension of~$D$. Moreover, if  surrogates are introduced  for innermost  non-atomic concepts first, the procedure runs in polytime in the size of $D$.

\paragraph{Spy Points: Eliminating the Universal Role}
Our next observation allows us to eliminate occurrences of the universal role from ontologies. 

\begin{restatable}{lemma}{spypointreduction}
\label{lemma:spy-point-reduction}
Let $\Omc$ be an $\MLALCOud$ ontology in normal form. Denote by $\Omc'$ the $\MLALCOd$ ontology obtained from $\Omc'$ by replacing
\begin{itemize}
\item each CI of the form $B \sqsubseteq \exists u.B'$ with $B \sqsubseteq \exists r.B'$, and
\item each CI of the form $\exists u.B \sqsubseteq B'$ with the following:
%
\begin{equation*}
\top  \sqsubseteq \exists r.\{e\},\ \  A \sqsubseteq \{e\}, \ \
\neg B'  \sqsubseteq \exists r.A, \ \
\exists r.A \sqsubseteq \neg B, 
\end{equation*}
\end{itemize}
where $r$, $e$ and $A$ are fresh role, nominal and concept names, respectively. %
Then $\Omc'$ is a model conservative extension of~$\Omc$, and the size of $\Omc'$ is linear in the size of $\Omc$.
\end{restatable}
Intuitively, positive occurrences of $\exists u.B'$ ensure that $B'$ is non-empty, which can also be achieved with a fresh role~$r$. For negative occurrences of $\exists u.B'$, we use a spy-point $e$, which is accessible, via a fresh $r$, from every domain element and belongs to $A$ whenever $\neg B'$ is non-empty (that is, whenever $B'$ does not coincide with the domain). If this is the case, then no domain element can be in~$B$, which, by contraposition, implies $\exists u.B\sqsubseteq B'$. Note that $e$ can be rigid.

\paragraph{From Nominals to Definite Descriptions and Back}


We first observe that nominals can be easily encoded with definite descriptions. Indeed, given an ontology $\Omc$, take a fresh concept name $N_{a}$ for each individual name $a$ in~$\Omc$, and let $\Omc'$ be the result of replacing every occurrence of $\{ a \}$ in~$\Omc$
with~$\{ \defdes N_{a} \}$. Clearly, $\Omc'$ is a model conservative extension of $\Omc$,
and vice versa.
Note that $\Omc'$ is in the fragment $\MLALCud$ without nominals, which 
has no distinction 
between partial and total interpretations,
and
between the RDA and no-RDA cases.
Thus, we have the following result.
\begin{restatable}{proposition}{nomnonrdatodefdes}
\label{prop:nomnonrdatodefdes}
In $\ML^n_{\ALCOu}$ and $\MLALCOud$,
concept $\Cmc$-satisfiability \textup{(}under global ontology\textup{)} is polytime-reducible to
$\MLALCud$ concept $\Cmc$-satisfiability \textup{(}under global ontology, respectively\textup{)},
with both constant and expanding domains.
\end{restatable}

Conversely,
we now show how to replace definite
descriptions $\iota C$ with nominals using the universal role.


\begin{restatable}{proposition}{redmludtomlu}
\label{lemma:redmludtomlu}
$\MLALCOud$ concept $\Cmc$-satisfiability \textup{(}under global ontology\textup{)} is polytime-reducible to
$\MLALCOu$
concept $\Cmc$-satisfiability \textup{(}under global ontology, respectively\textup{)},
with both constant and expanding domains.
\end{restatable}
The proof reduces the total $\Cmc$-satisfiability problems in $\MLALCOud$ to total $\Cmc$-satisfiability in $\MLALCOu$, which, by Propositions~\ref{lemma:redpartialtototal} and \ref{lemma:redtotaltopartial} gives us the required result. We sketch the case of the global ontology. By Lemma~\ref{lemma:ontology:normal-form}, we can assume that the given $\Omc$ is in normal form. Let $\Omc\red$ be the result of replacing each $A_{\{ \defdes  B \}} \equiv \{\defdes B\} $ in $\Omc$ with CIs
\begin{equation*}
A_{\{\defdes B\}} \sqsubseteq B\sqcap \{a_B \} \text{ and } B \sqcap \forall u.(B\Rightarrow \{a_B\}) \sqsubseteq A_{\{\defdes B\}}, 
\end{equation*}
where $a_B$ is a fresh individual name. Intuitively, the first CI ensures that the surrogate for $\{\defdes B\}$ belongs to $B$ and is never interpreted by more than one domain element. The second CI ensures that if $B$ is a singleton, then that element belongs to the surrogate for $\{\defdes B\}$. Formally, we show that $\Omc\red$ is a model conservative extension of $\Omc$.
The case of concept satisfiability relies on normal form of concepts (Lemma~\ref{lem:normalformconc}) and is treated in Appendix~\ref{app:reductions}.


\paragraph{Expanding to Constant Domains}

It is known that, for the satisfiability problems, the interpretations with expanding domains can be simulated by constant domain interpretations, where a fresh concept name representing the domain is used to relativise concepts and CIs; see e.g.,~\cite[Proposition 3.32~(ii), (iv)]{GabEtAl03}. We restate this standard result in our setting for completeness:
%
\begin{restatable}{proposition}{redexptoconst}
\label{prop:redexptoconst}
In $\ML^n_{\ALCOu}$
and $\MLALCOud$,
concept  $\Cmc$-satisfiability \textup{(}under global ontology\textup{)} with expanding domains is polytime-re\-ducible to concept $\Cmc$-satisfiability \textup{(}under global ontology, respectively\textup{)} with constant domain.
\end{restatable}

In the sequel,
we implicitly  adopt  the \emph{constant domain assumption}~\cite{GabEtAl03}, and explicitly write when we consider interpretations with expanding domains instead.




\section{Non-Rigid Designators and Counting}
\label{sec:counting}

\newcommand{\MLDiff}{\mathcal{ML}_{\text{\textup{Diff}}}^{n}}	
	
We prove a strong link between non-rigid designators and the first-order one-variable modal logic enriched with the `elsewhere' quantifier, $\MLDiff$, introduced and investigated in~\cite{DBLP:conf/aiml/HampsonK12,HamKur15,DBLP:conf/aiml/Hampson16}. We define $\MLDiff$ using DL-style syntax: concepts in $\MLDiff$ are of the form 
\begin{equation*}
  C ::= A \mid \lnot
  C \mid (C \sqcap C)
\mid \exists u.C
\mid \exists^{\ne} u.C
\mid \Diamond_i C,
\end{equation*}
where $i\in I$. Observe that the language has no terms and no roles apart from the universal role $u$. 
All constructs are interpreted as before,
with the addition of
%
\begin{align*}
  (\exists^{\ne} u.C)^{\Imc_{w}} &= \bigl\{ d \in \Delta^w \mid 
  C^{\Imc_{w}} \setminus \{ d \} \ne \emptyset
  \bigr\}.
\end{align*}
Note that our language contains existential quantification, which in~\cite{HamKur15} is introduced as an abbreviation for $C \sqcup \exists^{\ne} u.C$. In fact,
$\MLDiff$ can be regarded as a basic first-order modal logic with counting because the counting quantifier $\exists^{=1}u.C$ (`there is exactly one $C$') with
\begin{align*}
	(\exists^{=1} u.C)^{\Imc_{w}} &= \bigl\{ d \in \Delta^w \mid 
	|C^{\Imc_{w}}|=1\bigr\}
\end{align*}
is clearly logically equivalent to the  $\MLDiff$-concept \mbox{$\exists u.(C \sqcap \neg \exists^{\ne}u.C)$}. Conversely, $\exists^{\ne}u.C$ is logically equivalent to
$\exists u.C\sqcap (C \Rightarrow \neg \exists^{=1} u.C)$.
So, one could replace $\exists^{\ne}$ by $\exists^{=1}$ in the definition $\MLDiff$ and obtain a logic with exactly the same expressive power.

\begin{restatable}{theorem}{diff}
\label{th:diff}
%
%
\textup{(1)} $\Cmc$-satisfiability of $\MLALCOu$-concepts \textup{(}under global ontology\textup{)}
	can be reduced in double exponential time to $\Cmc$-satisfiability of $\MLDiff$-concepts \textup{(}under global ontology, respectively\textup{)}, with both constant and expanding domains. 
	
\textup{(2)} Conversely, $\Cmc$-satisfiability of $\MLDiff$-concepts \textup{(}under global ontology\textup{)} is polytime-reducible to $\Cmc$-satisfiability of $\MLALCOu$-concepts \textup{(}under global ontology, respectively\textup{)},  with both constant and expanding domains. 
\end{restatable}
We first give the main ingredients for the proof of Item~(1) with global ontology~$\Omc$. If $\Omc$ contains no roles apart from $u$, then we introduce a surrogate concept name $\{a\}^\sharp$ for each individual name $a$ in $\Omc$ and denote by $\Omc^\sharp$ the result of replacing each $\{a\}$ with $\{a\}^\sharp$ and then extending the ontology with CIs of the form $\top \sqsubseteq \exists^{=1} u.\{a\}^\sharp$. The resulting $\MLDiff$ ontology $\Omc^\sharp$ is clearly a model conservative extension of~$\Omc$, which completes the reduction. If, however, $\Omc$ contains role names, then we apply the quasimodel technique~\cite{GabEtAl03} to deal with binary relations (roles), whose interpretations in different worlds are independent from each other. This technique is also used in Sections~\ref{sec:reasoning} and~\ref{sec:reasontfdl}. In quasimodels, each DL interpretation is represented as a \emph{quasistate}, which is a non-empty set of \emph{types}, maximal consistent sets of subconcepts of $\Omc$ and their negations: each type~$\contp$ represents all domain elements satisfying the concepts in~$\contp$. 
In this proof, we can characterise  possible quasistates for $\Omc$  using concepts: the \emph{description} $\Xi_{\settp}$ of a quasistate $\settp$ is  
\begin{equation*}
\forall u.\bigsqcup_{\contp\in\settp} \contp \ \ \sqcap \ \ \bigsqcap_{\contp\in\settp}  \exists u.\contp,
\end{equation*}
where $\contp$ is the conjunction ($\sqcap$) of all concepts $C\in\contp$. The set $\Smc_\Omc$ of quasistates that can possibly occur in interpretations satisfying $\Omc$ can be obtained by checking whether the \emph{modal abstraction} $\Xi_{\settp}^\ast$ of its description $\Xi_{\settp}$  is satisfiable, where $\cdot^\ast$ is the result of replacing subconcepts starting with a modal operator by fresh concept names. Since satisfiability of $\ALCOu$-concepts under $\ALCOu$ ontologies is in \textsc{ExpTime}, the set~$\Smc_\Omc$ can be computed in double exponential time. In order to ensure that quasistates fit together to form a representation of an interpretation satisfying $\Omc$, we use the \emph{DL-abstraction} $\cdot^\sharp$ described above, except that now we replace not only nominals but also existential restrictions with fresh concept names. We construct $\MLDiff$-ontology $\Omc^\sharp$ by replacing CIs $C\sqsubseteq D$ in $\Omc$ with $C^\sharp\sqsubseteq D^\sharp$ and extending the result with CIs  $\top\sqsubseteq \bigsqcup_{\settp\in \Smc_\Omc} \Xi_{\settp}^\sharp$  and  $\top \sqsubseteq \exists^{=1} u.\{a\}^\sharp$, for all individual names $a$. Then we show that $A$ is $\Cmc$-satisfiable under $\Omc$ iff $\contp^\sharp$ is $\Cmc$-satisfiable under $\Omc^\sharp$, for some type $\contp$ containing~$A$. 
See Appendix~\ref{app:counting} for full details and the case of concept satisfiability (without ontology).


Item (2) is proved by lifting to the modal description logic setting the observation of \citeauthor{DBLP:journals/jphil/GargovG93}~(\citeyear{DBLP:journals/jphil/GargovG93}) that the difference modality and nominals are mutually interpretable by each other; see the equivalences that show the same expressive power of $\exists^{\ne}$ and $\exists^{= 1}$.

Using Theorem~\ref{th:diff}, one can transfer a large number of (un)decidability results and lower complexity bounds from first-order modal logics with `elsewhere' to modal DLs with non-rigid designators; see, e.g., Theorem~\ref{thm:k:global:sat} in Section~\ref{sec:reasoning}. Conversely, the results proved below entail new results for first-order modal logics with `elsewhere' and/or counting.

%
%

\section{Reasoning in Modal Free Description Logics}
\label{sec:reasoning}
%
%

The aim of this section is to show the results presented in Table~\ref{table:complexity} and also discuss a few related logics.
For some basic frame classes, it will be convenient to use standard modal logic notation when discussing the satisfiability problem.
So, given a propositional modal logic $L$ with $n$ operators and a DL fragment $\DL$, the problem of  
\emph{$L_{\DL}$ concept satisfiability} (\emph{under global ontology}) is the
 problem of deciding $\mathcal{C}_{L}$-satisfiability of $\mathcal{ML}^{n}_{\DL}$-concepts \textup{(}under global ontology, respectively\textup{)}, where $\mathcal{C}_{L}$ is the class of all frames validating $L$.  We focus on the modal logics $L$ characterised by the following classes $\Cmc_L$ of frames, with $n \geq 1$:
\begin{description}
\item[$\K^n$\textnormal{:}] $\Cmc_L$ is the class of all frames $(W,R_1,\dots,R_n)$;
\item[$\Sfive^n$\textnormal{:}] $\Cmc_L$ is the class of all frames $(W,R_1,\dots,R_n)$ such that the $R_i$ are equivalence relations;
\item[$\K^{\ast n}$\textnormal{:}] $\Cmc_L$ is the class of all frames $(W,R_{1},\ldots,R_{n},R)$ such that $R$ is the transitive closure of $R_{1}\cup\cdots \cup R_{n}$;
\item[$\Kfn$\textnormal{:}] $\Cmc_L$
is as for $\K^{\ast n}$ and, in addition, $W$ is finite and $R$ is irreflexive (in other words, there are no chains $w_{0}R_{i_{1}}w_{1}\cdots R_{i_{n}}w_{n}$ with $w_{0}=w_{n}$).
\end{description}

To illustrate the language $\Sfive^{2}_{\ALCOud}$, we express that Agent 2 \emph{knows that} Agent 1 does \emph{not know of} the General Chair of KR 2024 that they are the General Chair of the KR Conference held in Southeast Asia:
\begin{multline*}
\Box_{2} \exists u. (
\{ \defdes \exists \mathsf{isGenChair}. \{ \mathsf{kr24} \} \} \sqcap
\lnot \Box_{1} \{ \defdes \exists \mathsf{isGenChair}. \\ \{ \defdes ( \mathsf{KRConf} \sqcap \exists \mathsf{hasLoc}.\mathsf{SEAsiaLoc}) \} \}  ).
\end{multline*}%

The two main new results in Table~\ref{table:complexity} are the \NExpTime{} upper bound for $\K^{n}$ and $\Sfive^{n}$ and decidability of $\Kfn$. The remaining results are by (sometimes non-trivial) reductions to known results.
\begin{restatable}{theorem}{modconcdec}
	\label{thm:modconcdec}
	For $L\in \{ \K^{n},\Sfive^{n} \}$ with $n \geq 1$, $L_{\ALCOud}$
	concept
	satisfiability 
	is in $\NExpTime$ with both constant and expanding domains.
\end{restatable}
We provide a sketch of the main ideas of the proof for $\Sfive^{n}$.
First, observe  (using an unfolding argument) that any satisfiable concept $C$ is satisfied in world $w_{0}$ in a model based on a frame $\mathfrak{F}=(W,R_{1},\ldots,R_{n})$ such that the domain $W$ of $\mathfrak{F}$ is a prefix-closed set of words of the form 
\begin{equation*}
	\avec{w}=w_{0}i_{0}w_{1}\cdots i_{m-1}w_{m},
\end{equation*}
where $1\leq i_{j}\leq n$, $i_{j}\not=i_{j+1}$, and each $R_{i}$ is 
the smallest equivalence relation containing all pairs of the form $(\avec{w},\avec{w}iw)\in W\times W$. One can assume that $m$ is smaller than the modal depth of $C$. Next, one can substitute the first-order domain by quasistates (as introduced in Section~\ref{sec:counting}) and work with \emph{quasimodels} $\Qmf = (\mathfrak{F}, \funcand, \runs)$, in which $\funcand$ associates a quasistate with any world, and a set of \emph{runs} $\runs$ represents first-order domain elements as functions mapping every world $w$ to a type in $\funcand(w)$; note that runs were implicit in the proof of Theorem~\ref{th:diff} and could be defined as the domain elements in the interpretations for $\MLDiff$. If $\contp\in \funcand(w)$ contains a nominal, there is only one $r\in \runs$ with $r(w)=\contp$; this condition was expressed by the  CIs of the form~$\top\sqsubseteq \exists^{=1} u.\{a\}^\sharp$. Now, one can apply selective filtration and some surgery to such a quasimodel
in order
to obtain a quasimodel
with at most exponential outdegree (and so of at most exponential size), from which one can then extract
a model of exponential size. 
\begin{restatable}{theorem}{moddickson}
	\label{thm:moddickson}
$\Kfn_{\ALCOud}$ concept
	satisfiability under global ontology is decidable with expanding domains,
	 for $n\geq 1$.
\end{restatable}
The proof is again based on appropriate quasimodels, which are now based on expanding domain models using finite frames $\mathfrak{F}=(W,R_{1},\ldots,R_{n},R)$ such that the transitive closure $R$ of $R_{1}\cup \cdots \cup R_{n}$
contains no cycles. Unfolding shows that we may assume
that $(W,R)$ is a forest. We show decidability by proving a recursive bound on the size of these models. Let
$\extN = \mathbb{N}\cup \{\infty\}$, where we assume  $m \leq \infty$ and $m + \infty = \infty$, for all $m\in\extN$.
We fix an ordering $\contp_{1},\ldots,\contp_{k}$ of the types and represent a quasistate as a vector $(x_{1},\ldots,x_{k})\in \extN^{k}$, where $x_{i}$ is the number of domain elements that satisfy type $\contp_{i}$ in a world (equivalently, the number of runs through $\contp_{i}$). Let $|\avec{x}|=x_{1}+\cdots + x_{k}$
for $\avec{x}=(x_{1},\ldots, x_{k})\in\extN^k$. Observe that expanding domains correspond to the condition that $wR_{i}v$ implies $|\avec{x}| \leq |\avec{y}|$,
where  $\funcand(w), \funcand(v)$ are represented by $\avec{x}, \avec{y}\in\extN^k$, respectively. 
To obtain a recursive bound on the size of a finite model satisfying a concept we then apply Dickson's Lemma to the quasistates. Define the product ordering $\leq$ on $\extN^{k}$ by setting $\avec{x}\leq \avec{y}$ if $x_{i}\leq y_{i}$ for all $1\leq i \leq k$, where $\avec{x}=(x_{1},\ldots,x_{k})$ and $\avec{y}=(y_{1},\ldots,y_{k})$. 
A pair $\avec{x},\avec{y}$ with $\avec{x} \leq \avec{y}$ is called an \emph{increasing pair}. Dickson's Lemma states every infinite sequence $\avec{x}_{1},\avec{x}_{2},\ldots\in \extN^{k}$ contains an
increasing pair $\avec{x}_{i_{1}},\avec{x}_{i_{2}}$ with $i_{1}<i_{2}$.
In fact, assuming $|\avec{x}_{i}|\leq |\avec{x}_{i+1}|$ for all $i\geq 0$ and
given recursive bounds on $|\avec{x}_{1}|$ and
$|\avec{x}_{i+1}|-|\avec{x}_{i}|$, one can compute a recursive
bound on the length of the longest sequence without any increasing pair~\cite{DBLP:conf/lics/FigueiraFSS11}. Now, the proof of a recursive bound
on the size of a satisfying model consists in manipulating the quasimodel so that the outdegree of the forest is recursively bounded and Dickson's Lemma becomes applicable. The expanding domain assumption is crucial for this.  
 
We comment on the remaining results in Table~\ref{table:complexity}.
The $\NExpTime$-hardness results already hold without nominals~\cite[Theorem 14.14]{GabEtAl03} (the proof goes through also with expanding domain). The lower bounds for $\Kfn$ and $\K^{\ast n}$
follow from the following lemma and the corresponding lower bounds proved in the next section for temporal DLs (Table~\ref{table:complexity:temporal}).   
\begin{lemma}\label{lemma:ltl-to-kfn}
Concept satisfiability for $\LTLf_{\ALCOu}$ and $\LTL_{\ALCOu}$ are polytime-reducible to concept satisfiability for $\Kfn_{\ALCOu}$ and $\K^{\ast n}_{\ALCOu}$,
respectively, with and without ontology and with both constant and expanding domains.
\end{lemma}
The proof of this reduction from logics of linear frames with transitive closure to logics of branching frames with transitive closure is not trivial but can be done by adapting the reduction given in the proof of Theorem 6.24 in \cite{GabEtAl03} for product modal logics. 

Finally,  undecidability of concept satisfiability under global ontology for $\K^{n}$
with $n \geq 1$ follows from undecidability of  $\MLDiff$~\cite{DBLP:conf/aiml/HampsonK12} using the reduction of Theorem~\ref{th:diff}. The result for  $\Sfive^{n}$ with $n \geq 2$ can be obtained in a similar way, see, e.g.,~\cite{GabEtAl03}.

\begin{theorem}\label{thm:k:global:sat}
Concept satisfiability under global ontology is undecidable with constant domains for $\Sfive^n_{\ALCOu}$, with $n \geq 2$, and  for $\K^n_{\ALCOu}$ with $n \geq 1$.
\end{theorem}

For many important modal logics the decidability status
of modal DLs with non-rigid designators remains open. Most prominently, for the modal logics of (reflexive) transitive frames ${\bf K4}$ (and ${\bf S4}$, respectively), decidability of concept satisfiability with or without ontologies and with expanding or constant domains is open. As a ``finitary'' approximation and a first step to understand ${\bf K4}$ and ${\bf S4}$, we prove, as a specialisation of the proof of Theorem~\ref{thm:moddickson}, decidability of concept satisfiability for the G\"odel-L\"ob logic $\GL_{\ALCOud}$ ($\GL$ is the logic of all transitive and Noetherian\footnote{$(W,R)$ is \emph{Noetherian} if there is no infinite chain $w_{0}Rw_{1}R\cdots$ with $w_{i}\not= w_{i+1}$.} frames~\cite{boolos1995logic}) and Grzegorczyk's logic $\Grz_{\ALCOud}$ ($\Grz$
is the logic of all reflexive and transitive Noetherian frames~\cite{grzegorczyk1967some})
in expanding domain models with and without ontologies.
Alternatively, the decidability of concept satisfiability for
$\GL_{\ALCOud}$ and $\Grz_{\ALCOud}$ can be proved by a double-exponential-time reduction (similar to Theorem~\ref{th:diff}) to satisfiability in expanding domain
products of transitive Noetherian frames, which is known to be decidable~\cite{DBLP:journals/apal/GabelaiaKWZ06}.
%



\section{Reasoning in Temporal Free Description Logics}
\label{sec:reasontfdl}

In this section, we consider the temporal setting.
For the temporal DL language, we build $\TLALCOud$
\emph{terms}, \emph{concepts}, \emph{concept inclusions} and
\emph{ontologies} similarly to the $\MLALCOud$ case, with $n = 2$: the
language has two modalities --- temporal operators `sometime in
the future', $\Diamond$, and `at the next moment',~$\Next$.
In particular, the $\TLALCOud$ \emph{concepts} are defined by the grammar
\begin{gather*}
C ::= A
 \mid \{ \tau \}
 \mid \lnot C
 \mid (C \sqcap C)
\mid \exists r.C
\mid \exists u.C
\mid \Diamond C
\mid \Next C,
\end{gather*}
where $\tau$ is a $\TLALCOud$ \emph{term}, defined as in Section~\ref{subsec:ml-nonrig}.

A  \emph{flow of time} $\Fmf$ is a pair $(T, <)$, where $T$ is either the set  $\Nbl$ of 
non-negative integers or a subset of $\Nbl$ of the form $[0, n]$, for $n\in\Nbl$, and $<$ is the strict linear order on $T$. Elements of $T$ are called \emph{instants} (rather than worlds). A flow of time $(T, <)$ naturally gives rise to a frame $(T, <, \textit{succ})$ for  $\TLALCOud$, where $\textit{succ}$ is the successor relation: $\textit{succ} = \{ (t, t+1) \mid t, t+1\in T\}$. So, we will often say that an  interpretation $\Mmf$ is based on a flow $\Fmf$ if its frame is induced by $\Fmf$. We
denote it (with an abuse of notation) by
$\Mmf = (\Delta, (\Imc_{t})_{t \in T})$. 
If $\Mmf$ is based
on  $(\Nbl, <)$, then we call it an \emph{$\LTLALCOud$  interpretation}; if it is
based on $([0, n], <)$, for some $n\in\Nbl$, it is
called an \emph{$\LTLfALCOud$  interpretation}.

Given
$\Mmf = (\Delta, (\Int_t)_{t\in T})$, the \emph{value} of a $\TLALCOud$ term~$\tau$ at $t\in T$ and
the \emph{extension} of a $\TLALCOud$ concept~$C$ at $t\in T$, are defined
as in the modal case for $n = 2$: in particular, we have
\begin{equation*}
  (\Next D)^{\Imc_{t}}  \!=  \begin{cases} D^{\Imc_{t+1}},\hspace*{-0.7em} & \text{if } t + 1 \in T,\\\emptyset, & \text{otherwise},\end{cases} \text{ and }
  (\Diamond D)^{\Imc_t} \! = \hspace*{-1.7em}\bigcup_{t' \in T \text{ with }  t < t'} \hspace*{-1.9em} D^{\Imc_{t'}}.			
\end{equation*}
Note that $\Diamond$ is interpreted by $<$ and so
does not include the current instant.

We will consider restrictions of the base language $\TLALCOud$ along
both the DL and temporal dimensions. First, $\TLALCOd$, $\TLALCOu$ and
$\TLALCO$ stand for the fragments of $\TLALCOud$ without the universal
role, definite descriptions and both constructs, respectively.  In
addition to the basic free description logic $\ALCOud$, we define temporal extensions of the light-weight free DL $\ELOud$, which does not contain negation (and so disjunction). 
More precisely, the language $\mathcal{TL}_{\ELOud}$ is obtained from
$\TLALCOud$ by allowing only $\top$ (considered as a primitive
logical symbol), concept names, term nominals, conjunctions
and existential restrictions in the construction of concepts.
%
Then, by removing the universal role or/and
definite descriptions, we define $\TL_{\ELOd}$, $\TL_{\ELOu}$ and
$\TL_{\ELO}$  in the obvious way. 

In the temporal dimension, given a DL $\DL$, the%
\emph{$\Diamond\Box$-fragment},
$\TL^{\Diamond}_{\DL}$,
and
the
\emph{$\Next$-fragment},
$\TL^{\textstyle\circ}_{\DL}$,
are obtained from $\TL_{\DL}$
by disallowing the $\Next$ and the $\Diamond/\Box$ operators, respectively. Both
fragments
correspond to the unimodal language $\ML^1_{\DL}$, but with different accessibility relations.

In the following we will combine the syntactic restrictions
(fragments) with the semantic restrictions on interpretations and refer, for example, to
the satisfiability problem for $\mathcal{TL}^\Diamond_{\ALCOu}$
concepts in $\LTLf_{\ALCOud}$ interpretations simply as
$\LTLfdALCOu$ concept satisfiability.

As an example,
in $\LTLf_{\ALCOud}$, we express that whoever is a Program Chair of KR will not be Program Chair of KR again, but is always appointed as either the General Chair or a PC member of next year's KR, by means of the CI:
	\begin{multline*}
\exists {\sf isProgChair}.\{ {\sf kr} \} \sqsubseteq \lnot \Diamond \exists {\sf isProgChair}.\{ {\sf kr} \} \ \sqcap \\
							 ( \{ \defdes \exists {\sf isGenChair}. \Next \{ {\sf kr} \} \} \sqcup \exists {\sf isPCMember}. \Next \{ {\sf kr } \} ).
	\end{multline*}%


%


\begin{table*}[t]%
\centering%
\newcommand{\Exp}{\textsc{Exp}}%
\newcommand{\refr}[1]{~\tiny[#1]}%
\tabcolsep=5.5pt%
\begin{tabular}{lcccc}\toprule
	       \multicolumn{1}{c}{\multirow{2}{*}{temporal logic}} & \multicolumn{2}{c}{concept satisfiability} &  \multicolumn{2}{c}{concept sat. under global ontology}\\
		 & const.\ domain & exp.\ domains & const.\ domain & exp.\ domains  \\
		\midrule
$\LTLd_{\ALCOu}$, $\LTL_{\ALCOu}$ & $\Sigma_1^1$-complete\refr{T~\ref{th:temp1:const}} & undecidable\refr{T~\ref{th:temp1:exp}.1} & $\Sigma_1^1$-complete\refr{T~\ref{th:temp1:const}}   & undecidable\refr{T~\ref{th:temp1:exp}.1}  \\
$\LTLfd_{\ALCOu}$, $\LTLf_{\ALCOu}$ & undecidable\refr{T~\ref{th:temp1:const}} & \tabcolsep=0pt\begin{tabular}{c}decidable,\\[-2pt]Ackermann-hard\end{tabular}\refr{T~\ref{th:temp1:exp}.2} & undecidable\refr{T~\ref{th:temp1:const}}  & \tabcolsep=0pt\begin{tabular}{c}decidable,\\[-2pt]Ackermann-hard\end{tabular}\refr{T~\ref{th:temp1:exp}.2} \\
$\LTLd_{\ALCO}$, $\LTL_{\ALCO}$ & ? & ? & $\Sigma_1^1$-complete\refr{T~\ref{th:temp1:const}}  & undecidable\refr{T~\ref{th:temp1:exp}.1} \\
$\LTLfd_{\ALCO}$, $\LTLf_{\ALCO}$ & ? & ? & undecidable\refr{T~\ref{th:temp1:const}}  & \tabcolsep=0pt\begin{tabular}{c}decidable,\\[-2pt]Ackermann-hard\end{tabular}\refr{T~\ref{th:temp1:exp}.2} \\
$\LTLo_{\ALCOu}$ / $\LTLo_{\ALCO}$ & \Exp-c. / in \Exp{}\refr{T~\ref{thm:temp3}} & \Exp-c. / in \Exp{}\refr{T~\ref{thm:temp3}} & undecidable\refr{T~\ref{thm:temp4}} & ? \\[2pt]
$\LTLfo_{\ALCOu}$ / $\LTLfo_{\ALCO}$ & \Exp-c. / in \Exp{}\refr{T~\ref{thm:temp3}} & \Exp-c. / in \Exp{}\refr{T~\ref{thm:temp3}} & undecidable\refr{T~\ref{thm:temp4}} & decidable \\\bottomrule
\end{tabular}
	\caption{Concept satisfiability (under global ontology) for temporal DLs}\label{table:complexity:temporal}
\end{table*}



We
begin the study of the satisfiability problems for temporal DLs based on $\ALCO$ by showing that concept satisfiability in constant domains 
is undecidable or even $\Sigma_1^1$-complete (highly undecidable
in the analytical hierarchy) over the infinite flow of time $(\Nbl,<)$. This is very different from the classical case with RDA,  which was
  shown to be decidable in the absence of definite
  descriptions~\cite[Theorem 14.12]{GabEtAl03}. 

\begin{restatable}{theorem}{tempone}
\label{th:temp1:const}
With constant domains, concept satisfiability is  $\Sigma_1^1$-complete for $\LTLd_{\ALCOu}$   and $\LTL_{\ALCOu}$, and undecidable for
  $\LTLfd_{\ALCOu}$; also, concept satisfiability under global ontology is
  $\Sigma_1^1$-complete for $\LTLd_{\ALCO}$ and $\LTL_{\ALCO}$, and undecidable for
  $\LTLfd_{\ALCO}$.
\end{restatable}

The lower bounds follow,
using Theorem~\ref{th:diff}~(2), from the respective undecidability results for  the first-order one-variable temporal logic with counting~\cite{HamKur15}.
Observe that,
for $\LTLd_{\ALCO}$, we apply the spy-point universal role elimination of Lemma~\ref{lemma:spy-point-reduction}. Moreover, Proposition~\ref{lemma:redpartialtototal} also applies to $\LTLd_{\ALCO}$ under global ontology, and so the results hold in both total and partial interpretations.

They could also be proven more directly by encoding the $\Sigma_1^1$-complete recurrence and undecidable reachability problems for the Minsky counter machines~\cite{Min61,DBLP:journals/jacm/AlurH94}.
Intuitively, the value of a
counter of the Minsky machine can be represented as the cardinality of a certain concept. Then non-rigid nominals (or indeed the counting quantifier) can be used to ensure that the value of the counter is incremented/decremented (depending on
the command) by the transition: for instance, a
CI of the form
\begin{equation*}
Q_i \sqsubseteq \Next R_k \Leftrightarrow (R_k \sqcup \{ a_k \})
\end{equation*}
could be used to say that from state $Q_i$, the value of counter $k$ is increased by one. Note, however, that this CI uses the $\Next$ operator. Without it, the proof is considerably more elaborate and represents a
counter as a pair of concepts: $R_k$ is used to increment the counter, while $S_k$ to decrement it, so that the counter value is the cardinality of $R_k\sqcap \neg S_k$. Both concepts are made `monotone': $R_k \sqsubseteq \Box R_k$ and $S_k \sqsubseteq \Box S_k$, and for each transition of the Minsky machine, the non-rigid nominals pick an element that, for example, has never been in $R_k$ before but will remain in $R_k$ from the next instant on: $\neg R_k \sqcap \Box R_k$.  A sequence of these elements allows us to linearly order the domain and  construct a `diagonal' in the two-dimensional interpretation necessary for the encoding of the computation using only the $\Diamond$ operator.  

Reasoning in expanding domains turns out to be less complex, and we obtain the following:  
\begin{restatable}{theorem}{temponeexp}
\label{th:temp1:exp}
\textup{(1)}  With expanding domains, concept satisfiability is
  undecidable for $\LTLd_{\ALCOu}$, and concept satisfiability under global ontology is
  undecidable for $\LTLd_{\ALCO}$.

\textup{(2)} With expanding domains, concept satisfiability  \textup{(}under global ontology\textup{)}
  is decidable  for
  $\LTLf_{\ALCOud}$.
However, both problems are Ackermann-hard for
  $\LTLfd_{\ALCOu}$\textup{;} 
    moreover, concept satisfiability under global ontology is Ackermann-hard for  $\LTLfd_{\ALCO}$.
\end{restatable}
Undecidability and Ackermann-hardness are proven similarly to Theorem~\ref{th:temp1:const}. In this case, however, the master problems are, respectively, the $\omega$-reachability and reachability problems for lossy Minsky machines~\cite{DBLP:conf/cade/KonevWZ05,DBLP:conf/mfcs/Schnoebelen10},
which in addition to normal transitions can also arbitrarily decrease the counter values. Such computations can be naturally encoded \emph{backwards} in interpretations with expanding domains: the arbitrary decreases of counter values correspond to the extension of the interpretation domain with fresh elements.   

The positive decidability results over the finite flows of time follows from Theorem~\ref{thm:moddickson} by Lemma~\ref{lemma:ltl-to-kfn} (together with the reduction in Proposition~\ref{lemma:redmludtomlu}).

\paragraph{Temporal Free DLs Based on \ELO{}.}

Next, we transfer the above results to the \ELO{} family. As $\TL_{\ELO}$ concepts do not contain negation and the empty concept ($\bot$), they are trivially satisfiable. Thus, our main reasoning problem is based on the notion of entailment (rather than satisfiability).
\begin{description}
%
\item[CI Entailment  (over Finite Flows):] Given a $\TL_{\DL}$-CI $C_1\sqsubseteq C_2$ and a $\TL_{\DL}$-ontology $\Omc$, is it the case that $C_1^{\Imc_t} \subseteq C_2^{\Imc_t}$, for every $t\in T$ in  every  interpretation $\Mmf$ satisfying $\Omc$ and based  on  $(\Nbl, <)$ (every finite flow, respectively)?
\end{description}

It turns out that disjunction can be modelled in the temporal extension of $\ELO$ with the help of the $\Diamond$ modality~\cite{AKLWZ07}: intuitively, any CI of the form $\top \sqsubseteq B_1 \sqcup B_2$ is replaced with $\top \sqsubseteq\exists q.(\Diamond X_1 \sqcap \Diamond X_2)$, which says that $X_1$ and $X_2$ occur in some order in the future (possibly on another domain element). It then remains to check the order of $X_1$ and $X_2$ and, if, say, $X_1$ precedes $X_2$, then $B_1$ is chosen, otherwise $B_2$ is chosen. So, this reduction shows that the entailment problem for the fragments of $\TL^{\Diamond}_{\ELO}$ essentially has the same complexity as the complement of the satisfiability problem for the corresponding $\TL^{\Diamond}_{\ALCO}$ fragment:

\begin{restatable}{theorem}{temptwo}
\label{thm:temp2}
\textup{(1)} CI
entailment with constant domains  is $\Pi_1^1$-complete for $\LTLd_{\ELO}$ and 
 undecidable for $\LTLfd_{\ELO}$.
  
\textup{(2)} CI entailment with expanding domains is undecidable for $\LTLd_{\ELO}$. 

\textup{(3)}  CI entailment with expanding domains  is decidable but Ackermann-hard for $\LTLfd_{\ELO}$.
\end{restatable}
%

\paragraph{Next-Only Temporal Free DLs.}

As we have seen above, the $\Diamond$-only fragments normally exhibit the same bad computational behaviour as the full logics with both $\Diamond$ and $\Next$. We now provide some results for the fragments that contain only~$\Next$. We begin with some positive results for the satisfiability problem (without global ontology):

\begin{restatable}{theorem}{tempthree}
\label{thm:temp3}
With constant and with expanding domains,  concept satisfiability is \ExpTime-complete for $\LTLo_{\ALCOu}$ and $\LTLfo_{\ALCOu}$ and in \ExpTime{} for  $\LTLo_{\ALCO}$ and  $\LTLfo_{\ALCO}$. 
\end{restatable}

The  \ExpTime{} upper complexity bound can be shown by a type elimination procedure, similarly to the case of the product $\mathbf{Alt} \times \K_n$ of modal logics $\mathbf{Alt}$,  whose accessibility relation is a partial function, and multi-modal $\K_n$, which is a notational variant of $\ALC$; see~\cite[Theorem 6.6]{GabEtAl03}. One has to, in addition, take care of nominals and the universal role, but that can be done in exponential time. The matching lower bound is inherited from $\ALCOu$, but for the fragment without the universal role the exact complexity remains an open problem.

Our final result indicates that with the global ontology, the $\Next$-fragments behaves nearly as badly as the full language:

\begin{restatable}{theorem}{tempfour}
\label{thm:temp4}
With constant domains,  concept satisfiability under global ontology is undecidable for $\LTLo_{\ALCO}$ and $\LTLfo_{\ALCO}$.  
\end{restatable}

The proof is by a direct reduction of the reachability problem for Minsky machines, similarly to the simplified sketch for Theorem~\ref{th:temp1:const}; note the proof makes use of the spy-point universal role  elimination in Lemma~\ref{lemma:spy-point-reduction}.


\section{Discussion and Future Work}
\label{sec:conc}

We have introduced and investigated novel fragments of first-order modal logic with non-rigid (and possibly non-referring) individual names and definite descriptions. Potential applications that remain to be explored include business process management, where formalisms for representing the dynamic behaviour of data and information are crucial~\cite{DelEtAl23,DeuEtAl18}, and context, knowledge or standpoint-dependent reasoning for which possible worlds semantics is needed~\cite{GhiSer17,AlvEtAl23}.

Besides the open decidability problems discussed above, future research directions include the extension of our results to more expressive monodic fragments~\cite{GabEtAl03,HodEtAl02},
automated support for the construction of definite descriptions and referring expressions~\cite{ArtEtAl21a,KurEtAl23}, the design of `practical' reasoning algorithms for the languages considered here, and the extension of our results to  modal DLs with hybrid~\cite{Bra14,IndZaw23}, branching-time~\cite{HodEtAl02,GutEtAl12}, dynamic~\cite{Har79}, or non-normal operators~\cite{DalEtAl23}.

\section*{Acknowledgements}
Andrea Mazzullo acknowledges the support of the MUR PNRR project FAIR - Future AI Research (PE00000013) funded by the NextGenerationEU.
\bibliographystyle{kr}
\bibliography{bibliography}

\providecommand{\noopsort}[1]{}
\begin{thebibliography}{}

\bibitem[\protect\citeauthoryear{Alur and
  Henzinger}{1994}]{DBLP:journals/jacm/AlurH94}
Alur, R., and Henzinger, T.~A.
\newblock 1994.
\newblock A really temporal logic.
\newblock {\em J. {ACM}} 41(1):181--204.

\bibitem[\protect\citeauthoryear{Artale and Franconi}{2005}]{ArtFra05}
Artale, A., and Franconi, E.
\newblock 2005.
\newblock Temporal description logics.
\newblock In {\em Handbook of Temporal Reasoning in Artificial Intelligence},
  volume~1 of {\em Foundations of Artificial Intelligence}. Elsevier.
\newblock  375--388.

\bibitem[\protect\citeauthoryear{Artale \bgroup et al\mbox.\egroup
  }{2007}]{AKLWZ07}
Artale, A.; Kontchakov, R.; Lutz, C.; Wolter, F.; and Zakharyaschev, M.
\newblock 2007.
\newblock Temporalising tractable description logics.
\newblock In {\em Proc. of the 14th International Symposium on Temporal
  Representation and Reasoning {(TIME-07)}},  11--22.
\newblock IEEE Computer Society.
\newblock ISBN: 0-7695-2836-8.

\bibitem[\protect\citeauthoryear{Artale \bgroup et al\mbox.\egroup
  }{2014}]{ArtEtAl14}
Artale, A.; Kontchakov, R.; Ryzhikov, V.; and Zakharyaschev, M.
\newblock 2014.
\newblock A cookbook for temporal conceptual data modelling with description
  logics.
\newblock {\em {ACM} Trans. Comput. Log.} 15(3):25:1--25:50.

\bibitem[\protect\citeauthoryear{Artale \bgroup et al\mbox.\egroup
  }{2020}]{ArtEtAl20b}
Artale, A.; Mazzullo, A.; Ozaki, A.; and Wolter, F.
\newblock 2020.
\newblock On free description logics with definite descriptions.
\newblock In {\em {Proceedings of the 33rd International Workshop on
  Description Logics (DL-20)}}, volume 2663 of {\em {CEUR} Workshop
  Proceedings}.
\newblock CEUR-WS.org.

\bibitem[\protect\citeauthoryear{Artale \bgroup et al\mbox.\egroup
  }{2021}]{ArtEtAl21a}
Artale, A.; Mazzullo, A.; Ozaki, A.; and Wolter, F.
\newblock 2021.
\newblock On free description logics with definite descriptions.
\newblock In {\em {Proceedings of the 18th International Conference on
  Principles of Knowledge Representation and Reasoning (KR-21)}},  63--73.

\bibitem[\protect\citeauthoryear{Artale \bgroup et al\mbox.\egroup
  }{2024}]{ArtEtAl24}
Artale, A.; Kontchakov, R.; Mazzullo, A.; and Wolter, F.
\newblock 2024.
\newblock {Non-Rigid Designators in Modal and Temporal Free Description Logics
  (Extended Version)}.
\newblock {\em CoRR} abs/2405.07656.

\bibitem[\protect\citeauthoryear{Artale, Lutz, and Toman}{2007}]{ArtEtAl07a}
Artale, A.; Lutz, C.; and Toman, D.
\newblock 2007.
\newblock A description logic of change.
\newblock In {\em {IJCAI}},  218--223.

\bibitem[\protect\citeauthoryear{Bencivenga}{2002}]{Ben02}
Bencivenga, E.
\newblock 2002.
\newblock Free logics.
\newblock In {\em {Handbook of Philosophical Logic}}. Springer.
\newblock  147--196.

\bibitem[\protect\citeauthoryear{Boolos}{1995}]{boolos1995logic}
Boolos, G.
\newblock 1995.
\newblock {\em The logic of provability}.
\newblock Cambridge university press.

\bibitem[\protect\citeauthoryear{Borgida, Toman, and Weddell}{2016}]{BorEtAl16}
Borgida, A.; Toman, D.; and Weddell, G.~E.
\newblock 2016.
\newblock On referring expressions in query answering over first order
  knowledge bases.
\newblock In {\em Proceedings of the 15th International Conference on
  Principles of Knowledge Representation and Reasoning ({KR}-16)},  319--328.
\newblock {AAAI} Press.

\bibitem[\protect\citeauthoryear{Borgida, Toman, and Weddell}{2017}]{BorEtAl17}
Borgida, A.; Toman, D.; and Weddell, G.~E.
\newblock 2017.
\newblock Concerning referring expressions in query answers.
\newblock In {\em Proceedings of the 26th International Joint Conference on
  Artificial Intelligence, ({IJCAI}-17)},  4791--4795.
\newblock ijcai.org.

\bibitem[\protect\citeauthoryear{Bra{\"u}ner and Ghilardi}{2007}]{BraGhi07}
Bra{\"u}ner, T., and Ghilardi, S.
\newblock 2007.
\newblock {First-order Modal Logic}.
\newblock In {\em Handbook of Modal Logic}. Elsevier.
\newblock  549--620.

\bibitem[\protect\citeauthoryear{Bra{\"{u}}ner}{2014}]{Bra14}
Bra{\"{u}}ner, T.
\newblock 2014.
\newblock First-order hybrid logic: introduction and survey.
\newblock {\em Log. J. {IGPL}} 22(1):155--165.

\bibitem[\protect\citeauthoryear{Calvanese \bgroup et al\mbox.\egroup
  }{2008}]{CalEtAl08}
Calvanese, D.; De~Giacomo, G.; Lembo, D.; Lenzerini, M.; and Rosati, R.
\newblock 2008.
\newblock Inconsistency tolerance in {P2P} data integration: An epistemic logic
  approach.
\newblock {\em Inf. Syst.} 33(4-5):360--384.

\bibitem[\protect\citeauthoryear{Cocchiarella}{1984}]{Coc84}
Cocchiarella, N.~B.
\newblock 1984.
\newblock Philosophical perspectives on quantification in tense and modal
  logic.
\newblock II: Extensions of Classical Logic:309--353.

\bibitem[\protect\citeauthoryear{Console and Lenzerini}{2020}]{ConLen20}
Console, M., and Lenzerini, M.
\newblock 2020.
\newblock Epistemic integrity constraints for ontology-based data management.
\newblock In {\em Proceedings of the 34th {AAAI} Conference on Artificial
  Intelligence ({AAAI}-20)},  2790--2797.
\newblock {AAAI} Press.

\bibitem[\protect\citeauthoryear{Corsi and Orlandelli}{2013}]{CorOrl13}
Corsi, G., and Orlandelli, E.
\newblock 2013.
\newblock Free quantified epistemic logics.
\newblock {\em Studia Logica} 101(6):1159--1183.

\bibitem[\protect\citeauthoryear{Dalmonte \bgroup et al\mbox.\egroup
  }{2023}]{DalEtAl23}
Dalmonte, T.; Mazzullo, A.; Ozaki, A.; and Troquard, N.
\newblock 2023.
\newblock Non-normal modal description logics.
\newblock In Gaggl, S.~A.; Martinez, M.~V.; and Ortiz, M., eds., {\em Logics in
  Artificial Intelligence - 18th European Conference, {JELIA} 2023, Dresden,
  Germany, September 20-22, 2023, Proceedings}, volume 14281 of {\em Lecture
  Notes in Computer Science},  306--321.
\newblock Springer.

\bibitem[\protect\citeauthoryear{Delgrande \bgroup et al\mbox.\egroup
  }{2023}]{DelEtAl23}
Delgrande, J.~P.; Glimm, B.; Meyer, T.~A.; Truszczynski, M.; and Wolter, F.
\newblock 2023.
\newblock Current and future challenges in knowledge representation and
  reasoning.
\newblock {\em CoRR} abs/2308.04161.

\bibitem[\protect\citeauthoryear{Deutsch \bgroup et al\mbox.\egroup
  }{2018}]{DeuEtAl18}
Deutsch, A.; Hull, R.; Li, Y.; and Vianu, V.
\newblock 2018.
\newblock Automatic verification of database-centric systems.
\newblock {\em {ACM} {SIGLOG} News} 5(2):37--56.

\bibitem[\protect\citeauthoryear{Donini \bgroup et al\mbox.\egroup
  }{1998}]{DonEtAl98}
Donini, F.~M.; Lenzerini, M.; Nardi, D.; Nutt, W.; and Schaerf, A.
\newblock 1998.
\newblock An epistemic operator for description logics.
\newblock {\em Artif. Intell.} 100(1-2):225--274.

\bibitem[\protect\citeauthoryear{Fagin \bgroup et al\mbox.\egroup
  }{1995}]{DBLP:books/mit/FHMV1995}
Fagin, R.; Halpern, J.~Y.; Moses, Y.; and Vardi, M.~Y.
\newblock 1995.
\newblock {\em Reasoning About Knowledge}.
\newblock {MIT} Press.

\bibitem[\protect\citeauthoryear{Figueira \bgroup et al\mbox.\egroup
  }{2011}]{DBLP:conf/lics/FigueiraFSS11}
Figueira, D.; Figueira, S.; Schmitz, S.; and Schnoebelen, P.
\newblock 2011.
\newblock Ackermannian and primitive-recursive bounds with {Dickson's} lemma.
\newblock In {\em Proceedings of the 26th Annual {IEEE} Symposium on Logic in
  Computer Science, {LICS} 2011, June 21-24, 2011, Toronto, Ontario, Canada},
  269--278.
\newblock {IEEE} Computer Society.

\bibitem[\protect\citeauthoryear{Fitting and Mendelsohn}{2012}]{FitMen12}
Fitting, M., and Mendelsohn, R.~L.
\newblock 2012.
\newblock {\em First-order {M}odal {L}ogic}.
\newblock Springer Science \& Business Media.

\bibitem[\protect\citeauthoryear{Fitting}{2004}]{Fit04}
Fitting, M.
\newblock 2004.
\newblock First-order intensional logic.
\newblock {\em Ann. Pure Appl. Log.} 127(1-3):171--193.

\bibitem[\protect\citeauthoryear{Gabbay \bgroup et al\mbox.\egroup
  }{2003}]{GabEtAl03}
Gabbay, D.~M.; Kurucz, A.; Wolter, F.; and Zakharyaschev, M.
\newblock 2003.
\newblock {\em Many-dimensional Modal Logics: Theory and Applications}.
\newblock North Holland Publishing Company.

\bibitem[\protect\citeauthoryear{Gabelaia \bgroup et al\mbox.\egroup
  }{2006}]{DBLP:journals/apal/GabelaiaKWZ06}
Gabelaia, D.; Kurucz, A.; Wolter, F.; and Zakharyaschev, M.
\newblock 2006.
\newblock Non-primitive recursive decidability of products of modal logics with
  expanding domains.
\newblock {\em Ann. Pure Appl. Log.} 142(1-3):245--268.

\bibitem[\protect\citeauthoryear{Gargov and
  Goranko}{1993}]{DBLP:journals/jphil/GargovG93}
Gargov, G., and Goranko, V.
\newblock 1993.
\newblock Modal logic with names.
\newblock {\em J. Philos. Log.} 22(6):607--636.

\bibitem[\protect\citeauthoryear{Garson}{2001}]{Gar01}
Garson, J.~W.
\newblock 2001.
\newblock Quantification in modal logic.
\newblock In {\em Handbook of philosophical logic}, volume II: Extensions of
  Classical Logic. Springer.
\newblock  267--323.

\bibitem[\protect\citeauthoryear{Ghidini and Serafini}{2017}]{GhiSer17}
Ghidini, C., and Serafini, L.
\newblock 2017.
\newblock Distributed first order logic.
\newblock {\em Artif. Intell.} 253:1--39.

\bibitem[\protect\citeauthoryear{Grzegorczyk}{1967}]{grzegorczyk1967some}
Grzegorczyk, A.
\newblock 1967.
\newblock Some relational systems and the associated topological spaces.
\newblock {\em Fundamenta Mathematicae} 3(60):223--231.

\bibitem[\protect\citeauthoryear{Guti{\'{e}}rrez{-}Basulto, Jung, and
  Lutz}{2012}]{GutEtAl12}
Guti{\'{e}}rrez{-}Basulto, V.; Jung, J.~C.; and Lutz, C.
\newblock 2012.
\newblock Complexity of branching temporal description logics.
\newblock In {\em Proceedings of the 20th European Conference on Artificial
  Intelligence ({ECAI}-12)}, volume 242 of {\em Frontiers in Artificial
  Intelligence and Applications},  390--395.
\newblock {IOS} Press.

\bibitem[\protect\citeauthoryear{{Gómez Álvarez}, Rudolph, and
  Strass}{2023}]{AlvEtAl23}
{Gómez Álvarez}, L.; Rudolph, S.; and Strass, H.
\newblock 2023.
\newblock Tractable diversity: Scalable multiperspective ontology management
  via standpoint {EL}.
\newblock In {\em Proceedings of the Thirty-Second International Joint
  Conference on Artificial Intelligence, {IJCAI} 2023, 19th-25th August 2023,
  Macao, SAR, China},  3258--3267.
\newblock ijcai.org.

\bibitem[\protect\citeauthoryear{Hampson and
  Kurucz}{2012}]{DBLP:conf/aiml/HampsonK12}
Hampson, C., and Kurucz, A.
\newblock 2012.
\newblock On modal products with the logic of 'elsewhere'.
\newblock In {\em Advances in Modal Logic 9, papers from the ninth conference
  on "Advances in Modal Logic," held in Copenhagen, Denmark, 22-25 August
  2012},  339--347.

\bibitem[\protect\citeauthoryear{Hampson and Kurucz}{2015}]{HamKur15}
Hampson, C., and Kurucz, A.
\newblock 2015.
\newblock Undecidable propositional bimodal logics and one-variable first-order
  linear temporal logics with counting.
\newblock {\em {ACM} Trans. Comput. Log.} 16(3):27:1--27:36.

\bibitem[\protect\citeauthoryear{Hampson}{2016}]{DBLP:conf/aiml/Hampson16}
Hampson, C.
\newblock 2016.
\newblock Decidable first-order modal logics with counting quantifiers.
\newblock In {\em Advances in Modal Logic 11, proceedings of the 11th
  conference on "Advances in Modal Logic," held in Budapest, Hungary, August 30
  - September 2, 2016},  382--400.

\bibitem[\protect\citeauthoryear{Harel, Kozen, and
  Tiuryn}{2001}]{DBLP:journals/sigact/HarelKT01}
Harel, D.; Kozen, D.; and Tiuryn, J.
\newblock 2001.
\newblock Dynamic logic.
\newblock {\em {SIGACT} News} 32(1):66--69.

\bibitem[\protect\citeauthoryear{Harel}{1979}]{Har79}
Harel, D.
\newblock 1979.
\newblock {\em First-Order Dynamic Logic}, volume~68 of {\em Lecture Notes in
  Computer Science}.
\newblock Springer.

\bibitem[\protect\citeauthoryear{Hodkinson, Wolter, and
  Zakharyaschev}{2002}]{HodEtAl02}
Hodkinson, I.~M.; Wolter, F.; and Zakharyaschev, M.
\newblock 2002.
\newblock Decidable and undecidable fragments of first-order branching temporal
  logics.
\newblock In {\em Proceedings of the 17th {IEEE} Symposium on Logic in Computer
  Science {(LICS-02)}},  393--402.
\newblock {IEEE} Computer Society.

\bibitem[\protect\citeauthoryear{Indrzejczak and Zawidzki}{2021}]{IndZaw21}
Indrzejczak, A., and Zawidzki, M.
\newblock 2021.
\newblock Tableaux for free logics with descriptions.
\newblock In Das, A., and Negri, S., eds., {\em Proceedings of the 30th
  International Conference on Automated Reasoning with Analytic Tableaux and
  Related Methods ({TABLEAUX}-21)}, volume 12842 of {\em Lecture Notes in
  Computer Science},  56--73.
\newblock Springer.

\bibitem[\protect\citeauthoryear{Indrzejczak and Zawidzki}{2023}]{IndZaw23}
Indrzejczak, A., and Zawidzki, M.
\newblock 2023.
\newblock Definite descriptions and hybrid tense logic.
\newblock {\em Synthese} 202(3):98.

\bibitem[\protect\citeauthoryear{Indrzejczak}{2020}]{Ind20}
Indrzejczak, A.
\newblock 2020.
\newblock Existence, definedness and definite descriptions in hybrid modal
  logic.
\newblock In {\em Proceedings of the 13th Conference on Advances in Modal Logic
  ({AiML}-20)},  349--368.
\newblock College Publications.

\bibitem[\protect\citeauthoryear{Indrzejczak}{2021}]{Ind21}
Indrzejczak, A.
\newblock 2021.
\newblock Free logics are cut-free.
\newblock {\em Stud Logica} 109(4):859--886.

\bibitem[\protect\citeauthoryear{Konev, Wolter, and
  Zakharyaschev}{2005}]{DBLP:conf/cade/KonevWZ05}
Konev, B.; Wolter, F.; and Zakharyaschev, M.
\newblock 2005.
\newblock Temporal logics over transitive states.
\newblock In Nieuwenhuis, R., ed., {\em Automated Deduction - CADE-20, 20th
  International Conference on Automated Deduction, Tallinn, Estonia, July
  22-27, 2005, Proceedings}, volume 3632 of {\em Lecture Notes in Computer
  Science},  182--203.
\newblock Springer.

\bibitem[\protect\citeauthoryear{Kr{\"{o}}ger and Merz}{2008}]{KroMer08}
Kr{\"{o}}ger, F., and Merz, S.
\newblock 2008.
\newblock {\em Temporal Logic and State Systems}.
\newblock Texts in Theoretical Computer Science. An {EATCS} Series. Springer.

\bibitem[\protect\citeauthoryear{Kurucz, Wolter, and
  Zakharyaschev}{2023}]{KurEtAl23}
Kurucz, A.; Wolter, F.; and Zakharyaschev, M.
\newblock 2023.
\newblock Definitions and (uniform) interpolants in first-order modal logic.
\newblock {\em CoRR} abs/2303.04598.

\bibitem[\protect\citeauthoryear{Lehmann}{2002}]{Leh02}
Lehmann, S.
\newblock 2002.
\newblock More free logic.
\newblock In {\em {Handbook of Philosophical Logic}}. Springer.
\newblock  197--259.

\bibitem[\protect\citeauthoryear{Lutz, Wolter, and
  Zakharyaschev}{2008}]{LutEtAl08}
Lutz, C.; Wolter, F.; and Zakharyaschev, M.
\newblock 2008.
\newblock Temporal description logics: {A} survey.
\newblock In {\em Proceedings of the 15th International Symposium on Temporal
  Representation and Reasoning ({TIME}-08)},  3--14.
\newblock {IEEE} Computer Society.

\bibitem[\protect\citeauthoryear{Mehdi and Rudolph}{2011}]{MehRud11}
Mehdi, A., and Rudolph, S.
\newblock 2011.
\newblock Revisiting semantics for epistemic extensions of description logics.
\newblock In {\em Proceedings of the 25th {AAAI} Conference on Artificial
  Intelligence ({AAAI}-11)}.
\newblock {AAAI} Press.

\bibitem[\protect\citeauthoryear{Minsky}{1961}]{Min61}
Minsky, M.~L.
\newblock 1961.
\newblock Recursive unsolvability of {Post's} problem of "tag" and other topics
  in theory of {Turing} machines.
\newblock {\em Annals of Mathematics} 74(3):437--455.

\bibitem[\protect\citeauthoryear{Neuhaus, Kutz, and Righetti}{2020}]{NeuEtAl20}
Neuhaus, F.; Kutz, O.; and Righetti, G.
\newblock 2020.
\newblock Free description logic for ontologists.
\newblock In {\em Proceedings of the Joint Ontology Workshops ({JOWO}-20)},
  volume 2708 of {\em {CEUR} Workshop Proceedings}.
\newblock CEUR-WS.org.

\bibitem[\protect\citeauthoryear{Orlandelli}{2021}]{Orl21}
Orlandelli, E.
\newblock 2021.
\newblock Labelled calculi for quantified modal logics with definite
  descriptions.
\newblock {\em J. Log. Comput.} 31(3):923--946.

\bibitem[\protect\citeauthoryear{Reiter and Dale}{2000}]{ReiDal00}
Reiter, E., and Dale, R., eds.
\newblock 2000.
\newblock {\em {Building Natural Language Generation Systems}}.
\newblock Cambridge University Press.

\bibitem[\protect\citeauthoryear{Russell}{1905}]{Rus05}
Russell, B.
\newblock 1905.
\newblock {On Denoting}.
\newblock {\em Mind} 14(56):479--493.

\bibitem[\protect\citeauthoryear{Schnoebelen}{2010}]{DBLP:conf/mfcs/Schnoebelen10}
Schnoebelen, P.
\newblock 2010.
\newblock Revisiting {Ackermann}-hardness for lossy counter machines and reset
  {Petri} nets.
\newblock In {\em Mathematical Foundations of Computer Science 2010, 35th
  International Symposium, {MFCS} 2010, Brno, Czech Republic, August 23-27,
  2010. Proceedings}, volume 6281 of {\em Lecture Notes in Computer Science},
  616--628.
\newblock Springer.

\bibitem[\protect\citeauthoryear{Stalnaker and Thomason}{1968}]{StaTho68}
Stalnaker, R.~C., and Thomason, R.~H.
\newblock 1968.
\newblock Abstraction in first-order modal logic.
\newblock {\em Theoria} 34(3):203--207.

\bibitem[\protect\citeauthoryear{Walega and Zawidzki}{2023}]{WalZaw23}
Walega, P.~A., and Zawidzki, M.
\newblock 2023.
\newblock Hybrid modal operators for definite descriptions.
\newblock In Gaggl, S.~A.; Martinez, M.~V.; and Ortiz, M., eds., {\em Logics in
  Artificial Intelligence - 18th European Conference, {JELIA} 2023, Dresden,
  Germany, September 20-22, 2023, Proceedings}, volume 14281 of {\em Lecture
  Notes in Computer Science},  712--726.
\newblock Springer.

\bibitem[\protect\citeauthoryear{Wolter and Zakharyaschev}{1998}]{WolZak98}
Wolter, F., and Zakharyaschev, M.
\newblock 1998.
\newblock Temporalizing description logics.
\newblock In {\em {Proceedings of the 2nd International Symposium on Frontiers
  of Combining Systems (FroCoS-98)}},  104--109.
\newblock Research Studies Press/Wiley.

\bibitem[\protect\citeauthoryear{Wolter and Zakharyaschev}{2001}]{WolZak01}
Wolter, F., and Zakharyaschev, M.
\newblock 2001.
\newblock {Decidable Fragments of First-Order Modal Logics}.
\newblock {\em J. Symb. Log.} 66(3):1415--1438.

\end{thebibliography}

\newpage

\appendix


\section{Details on Introduction and Preliminaries}
\label{sec:introprel}

We define the \emph{reflexive diamond} operator as
$\Diamond_{i}^{+} C = C \sqcup \Diamond_{i} C$, and the \emph{reflexive box}
operator as $\Box_{i}^{+} C = \neg\Diamond_i^+\neg C$. 

The
set of \emph{subconcepts} of a concept $C$, denoted by $\sub{C}$, is defined 
inductively as follows:
\begin{align*}
\sub{A} & = \{ A\}, \\ 
\sub{\{a\}} & = \{ \{a\} \},\\
\sub{\{\defdes C\}} & = \{\{\defdes C\}\} \cup \sub{C}, \\ 
\sub{\neg C}  & = \{\neg C\}\cup \sub{C}, \\ 
\sub{C\sqcap D} & = \{ C\sqcap D\} \cup \sub{C} \cup \sub{D},\\
\sub{\exists s. C} & = \{ \exists s. C\} \cup\sub{C}, \ \text{with $s \in \NR \cup \{ u \}$},\\
\sub{\Diamond_i C} & = \{\Diamond_i C\} \cup \sub{C}.
\end{align*}

The
\emph{modal depth} of terms and concepts is defined by mutual induction:
\begin{align*}
\md(a) & = 0, \\ \md(\defdes C) & = \md(C), \\
\md(A) & = 0,  \\  \md(\{ \tau \}) & = \md(\tau), \\
\md(\lnot C) & = \md(C), \hspace*{-20em}\\
\md(C \sqcap D) & = \max\{ \md(C), \md(D) \}, \\
\md(\exists s.C) & = \md(C), \ \text{with $s \in \NR \cup \{ u \}$}, \\
\md(\Diamond_{i} C) & = \md(C) + 1. \hspace*{-20em}
\end{align*}
The \emph{modal depth} of a CI or an ontology is the maximum modal depth of concepts that occur in them.

\paragraph{Standard Translation to First-Order Modal Logic}
The following definitions are adjusted from~\cite{FitMen12}; see, e.g., Definitions~11.1.1-4.
%
The alphabet of the \emph{quantified modal language}, $\QMLdl$, consists of:
countably infinite and pairwise disjoint sets of
\emph{predicates} $\NPr$
(of fixed arities $\geq 0$,
with a distinguished \emph{equality} binary predicate, $=$),
\emph{individual names} $\NI$
and \emph{variables} $\Var$;
the \emph{Boolean operators} $\lnot, \land$;
the \emph{existential quantifier} $\exists$; 
the \emph{predicate abstraction operator} $\lambda$;
the \emph{definite description operator} $\defdes$;
and the \emph{modal operators} $\Diamond_{i}$ (\emph{diamond}), for each \emph{modality} $i \in I$.
%
\emph{Terms} $\tau$ and \emph{formulas} $\p$ of $\QMLdl$ are defined by mutual induction:
\begin{align*}
\tau ::= & \ x \mid a \mid \defdes x . \p,\\
\p ::= 
& \
\begin{aligned}[t]
P(x_1, \ldots, x_n)
\mid
x_1 = x_2 \mid 
\neg \p & \mid (\p \land \p) \mid  \ \exists x\, \p\\ \mid & \ \Diamond_{i} \p 
\mid \ \lp \lambda x . \p \rp (\tau),
\end{aligned}
\end{align*}
where
$a \in \NI$,
$P \in \NPr$ ($n$-ary),
and
$x, x_{1}, \ldots, x_{n} \in \Var$.
Standard abbreviations are assumed, and \emph{free variables} are defined as in~\cite{FitMen12}.

A \emph{partial interpretation with expanding domains}
is a structure
$\Mmf = (\Fmf, \Delta, \Imc)$,
where
$\Fmf = (W, \{ R_{i} \}_{i \in I})$ is a \emph{frame}, with $W$ being a non-empty set of \emph{worlds} and $R_{i} \subseteq W \times W$ being an \emph{accessibility relation} on $W$, for each {modality} $i \in I$;
$\Delta$ is a function associating with every $w \in W$ a non-empty set, $\Delta^{w}$, called the \emph{domain of $w$ in~$\Mmf$}, such that $\Delta^{w} \subseteq \Delta^{v}$, whenever $w R_{i} v$, for some $i \in I$;
$\I$ is a function associating with each $w \in W$ a \emph{partial} first-order interpretation $\Imc_{w}$ with domain $\Delta$ so that $P^{\Imc_{w}} \subseteq \Delta^n$,
for each predicate $P \in \NPr$ of arity $n$,
and
$a^{\Imc_{w}} \in \Delta$, for \emph{some} subset of constants $a \in \NI$,
with no additional requirement
(in particular, there is no assumption that all constants are
{rigid designators} in $\Mmf$).
%
An \emph{assignment in $\Mmf$} is a function~$\assign$ from $\Var$ to $\Delta$.
An \emph{$x$-variant} of an assignment $\assign$ is an assignment that can differ from $\assign$ only on $x$.
The definitions of \emph{value} $\tvalue{\tau}{w}{\assign}$ of term $\tau$ under assignment $\assign$
at world $w$ of $\Mmf$, and \emph{satisfaction} $\Mmf, w \models^\assign \p$ of formula $\p$ under assignment $\assign$
at world $w$ in $\Mmf$ are defined by mutual induction. First, we have
\[
\tvalue{\tau}{w}{\assign} =
\begin{cases}
\assign(x), & \text{ if } \tau \text{ is } x \in \Var; \\
a^{\Imc_{w}}, & \text{ if } \tau \text{ is } a \in \NI \text{ and $a^{\Imc_{w}}$ is defined};\\
\assign'(x), & \text{ if } \tau  \text{ is } \defdes x . \p \text{ and } \Mmf, w \models^{\assign'} \p, \text{ for }\\
& \text{ \emph{exactly one} $x$-variant $\assign'$ of $\assign$},\\
\text{undefined}, & \text{ otherwise}.
\end{cases}
\]
If $\tvalue{\tau}{w}{\assign}$ is defined, then we say that $\tau$ \emph{designates under $\assign$ at $w$} of $\Mmf$. Next, 
\[%
{\renewcommand{\arraystretch}{1.25}\renewcommand{\arraycolsep}{3pt}%
\begin{array}{lcl}
		\Mmf, w \models^{\assign} P(x_1, \ldots, x_n) & \text{iff} \ & (\assign(x_1), \ldots, \assign(x_n)) \in P^{\Imc_{w}}; \\
		\Mmf, w \models^{\assign} x_1 = x_2 & \text{iff}  \  & \assign(x_1) = \assign(x_2); \\
		\Mmf, w  \models^{\assign} \neg \varphi & \text{iff}  \ &  \Mmf, w  \not\models^{\assign} \varphi; \\
		\Mmf, w  \models^{\assign} \varphi \land \psi & \text{iff}  \  & \Mmf, w  \models^{\assign} \varphi \text{ and } \Mmf, w  \models^{\assign} \psi; \\
		\Mmf, w  \models^{\assign} \exists x\, \varphi & \text{iff}  \
		& \parbox[t][][t]{42mm}{$\Mmf, w \models^{\assign'} \varphi$, for some $x$-variant $\assign'$ of $\assign$;}\\
		\Mmf, w  \models^{\assign} \Diamond_{i} \varphi & \text{iff} \
		& \parbox[t][][t]{42mm}{$\Mmf, v \models^{\assign} \varphi$, for some $v \in W$ such that $wR_{i}v$;}\\
		\Mmf, w  \models^{\assign} \langle \lambda x . \p \rangle (\tau) & \text{iff} \
		& \parbox[t][][t]{42mm}{$\tau$ designates under $\assign$ at $w$ of $\Mmf$ and $\Mmf, w  \models^{\assign'}  \p$, where $\assign'$ is the $x$-variant of $\assign$ with $\assign'(x) = \tvalue{\tau}{w}{\assign}$.}
\end{array}
}%
\]
We now introduce the \emph{standard translation} of an $\MLALCOud$ concept $C$ into a $\QMLdl$ formula $\sttr{x}{C}$ with at most one free variable $x$, defined as follows:
\begin{align*}
	\sttr{x}{A} & = A(x),\\
	\sttr{x}{\{ a \}} & = \lp \lambda y. x = y \rp (a),\\
	\sttr{x}{\{ \defdes C \}} & = \lp \lambda y. x = y \rp (\defdes z . \sttr{z}{C}), \\
	\sttr{x}{\lnot C} & = \lnot \sttr{x}{C},\\
	\sttr{x}{C \sqcap D} & = (\sttr{x}{C} \land \sttr{x}{D}),\\
	\sttr{x}{\exists r.C} & = \exists y\, (r(x,y) \land \sttr{y}{C}),\\
	\sttr{x}{\exists u.C} & = \exists y\, \sttr{y}{C},\\
	\sttr{x}{\Diamond_{i} C} & = \Diamond_{i} \sttr{x}{C}.
\end{align*}

The following proposition shows that $\MLALCOud$ concepts can indeed be seen as a fragment of $\QMLdl$, via the standard translation above.
%

\begin{proposition}
\label{prop:standardtr}
\label{lemma:standardtr}
For every
$\MLALCOud$
concept $C$,
partial interpretation $\Mmf = (\Fmf, \Delta, \Imc)$, world $w$ of $\Mmf$, and $d \in \Delta^{{w}}$, we have
$d \in C^{\Int_{w}}$ iff
$\Mmf, w \models^{\assign} \sttr{x}{C}$, where $\assign$ is an assignment in $\Mmf$ such that $\assign(x) = d$.
\end{proposition}
\begin{proof}
By induction on the structure of $C$.

$C = A$. We have $d \in A^{\Int_{w}}$ iff $\assign(x) \in A^{\Imc_{w}}$, where $\assign$ is an assignment such that $\assign(x) = d$. That is, $\Mmf, w \models^{\assign} A(x)$.

$C = \{ a \}$.
We have $d \in \{ a \}^{\Int_{w}}$ iff $a$ designates at $w$ and $d = a^{\Imc_{w}}$, 
meaning that
$a$ designates under $\assign$ at $w$ and $\assign(x) = a^{\Imc_{w}}$, where $\assign$ is an assignment with $\assign(x) = d$.
Equivalently, $a$ designates under $\assign$ at $w$ and $\assign'(x) = \assign'(y)$, where $\assign'$ is a $y$-variant of $\assign$ such that $\assign'(y) = a^{\Imc_{w}}$.
The previous step means that $a$ designates under $\assign$ at $w$ and $\Mmf, w \models^{\assign'} x = y$, where $\assign'$ is a $y$-variant of $\assign$ such that $\assign'(y) = a^{\Imc_{w}}$.
That is, $\Mmf, w \models^{\assign} \lp \lambda y. x = y \rp (a)$.

$C = \{ \defdes C \}$.
We have $d \in \{ \defdes C \}^{\Int_{w}}$ iff $\defdes C$ designates at $w$ and $d = (\defdes C)^{\Imc_{w}}$,
i.e.,
$C^{\Imc_{w}} = \{ e \}$, for some $e \in \Delta^{w}$, and $d = e$.
This means that $C^{\Imc_{w}} = \{ e \}$, for some $e \in \Delta^{w}$, and $\assign(x) = e$, where $\assign$ is an assignment in~$\Mmf$ such that $\assign(x) = d$.
By induction hypothesis, we have equivalently that $\Mmf, w \models^{\assign'} \sttr{z}{C}$, where $\assign'$ is the (unique) $z$-variant of $\assign$ such that $\assign'(z) = e$,
and $\assign(x) = e$.
In other words, $\defdes z. \sttr{z}{C}$ designates under $\assign$ at $w$ of $\Mmf$,
and $\assign(x) = \defdes z. \sttr{z}{C}^{\Mmf, w}_{\assign}$.
The previous step can then be rewritten as:
$\defdes z. \sttr{z}{C}$ designates under $\assign$ at $w$ of $\Mmf$,
and $\assign(x) = \assign''(y)$, where $\assign''$ is a $y$-variant
of $\assign$ such that $\sigma''(y) = \defdes z. \sttr{z}{C}^{\Mmf, w}_{\assign}$.
Given that $\assign(x) = \assign''(x)$ (as $y$-variants, their values coincide on $x$), we have equivalently that
$\defdes z. \sttr{z}{C}$ designates under $\assign$ at $w$ of $\Mmf$
and $\assign''(x) = \assign''(y)$, i.e., $\Mmf, w \models^{\assign''} x = y$.
Since $\assign''$ is a $y$-variant of $\sigma$ such that $\assign''(y) = \defdes z. \sttr{z}{C}^{\Mmf, w}_{\assign}$, the last step is equivalent by definition to $\Mmf, w \models^{\assign} \lp \lambda y. x = y \rp (\defdes z. \sttr{z}{C})$.

$C = \lnot D$. We have $d \in (\lnot D)^{\Imc_{w}}$ iff $d \not \in D^{\Imc_{w}}$. By induction hypothesis, we have equivalently $\Mmf, w \not \models^{\assign} \sttr{x}{D}$, where $\assign$ is an assignment in $\Mmf$ such that $\assign(x) = d$.
That is,  $\Mmf, w \models^{\assign} \lnot \sttr{x}{D}$.

$C = D \sqcap E$.
We have $d \in (D \sqcap E)^{\Imc_{w}}$ iff $d \in D^{\Imc_{w}}$ and $d \in E^{\Imc_{w}}$. By induction hypothesis, we have equivalently $\Mmf, w \models^{\assign'} \sttr{x}{D}$ and $\Mmf, w \models^{\assign''} \sttr{x}{E}$, where $\assign', \assign''$ are assignments in $\Mmf$ such that $\assign'(x) = \assign''(x) = d$.
Hence, since the truth of $\sttr{x}{D}$ and $\sttr{x}{E}$ depends only on the values assigned to $x$, we have equivalently that $\Mmf, w \models^{\assign} \sttr{x}{D}$ and $\Mmf, w \models^{\assign} \sttr{x}{E}$, i.e., $\Mmf, w \models^{\assign} \sttr{x}{D} \land \sttr{x}{E}$, for an assignment $\assign$ in $\Mmf$ such that $\assign(x) = d$.

$C = \exists r. D$.
We have $d \in (\exists r. D)^{\Imc_{w}}$ iff there exists
$e \in D^{\Imc_{w}}$ such that $(d,e) \in r^{\Imc_{w}}$
By induction hypothesis, this means that $(d, \sigma''(y)) \in r^{\Imc_{w}}$ and $\Mmf, w \models^{\sigma''} \sttr{y}{D}$, for some assignment $\sigma''$ in $\Mmf$ such that $\sigma''(y) = e$.
Equivalently, we have 
$(\sigma'(x), \sigma'(y)) \in r^{\Imc_{w}}$ and $\Mmf, w \models^{\sigma'} \sttr{y}{D}$, for some assignment $\sigma'$ in $\Mmf$ such that $\assign'(x) = d$ and $\assign'(y) = e$, i.e.,
$\Mmf, w \models^{\assign'} r(x,y) \land \sttr{y}{D}$.
By considering the assignment $\sigma$ defined as $\sigma'$, except possibly on $y$, we have equivalently that $\sigma'$ is a $y$-variant of $\sigma$ such that  $\Mmf, w \models^{\assign'} r(x,y) \land \sttr{y}{D}$.
Hence, the previous step means that
$\Mmf, w \models^{\assign} \exists y(r(x,y) \land \sttr{y}{D})$, where $\assign$ is an assignment in $\Mmf$ such that $\assign(x) = d$.

$C = \exists u. D$. We have $d \in (\exists u. D)^{\Imc_{w}}$ iff there exists
$e \in D^{\Imc_{w}}$.
By induction hypothesis, this is equivalent to $\Mmf, w \models^{\sigma'} \sttr{y}{D}$, for some assignment $\sigma'$ in $\Mmf$ such that $\sigma'(y) = e$.
By considering the assignment $\sigma$ defined as $\sigma'$, except possibly on $y$, we have equivalently that $\sigma'$ is a $y$-variant of $\sigma$ such that  $\Mmf, w \models^{\assign'} \sttr{y}{D}$, that is, $\Mmf, w \models^{\assign} \exists y \sttr{y}{D}$.

$C = \Diamond_{i} D$. We have $d \in (\Diamond_{i} D)^{\Imc_{w}}$ iff there exists $v \in W$ such that $w R_{i} v$ and $d \in D^{\Imc_{v}}$. By induction hypothesis, the previous step means that there exists $v \in W$ such that $w R_{i} v$ and $\Mmf, v \models^{\sigma} \sttr{x}{D}$, with $\sigma$ assignment in $\Mmf$ such that $\sigma(x) = d$.
Equivalently, by definition, $\Mmf, w \models^{\sigma} \Diamond_{i} \sttr{x}{D}$.

\end{proof}

\paragraph{Extended Examples}

The following examples extend the formalisations of the scenarios discussed in the Introduction.


\begin{example}\em
\emph{Of} Pierre, Agent 1 \emph{knows} that he is the General Chair of KR 2024:
\[
\exists u. ( \{ \mathsf{pierre} \} \sqcap \Box_{1} \{ \defdes \exists \mathsf{isGenChair}. \{ \mathsf{kr24} \} \} ),
\]

Agent 1 does not know that the General Chair of KR 2024 is the General Chair of the KR Conference held in Southeast Asia: 
\begin{multline*}
\lnot \Box_{1} \exists u. ( \{ \defdes \exists \mathsf{isGenChair}. \{ \mathsf{kr24} \} \}  \sqcap \{ \defdes \exists \mathsf{isGenChair}. \\ \{ \defdes ( \mathsf{KRConf} \sqcap \exists \mathsf{hasLoc}.\mathsf{SEAsiaLoc}) \} \} ),
\end{multline*}

Therefore, Agent 1 does not know of Pierre that he is the General Chair of the KR Conference held in Southeast Asia: 
\begin{multline*}
	\exists u. ( \{ \mathsf{pierre} \} \sqcap
								 \lnot \Box_{1}  \{ \defdes \exists \mathsf{isGenChair}. \\\{ \defdes ( \mathsf{KRConf} \sqcap \exists
\mathsf{hasLoc}.\mathsf{SEAsiaLoc}) \} \} )
\end{multline*}
\end{example}

%
%
%
%

\begin{example}\em

Agent 2 knows \emph{that} Agent 1 knows \emph{of} the General Chair of KR 2024 that they are busy:
\[
\Box_{2} \exists u. ( \{ \defdes \exists \mathsf{isGenChair}. \{ \mathsf{kr24} \} \} \sqcap \Box_{1} \mathsf{Busy} ),
\]
abbreviated as $\Box_{2} [\Box_{1} \mathsf{Busy}( \defdes \exists \mathsf{isGenChair}. \{ \mathsf{kr24} \} ) ]$.

However, Agent 2 also knows that Agent 1 does \emph{not} know \emph{that} the General Chair of the KR Conference held in Southeast Asia is busy:
\begin{multline*}
\Box_{2} \lnot \Box_{1} \exists u. ( \{ \defdes \exists \mathsf{isGenChair}. \\ \{ \defdes ( \mathsf{KRConf} \sqcap \exists
\mathsf{hasLoc}.\mathsf{SEAsiaLoc}) \} \} \sqcap  \mathsf{Busy}  ),
\end{multline*}
abbreviated as 
\begin{multline*}
\Box_{2} \lnot \Box_{1} [\mathsf{Busy} ( \{ \defdes \exists \mathsf{isGenChair}. \\ \{ \defdes ( \mathsf{KRConf} \sqcap \exists
\mathsf{hasLoc}.\mathsf{SEAsiaLoc}) \} \}  ) ].
\end{multline*}
Hence, Agent 2 knows that Agent 1 does not know of the General Chair of KR 2024 that they are the General Chair of the KR Conference held in Southeast Asia:
\begin{multline*}
\Box_{2} \exists u. (
\{ \defdes \exists \mathsf{isGenChair}. \{ \mathsf{kr24} \} \} \sqcap
\lnot \Box_{1} \{ \defdes \exists \mathsf{isGenChair}. \\ \{ \defdes ( \mathsf{KRConf} \sqcap \exists \mathsf{hasLoc}.\mathsf{SEAsiaLoc}) \} \}  ),
\end{multline*}
abbreviated as
\begin{multline*}
\Box_{2} [
\lnot \Box_{1} \{ \defdes \exists \mathsf{isGenChair}.  \{ \defdes ( \mathsf{KRConf} \sqcap \\ \exists \mathsf{hasLoc}.\mathsf{SEAsiaLoc}) \} \} (\{ \defdes \exists \mathsf{isGenChair}. \{ \mathsf{kr24} \} \} )].
\end{multline*}
\end{example}


\begin{example}\em
KR24 is a rigid designator:
\[
	\{ \mathsf{kr24} \} \sqsubseteq \Box \{ \mathsf{kr24} \}, \quad \Diamond \{ \mathsf{kr24} \} \sqsubseteq \{ \mathsf{kr24} \}.
	\]
KR24 is the current KR Conference, but from next year there will be more:
	\begin{align*}
	& \exists u. (\{\mathsf{kr24}\} \sqcap \{ \mathsf{kr} \}),  \\
	& \top \sqsubseteq \Diamond^{+}  \exists u. \{\mathsf{kr}\} \\
	& \{\mathsf{kr}\} \sqsubseteq \lnot \Next \{ \mathsf{kr} \}
	\end{align*}
Whoever is a Program Chair of KR will not be Program Chair again, but always becomes either the General Chair or a PC member of next year's KR:
	\begin{multline*}
\exists {\sf isProgChair}.\{ {\sf kr} \} \sqsubseteq \lnot \Diamond \exists {\sf isProgChair}.\{ {\sf kr} \} \sqcap \\
							 ( \{ \defdes \exists {\sf isGenChair}. \Next \{ {\sf kr} \} \} \sqcup \exists {\sf isPCMember}. \Next \{ {\sf kr } \} ).
	\end{multline*}
\end{example}


\section{Proofs for Section~\ref{sec:reduction}}\label{app:reductions}

Given an $\MLALCOud$ concept $C$ and its subconcept $B$, we define the set of \emph{$B$-relevant paths in $C$} by induction on the structure of $C$ as follows ($D$ is a subconcept of $C$  possibly containing $B$ as its own subconcept):
\begin{align*}
\rpath(D, B) & = \begin{cases}\{\epsilon\}, & \text{ if  } D = B,\\\emptyset, & \text{ otherwise},\end{cases} \\
\rpath(\{ \defdes D \}, B) & = \rpath(D, B), \\
\rpath(\neg D, B) & = \rpath(D, B), \\
\rpath(D_1 \sqcap D_2, B) & = \rpath(D_1, B) \cup \rpath(D_2, B),\\
\rpath(\exists s.D, B) & = \rpath(D, B),  \text{ for  } s \in \NR \cup \{ u \},\\
\rpath(\Diamond_i D, B) & =  \{ i \cdot \pi \mid \pi \in \rpath(D, B) \}.
\end{align*}
It can be seen that  we have the following inclusions, which will be helpful in proofs by induction on the structure of concepts:
%
\begin{align*}
\rpath(C, \{\defdes B\}) & \subseteq \rpath(C,  B),\\
\rpath(C, \neg B) & \subseteq \rpath(C,  B),\\
\rpath(C, B_1 \sqcap B_2) & \subseteq \rpath(C,  B_1) \cap \rpath(C,  B_2),\\
\rpath(C, \exists s.B) & \subseteq \rpath(C,  B), \text{ for } s \in \NR \cup \{ u \},\\
\rpath(C, \Diamond_i B) & \subseteq \{ \pi \mid \pi \cdot i \in \rpath(C,  B)\}.
\end{align*}
%
Note
that the sets on the left are included in the sets on the right
because, for example, concept $C$ may contain fewer occurrences of $\neg B$ than of $B$, and so, not every $B$-relevant path in $C$ is a $\neg B$-relevant path in $C$.
The set of $B$-relevant paths in $C$ induces the set of worlds that are reachable via these sequences of~$\Diamond_{i_j}$ operators: given a world $w\in W$, we denote by $\rpathw(w, C, B)$ the set of \emph{$B$-relevant worlds for $C$ and $w$}, consisting of worlds $v\in W$ such that $w_0 R_{i_1} w_1 R_{i_2} \cdots R_{i_n} w_n$ for $(i_1, i_2, \dots, i_n)\in \rpath(C, B)$ and $w_0 = w$ and $w_n = v$. Note that $\rpathw(w, C, C) = \{ w\}$.
Moreover, the inclusions between
$\rpath(C, B)$ given above naturally translate into the following for $\rpathw(w, C,B)$:
\begin{align*}
\rpathw(w, C, \{\defdes B\}) & \subseteq \rpathw(w, C,  B),\\
\rpathw(w, C, \neg B) & \subseteq \rpathw(w, C,  B),\\
\rpathw(w, C, B_1 \sqcap B_2) & \subseteq \rpathw(w, C,  B_1) \cap \rpathw(w, C,  B_2),\\
\rpathw(w, C, \exists s.B) & \subseteq \rpathw(w, C,  B), \text{ for } s \in \NR \cup \{ u \},\\
\rpathw(w, C, \Diamond_i B) & \subseteq \{ v \mid vR_iu \text{ and } u \in\rpathw(w, C, B) \}.
\end{align*}

\rdatononrda*
\begin{proof}
%
Let $A$ be a concept name
and $\Omc$ an ontology. Define $\Omc'$ by adding to $\Omc$ the CIs
%
 \begin{equation}
 \label{eq:rdacis}
 \{a\} \sqsubseteq \Box_i\{a\},
  \end{equation}
for all $i\in I$ and all individual names $a$ occurring in~$\Omc$.
It can be seen that $A$
is $\Cmc$-satisfiable under $\Omc$ with the RDA iff $A$
is $\Cmc$-satisfiable under $\Omc'$.
Indeed, the $(\Rightarrow)$ direction is immediate, since
an
interpretation with
the
RDA
that satisfies~$\Omc$ satisfies the CIs of the form~\eqref{eq:rdacis}, and hence $\Omc'$.
For the $(\Leftarrow)$ direction, suppose that $\Mmf'$ is
an interpretation based on a frame from $\Cmc$, with either constant or expanding domains, such that $\Mmf' \models \Omc'$ and $A^{\Imc'_w} \neq \emptyset$, for some $w\in W$.
Due to the CIs of the form~\eqref{eq:rdacis}, 
for every $w, v \in W$ with $w R_{i} v$,  if $a$ occurs in $\Omc$ and $a^{\Imc'_{w}}$ is defined, then $a^{\Imc'_{v}}$ is defined and $a^{\Imc'_{w}} = a^{\Imc'_{v}}$.
We can then define $\Mmf$ as $\Mmf'$, except that all individual names $a$ that do not occur in $\Omc$ now fail to designate in $\Mmf$ in every world.
it can be seen that $\Mmf$ is based on the same frame and the same domains, satisfies the RDA and $A$ under $\Omc$.
\end{proof}

\begin{proposition}
\label{prop:rdatononrda:total}
In $\ML^n_{\ALCOu}$ and $\MLALCOud$,
total
concept $\Cmc$-satisfiability \textup{(}under global ontology\textup{)} with the RDA is poly\-time-reducible to
total
concept $\Cmc$-satisfiability \textup{(}under global ontology, respectively\textup{)},
with both constant and expanding domains.
\end{proposition}

\begin{proof}
An argument similar to the proof of Proposition~\ref{prop:rdatononrda} can be used to show that a concept name $A$
is satisfied in a total interpretation with the RDA based on a frame from $\Cmc$ under $\Omc$ iff $A$
is satisfied in a total interpretation based on a frame from $\Cmc$ under $\Omc$ extended with CIs of the form~\eqref{eq:rdacis}.



For the concept satisfiability problem, without global ontology, the reduction above needs to be modified. Let $C$ be a $\MLALCOud$ concept.
Denote by $C'$ the concept obtained from $C$ by adding to it
the conjuncts of the following form,
for individual names $a$ occurring in $C$:
\begin{equation}
\label{eq:rdainernal}
\Box^{\pi} \forall u.(\{a\} \Rightarrow \Box_{i} \{a\}), \text{ for all } \pi \cdot i \in \rpathp(C, \{ a \}),
\end{equation}
where $\rpathp(C, \{ a \})$ is the closure under taking prefixes of $\rpath(C, \{ a \})$.
%
%
We show that $C$ is $\Cmc$-satisfiable with the
RDA iff $C'$ is $\Cmc$-satisfiable.

For the $(\Rightarrow)$ direction, observe that any
interpretation $\Mmf$ with the 
RDA that satisfies $C$ also satisfies $C'$: all the conjuncts of the form~\eqref{eq:rdainernal} hold by the definition of the RDA.

For the $(\Leftarrow)$ direction, suppose that $\Mmf'$ is
a total
interpretation based on a frame from $\Cmc$ that satisfies $C'$ at world~$w$.
We define an interpretation $\Mmf$
that coincides with $\Mmf'$,
except that, for every $v \in W$, we set $a^{\Imc_{v}} = a^{\Imc'_{w}}$ (observe that, since $\Mmf'$ is total, $a^{\Imc'_{w}}$ is defined for every $a$).
Thus, $\Mmf$ is a total interpretation
satisfying the
RDA.

It remains to show the following:
\begin{claim}\label{claim:total:rda}
$C^{\Imc'_{w}} = C^{\Imc_{w}}$.
\end{claim}

\begin{proof}
By induction on the structure of $C$: for each subconcept $B$ of $C$, we show that 
\begin{equation}\label{eq:total:rda:claim}
B^{\Imc'_{v}} = B^{\Imc_{v}}, \text{ for every } v \in \rpathw(w, C, B).
\end{equation} 

The base case of $B = A$, for a concept name $A$, is straightforward.

For the base case of $B = \{ a \}$, where $a$ is an individual name, both $\{ a \}^{\Imc'_{v}}$ and $\{ a\}^{\Imc_{v}}$ are defined as the interpretations are total.
In addition, by construction, $\{ a \}^{\Imc'_{v}}$ coincides with $\{ a \}^{\Imc'_{w}} = \{ a \}^{\Imc_{w}}$, both of which are also defined. On the other hand, $\{ a\}^{\Imc_{v}}$ is equal to $\{ a\}^{\Imc_{w}}$ 
due to conjuncts~\eqref{eq:rdainernal} applied along the path connecting $w$ to $v$.
Thus, we obtain $a^{\Imc_v}  = \{ a\}^{\Imc_v}$, for every $v\in \rpathw(w, C, \{ a \})$, as required.

For the induction step, we need to consider the following cases. 

If $B$ is of the form $\{ \defdes B'\}$ or $\neg B'$ or $\exists s.B'$, then~\eqref{eq:total:rda:claim} is immediate from the induction hypothesis as $\rpathw(w,C, B) \subseteq  \rpathw(w,C, B')$. 

If $B$ is of the form $B_1 \sqcap B_2$, then~\eqref{eq:total:rda:claim} is also immediate from the induction hypothesis as $\rpathw(w,C,B) \subseteq  \rpathw(w,C, B_1)\cap  \rpathw(w,C,B_2)$. 

Finally, if $B = \Diamond_i B'$, then, by the induction hypothesis, we have $(B')^{\Imc'_v} = (B')^{\Imc_v}$, for all $v\in \rpathw(w,C, B')$. Since
$\rpathw(w,C, \Diamond_i B') \subseteq \{ v \mid v R_{i} u \ \text{and} \ u \in \rpathw(w, C, B')\}$, we obtain  $(\Diamond_i B')^{\Imc'_v} = (\Diamond_i B')^{\Imc_v}$, for all $v\in \rpathw(w,C, \Diamond_i B')$, as required.

This completes the proof of Claim~\ref{claim:total:rda}.
\end{proof}

%
It follows that $C$ is satisfied at $w$ in a
total
interpretation $\Mmf$ with the
RDA based on a frame in $\Cmc$.
  \end{proof}

Observe
that, for the reduction to hold in the global ontology case, an individual name $a$ does not have to designate at every world.
For instance, consider an
interpretation $\Mmf$
with $W = \{ w, v \}$, $wRv$, $\Delta = \{ d \}$,
and such that $a$ does not designate in $w$ but $a^{\Imc_{v}} = d$.
In this example, $a$ is rigid and the CI of the form~\eqref{eq:rdacis} is satisfied in both $w$ and $v$:
$a$ does not designate at $w$, and $v$ has no $R$-successors.


However, in the concept satisfiability case, we cannot enforce that any \emph{partial} interpretation satisfying at some world
$C'$,
i.e., 
$C$ and the additional conjuncts of the form~\eqref{eq:rdainernal}, can be transformed into a partial interpretation with the RDA that satisfies $C$, as witnessed by the following counterexample.
Let $C$ be the concept
\[
\forall u. \lnot \{ a \}  \sqcap \Box \exists u.\{ a \}, 
\]
and let $C'$ be the conjunction of $C$ with the concept $\forall u.(\{ a \} \Rightarrow \Box \{ a \})$.
Consider a partial interpretation $\Mmf$ with constant domain such that $w R v_{i}$ and $v_{i} R u$, for $i = 1,2$, and $\Delta = \{ d, e \}$. Moreover, let $a$ be non-designating in $w$, while $a^{\Imc_{v_{1}}} = d$, $a^{\Imc_{v_{2}}} = e$, and $a^{\Imc_{u}} = d$. It can be seen that $\Mmf$ is a partial interpretation satisfying $C'$ in $w$. However, we cannot turn $\Mmf$ into a partial interpretation with the RDA that satisfies $C$.
This is due to the fact that $C'$ is bounded on the levels of successors of $w$ that it can reach, whereas, under global ontology, the corresponding CI $\{ a \} \sqsubseteq \Box \{ a \}$ would hold at all worlds of the interpretation.

\redtotaltopartial*
\begin{proof}
  
First, consider the case of satisfiability under global ontology: let
$A$
be a concept
name
and $\Omc$ an ontology. Define $\Omc'$ by adding to $\Omc$ the CIs
\begin{equation}\label{eq:designate}
\top \sqsubseteq \exists u. \{ a \}
\end{equation}
for every $a$ occurring in $\Omc$.
Trivially, every total interpretation $\Mmf$ satisfying
$A$
under $\Omc$ also satisfies
$A$
under $\Omc'$. Conversely, every  interpretation $\Mmf$ satisfying
$A$
under $\Omc'$ can be easily extended to a total interpretation satisfying
$A$
under $\Omc'$ (and so under $\Omc$), as it only remains to choose the interpretation of individual names not occurring in
$\Omc$, which can be done in an arbitrary way. 

Now, consider the case of satisfiability: let $C$ be a concept. 
For the conjunction $C'$ of $C$ with concepts of the form
%
\begin{equation*}
\Box^\pi \exists u.\{a\}, \text{ for all } \pi \in \rpath(C, \{ a \}),
\end{equation*}
for each $a$ occurring in $C$,
it can be seen that any total interpretation satisfying $C$ satisfies $C'$ too; conversely, any interpretation $\Mmf'$ satisfying $C'$ can be turned into a total interpretation $\Mmf$ satisfying $C$, by ensuring that every individual name designates at each world. The proof that $C^{\Imc'_{w}} = C^{\Imc_{w}}$ is inductive and relies on the fact that $\Mmf'$ and $\Mmf$ coincide on the interpretation of every individual name $a$ from $\Omc$ in all $v\in  \rpathw(w, C, \{ a \})$; see Claim~\ref{claim:total:rda} for a similar proof.
%
%
\end{proof}

\redpartialtototal*
\begin{proof}
  For satisfiability under global ontology, let $A$ be a concept name
  and $\Omc$ an ontology. Let $\Omc'$ be the ontology obtained from $\Omc$ by replacing every nominal~$\{ a \}$ in~$\Omc$ with a fresh concept name
$N_{a}$,
and by adding the CIs
$N_a \sqsubseteq \{ a \}$,
  for all individual names $a$
  occurring in $\Omc$.
  It can be seen that $A$ is $\Cmc$-satisfiable
  under~$\Omc$ iff $A$ is $\Cmc$-satisfiable under~$\Omc'$ in total
  interpretations.
  
$(\Rightarrow)$ Given an interpretation $\Mmf$ such that $\Mmf \models \Omc$ and $A^{\Imc_{w}} \neq \emptyset$, for some world $w \in W$, define the total interpretation
$\Mmf'$ defined as $\Mmf$, except the following, for every $w \in W$:
\begin{itemize}
	\item $N_a^{\Imc'_{w}} = \{ a \}^{\Imc_{w}}$, for $a$ that occurs in $\Omc$;
	\item $a^{\Imc'_{w}} = a^{\Imc_{w}}$, if $a$ designates at $w$ in $\Mmf$; and $a^{\Imc'_{w}}$ is arbitrary, otherwise, for any individual name $a$.
\end{itemize}
It can be seen that
$\Mmf' \models \Omc'$ and $A^{\Imc'_{w}} \neq \emptyset$.

$(\Leftarrow)$
Given a  total interpretation $\Mmf'$ such that $\Mmf' \models \Omc'$ and $A^{\Imc'_{w}} \neq \emptyset$, for some world $w \in W$, we define an interpretation $\Mmf$ that coincides with $\Mmf'$, except for the following, for every $w \in W$:
\begin{itemize}
\item $a^{\Imc_{w}} = d$, if $a$ occurs in $\Omc$ and $N_{a}^{\Imc'_{w}} = \{ d \}$, for some $d \in \Delta^{w}$; and $a$ fails to designate at $w$ in $\Mmf$, otherwise.
\end{itemize}
It can be seen that $\Mmf \models \Omc$ and $A^{\Imc_{w}} \neq \emptyset$.

For the concept satisfiability problem, the reduction above is modified similarly to the proof of Proposition~\ref{lemma:redtotaltopartial}: let $C'$ be the conjunction result of replacing each $\{a\}$ with $N_a$ in $C$ and the following additional conjuncts
%
\begin{equation*}
\Box^{\pi}\forall u. (N_{a} \Rightarrow \{ a \}), \text{ for all  } \pi \in \rpath(C, \{ a\}).
\end{equation*}
%
%
It can be seen that any interpretation $\Mmf$ satisfying $C$ gives rise to a total interpretation $\Mmf'$ satisfying $C'$ by using the extension of $\{a\}^{\Imc_v}$ to interpret both $a$ and $N_a$ at any $v$ in $\Mmf'$ and arbitrarily choosing the value of $a$ at $v$ in $\Mmf'$ if it is undefined in $\Mmf$. The result clearly satisfies $C'$. Conversely, the additional conjuncts guarantee that $N_a$ behaves like a nominal term (its extension contains at most one element at every world) and so $C$ is satisfied in an interpretation $\Mmf$ obtained from any interpretation $\Mmf$ by redefining each $a$ as the corresponding $N_a$ (which may give a partial interpretation, of course). Again, the proof that $C^{\Imc'_{w}} = C^{\Imc_{w}}$ is inductive and relies on the fact that $\Mmf'$ and $\Mmf$ coincide on the interpretation of every individual name $a$ from $\Omc$ in all $v\in  \rpathw(w, C, \{ a \})$; see Claim~\ref{claim:total:rda} for a similar proof.
\end{proof}


An ontology is in \emph{normal form} if it consists of CIs of the form
\begin{align*}
A & \sqsubseteq \{ a \},  &  \{a\} & \sqsubseteq A,\\
A & \sqsubseteq \{ \iota A_1 \},  &  \{\iota A_1 \} & \sqsubseteq A,\\
A & \sqsubseteq \neg A_1,  &  \neg A_1 & \sqsubseteq A,\\
A & \sqsubseteq A_1\sqcap A_2, & A_1 \sqcap A_2 & \sqsubseteq A,\\
A & \sqsubseteq \exists s.A_1, & \exists s.A_1  & \sqsubseteq A,\\
A & \sqsubseteq \Diamond_i A_1, & \Diamond_i A_1  & \sqsubseteq A,
\end{align*}
where $A$, $A_1$ and $A_2$ are concept names, $s\in \NR \cup  \{u \}$ and $i\in I$.

\ontologynormalform*
\begin{proof}
Let $C$ be a subconcept in $\Omc$, and let $\Omc[C/A]$ be the result of replacing every occurrence of~$C$ in $\Omc$ with a fresh concept name $A$. Clearly, $\Omc[C/A]\cup \{ C \equiv A\}$ is a model conservative extension of $\Omc$. By repeated application of this procedure, starting from innermost connective first, we obtain in polynomial time an ontology $\Omc'$ in normal form having the same set of connectives as $\Omc$.
\end{proof}

\normalformconc*
\begin{proof}
Let $\Mmf$ be an interpretation satisfying $D$. We define an interpretation $\Mmf'$ that extends $\Mmf$ with the interpretation of $A$ by taking $A^{\Imc_w} = C^{\Imc_w}$. Clearly, $D'$ is satisfied in~$\Mmf'$. Conversely, suppose $D'$ is satisfied at $w$ in an interpretation~$\Mmf'$. The additional conjuncts ensure that $A^{\Imc'_v} = C^{\Imc'_v}$, for all $v\in \rpathw(w, D, C)$. In other words, we have $C[C/A]^{\Imc'_v} = C^{\Imc'_v}$, for all $v\in \rpathw(w, D, C)$.  We show
\begin{claim}
$D[C/A]^{\Imc'_w} = D^{\Imc'_w}$.
\end{claim}

\begin{proof}
We prove this  by induction on the structure of $D$ that, for all subconcepts $B$ of $D[C/A]$, we have 
\begin{equation}\label{eq:normal-form-conj:claim}
B[C/A]^{\Imc'_v} = B^{\Imc'_v}, \text{ for all } v\in \rpathw(w,D, B). 
\end{equation}
For the basis of induction, we need to consider two cases.
If $B$ is either a concept name different from $A$ or a term nominal $\{a\}$, then~\eqref{eq:normal-form-conj:claim} is immediate as $B[C/A]$ coincides with $B$,
If $B$ is $A$, then~\eqref{eq:normal-form-conj:claim} is by construction. 

For the induction step, we need to consider the following cases. 
If $B$ is of the form $\{ \defdes B'\}$ or $\neg B'$ or $\exists s.B'$, then~\eqref{eq:normal-form-conj:claim} is immediate from the induction hypothesis as $\rpathw(w,D, B) \subseteq  \rpathw(w,D, B')$. 
If $B$ is of the form $B_1 \sqcap B_2$, then~\eqref{eq:normal-form-conj:claim} is also immediate from the induction hypothesis as $\rpathw(w,D, B) \subseteq  \rpathw(w,D, B_1)\cap  \rpathw(w,D, B_2)$. Finally, if $B = \Diamond_i B'$, then, by the induction hypothesis, we have $B'[C/A]^{\Imc'_v} = (B')^{\Imc'_v}$, for all $v\in \rpathw(w,D, B')$. Since
$\rpathw(w, D, \Diamond_i B) \subseteq \{ v \mid v R_{i} u \ \text{and} \ u \in \rpathw(w, D, B)\}$, we obtain  $(\Diamond_i B')[C/A]^{\Imc'_v} = (\Diamond_i B')^{\Imc'_v}$, for all $v\in \rpathw(w,D, \Diamond_i B')$, as required.
\end{proof}

Also by induction on the structure of $D$, it can be seen that $\rpath(D', A) = \rpath(D, C)$, for any subconcept $C$ of $D$, and $\rpath(D',E) = \rpath(D,E)$, for any subconcept $E$ of $C$.
\end{proof}

\spypointreduction*
\begin{proof}
First, we consider the positive occurrences of the universal role.
Suppose $B \sqsubseteq \exists u.B'$ is satisfied at all $w$ in some interpretation $\Mmf$. Then we extend $\Mmf$ to $\Mmf'$ by interpreting the fresh role $r$ as the universal role at all $w$. Then $B \sqsubseteq \exists r.B'$ is satisfied at all $w$ in $\Mmf'$. 

Conversely, suppose $B \sqsubseteq \exists r.B'$ is satisfied at all $w$ in some $\Mmf'$. Then, clearly, $B \sqsubseteq \exists u.B'$ is satisfied in $\Mmf'$ as well. 

Second, we consider the negative occurrences of the universal role.
Suppose that $\exists u.B \sqsubseteq B'$ is satisfied at all $w$ in some interpretation $\Mmf$. Then we extend $\Mmf$ to $\Mmf'$ by interpreting the fresh role $r$ as the universal role at all $w$, choosing one arbitrary element as  the interpretation of the fresh nominal $s$, and, for every $w\in W$,  including $e^{\Imc'_w}$ in $A^{\Imc'_w}$   if $B^{\Imc_w} = \emptyset$ and leaving $A^{\Imc'_w}$ empty otherwise. We show that the four CIs are satisfied at all $w$ in $\Mmf'$. The first and second CIs are satisfied by consruction. If $(\neg B')^{\Imc'_w} =\emptyset$, then the third CI  is trivially satisfied at $w$ in $\Mmf'$. So, assume that $(\neg B')^{\Imc'_w} \ne\emptyset$. As $\exists u.B \sqsubseteq B'$ is satisfied at $w$ in $\Mmf$, we have $B^{\Imc'_w} = \emptyset$ and so the third CI is also satisfied in $\Mmf'$ because, by construction, $e^{\Imc'_w} \in A^{\Imc'_w}$. To show that the last CI is also satisfied at $w$ in $\Mmf'$, suppose that $e^{\Imc'_w}\in A^{\Imc'_w}$. Then, by construction, $B^{\Imc_w}  = \emptyset$ and so every domain element is in $(\neg B)^{\Imc'_w}$, as required by the fourth CI.

Conversely, suppose the four CIs are satisfied at all $w$ in some $\Mmf'$. By the first CI, $(d, e^{\Imc'_w})\in r^{\Imc'_w}$, for every element $d$ of the domain. Assume first that $ B^{\Imc'_w}\ne\emptyset$.  Take any $d \in B^{\Imc'_w}$. By the fourth CI, $e^{\Imc'_w} \notin A^{\Imc'_w}$. Therefore, by the third CI, every element of the domain is in $(B')^{\Imc'_w}$ for otherwise $e^{\Imc'_w} \in A^{\Imc'_w}$. Thus, $\exists u.B \sqsubseteq B'$ is satisfied at all $w$ in $\Mmf'$.
\end{proof}

\nomnonrdatodefdes*
\begin{proof}
We first consider the case of satisfiability under global ontology.
Given a concept name $A$ and an $\MLALCO$ ontology $\Omc$, take a fresh concept name $N_a$ for each individual name $a$ in $\Omc$, and let $\Omc'$ be the result of replacing every occurrence of $\{ a \}$ in
$\Omc$
with~$\{ \defdes N_a \}$. 
It can be seen that,
if
$A$ is satisfied in an interpretation $\Mmf$ under $\Omc$,
then
$A$ is satisfied in $\Mmf'$ under $\Omc'$,
where $\Mmf'$ is obtained from $\Mmf$ by setting, for every $w$ in $\Mmf$, $N_a^{\Imc'_{w}} = \{ a^{\Imc_{w}} \}$, if $a$ designates at $w$ in $\Mmf$, and  $N_a^{\Imc'_{w}} = \emptyset$, otherwise; note that the interpretation of nominals  can be chosen arbitrarily as $\Omc'$ does not contain them.
Conversely, if
$A$ is satisfied in an interpretation $\Mmf'$ under $\Omc'$, then $A$ is satisfied in $\Mmf$ under $\Omc$, where $\Mmf$ is obtained from $\Mmf'$ by setting $a^{\Imc_{w}} = (\defdes N_a)^{\Imc'_{w}}$, for every $w$ in $\Mmf$ and every $a$ occurring in $\Omc$; note that the resulting interpretation is not necessarily total and does not necessarily satisfy the RDA.

The same construction and argument work for the satisfiability problem too.
\end{proof}

\redmludtomlu*
\begin{proof}
By Propositions~\ref{lemma:redpartialtototal} and~\ref{lemma:redtotaltopartial}, it is sufficient to reduce $\MLALCOud$ \emph{total} concept $\Cmc$-satisfiability \textup{(}under global ontology\textup{)} to \emph{total} $\MLALCOu$ concept $\Cmc$-satisfiability \textup{(}under global ontology, respectively\textup{)}.

First, consider the case of satisfiability under global ontology: let $A$ be a concept name and $\Omc$ an $\MLALCOud$ ontology.  By Lemma~\ref{lemma:ontology:normal-form}, we assume that $\Omc$ is in normal form and, in particular,  (\emph{a}) all definite descriptions are applied to concept names only and (b) term nominals occur only in CIs of the form
$A_{\{\tau\}} \equiv\{ \tau \}$,
%
for $A_{\{\tau\}} \in \NC$ and terms $\tau$ that are either $b \in \NI$ or $\defdes B$, for $B\in \NC$.
%
%
We then define a translation of an $\MLALCOud$ ontology $\Omc$ in normal form into an $\MLALCOu$ ontology $\Omc\red$. 
Let $\Omc\red$ be the result of replacing each $A_{\{ \defdes  B \}} \equiv \{\defdes B\} $ in $\Omc$ with CIs of the form
%
\begin{equation*}
A_{\{\defdes B\}} \sqsubseteq B\sqcap \{a_B \} \text{ and } B \sqcap \forall u.(B\Rightarrow \{a_B\}) \sqsubseteq A_{\{\defdes B\}}, 
\end{equation*}
%
where $a_B$ is a fresh nominal. 
%
%
%
%
%
We then show the following.
\begin{claim}
\label{cla:alcodtoalcou}
$\Omc\red$ is a model conservative extension of $\Omc$.
\end{claim}
\begin{proof}
Clearly, for any interpretation $\Mmf\red$ with $\Mmf\red \models \Omc\red$, we have $\Mmf\red \models \Omc$ as well.

Conversely, given a total interpretation $\Mmf$ based on a frame in $\Cmc$  such that $\Mmf \models \Omc$,
we define the total interpretation
$\Mmf^{\ast}$ by extending $\Mmf$, for every $w \in W$, with 
\begin{itemize}
	\item
	 $a_{B}^{\Imc^{\ast}_{w}} = d$, if $B^{\Imc_{w}} = \{ d \}$, for some $d \in \Delta^{w}$; and 
	  $a_{B}^{\Imc^{\ast}_{w}}$ is arbitrary, otherwise.
\end{itemize}
It can be seen that
$\Mmf^{\ast} \models \Omc^{\ast}$.
\end{proof}

Hence, $A$ is satisfied under $\Omc$ in a total interpretation based on a frame from $\Cmc$ iff $A$ is satisfied under $\Omc\red$ in a total interpretation based on a frame from $\Cmc$.

Second, we consider the case of concept satisfiability.
Let $C$ be an $\MLALCOud$ concept.
By repeatedly applying  Lemma~\ref{lem:normalformconc}, we transform $C$ into normal form, where, in particular, each subconcept of $C$ of the form $\{\defdes B\}$, where $B$ is a concept name, has a surrogate $A_{\{ \defdes B \}}$ and corresponding conjuncts of the form~\eqref{eq:normalformconj}.
Consider now the  $\MLALCOu$  concept $C\red$ obtained from the normal form of $C$ by replacing
$\Box^\pi\forall u.(A_{\{ \defdes B \}} \equiv \{\defdes B\})$ with the following
\begin{align*}
& \Box^\pi \forall u.(A_{\{\defdes B\}} \Rightarrow B \sqcap \{a_B \}) \ \sqcap \\
& \Box^\pi\forall u.(B \sqcap \forall u.(B \Rightarrow \{a_B\}) \Rightarrow A_{\{\defdes B\}})
\end{align*}
where $a_B$ is a fresh nominal for $B$. Recall that $\pi$  range over all paths in $\rpath(C, \{\defdes B\})$.
%
%
We now show the following.

\begin{claim}
\label{cla:conceptflat}
$C\red$ is a model-conservative extension of $C$.
\end{claim}
\begin{proof}
It can be seen that any interpretation  $\Mmf\red$ that satisfies $C\red$ also satisfies $C$.

Conversely, given a total interpretation $\Mmf$ based on a frame from $\Cmc$ that satisfies $C$, we define the total interpretation $\Mmf\red$ extending $\Mmf$, for every $w \in W$, with
\begin{itemize}
	\item 
	$a_{B}^{\Imc^{\ast}_{w}} = d$, if $B^{\Imc_{w}} = \{ d \}$, for some $d \in \Delta^{w}$; and 
	 $a_{B}^{\Imc^{\ast}_{w}}$ is arbitrary, otherwise.
\end{itemize}
It can be seen that $\Mmf^{\ast}$ is a total interpretation based on a frame from $\Cmc$ that satisfies $C^{\ast}$.
\end{proof}

Hence, $C$ is total $\Cmc$-satisfiable iff $C\red$ is total $\Cmc$-satisfiable.
%
%
%
%
This concludes the proof of the proposition.
\end{proof}

\redexptoconst*
\begin{proof}
By
Propositions~\ref{lemma:redpartialtototal} and~\ref{lemma:redtotaltopartial}, it is sufficient to reduce $\MLALCOu$ \emph{total} concept $\Cmc$-satisfiability \textup{(}under global ontology\textup{)} with expanding domains to \emph{total} $\MLALCOu$ concept $\Cmc$-satisfiability \textup{(}under global ontology, respectively\textup{)} with constant domain.

Let $\Ex$ be a fresh concept name, used to represent the objects that actually \emph{exist} at a given world domain.
We introduce the $\cdot^{\exrel}$ \emph{relativisation},
mapping an
$\MLALCOu$ concept $C$ to an
$\MLALCOu$ concept $C^{\exrel}$ as follows, where $s \in \NR \cup \{ u \}$:
\begin{gather*}
	A^{\exrel} = A, \quad
	\{ a \}^{\exrel}  = \Ex \sqcap \{ a \}, \quad
	(\lnot D)^{\exrel} = \lnot D^{\exrel}, \\
	(\exists s. D)^{\exrel} = \exists s.(\Ex \sqcap D^{\exrel}), \
	(D \sqcap E)^{\exrel} = D^{\exrel} \sqcap E^{\exrel}, \\
	(\Diamond_{i} D)^{\exrel} = \Diamond_{i} (D)^{\exrel}.
\end{gather*}
The $\cdot^{\exrel}$ relativisation of CIs and ontologies is then obtained by replacing every concept $C$ with $C^{\exrel}$. 

For the case of satisfiability under global ontology, it can be seen that
$A$ is $\Cmc$-satisfiable in total interpretations with expanding domains under global ontology $\Omc$ iff
$\Ex \sqcap A$ is is $\Cmc$-satisfiable in total interpretations with constant domains under ontology that consists of $\Omc^{\exrel}$ along with the following CIs:
%
\begin{align*}
 \Ex & \sqsubseteq \Box_{i} \Ex, \text{ for all } i\in I;
\end{align*}
cf.~the translation of $\varphi$ in~\cite[Proposition 3.32~(ii), (iv)]{GabEtAl03}.

Similarly, for concept satisfiability, it can be seen that $C$ is $\Cmc$-satisfiable in total interpretations with expanding domains iff the conjunction $C'$ of $\Ex \sqcap C^{\exrel}$ with the following concepts:
\begin{align*}
& \Box^{\pi} \forall u.(\Ex \Rightarrow \Box_{i} \Ex), \text{ for all } i \in I \text{ and } \pi \cdot i \in \rpath(C, B)\\ & \hspace*{10em} \text{ with a subconcept } B \text{ of } C.
\end{align*}
is $\Cmc$-satisfiable in total interpretations with constant domains;
cf.~$C'$ in~\cite[Proposition 3.32~(ii)]{GabEtAl03}.
%
\end{proof}


\section{Proofs for Section~\ref{sec:counting}}\label{app:counting}

\diff*
\begin{proof}
(1) We first consider the case of satisfiability under global ontology.
Let $A$ be a concept name and $\Omc$ an $\MLALCOu$ ontology. We assume that $A$ occurs in~$\Omc$. Let $\con{\Omc}$ be the closure under single negation 
of the set of concepts occurring in $\Omc$. A \emph{type} for $\Omc$ is a subset $\contp$ of $\con{\Omc}$
such that
	$\neg C \in \contp$ iff $C\not \in \contp$, for all
	$\neg C \in \con{\Omc}$.
A \emph{quasistate for $D$} is a non-empty set $\settp$ of types for $\Omc$. The \emph{description of $\settp$} is the $\MLALCOu$-concept 
\begin{equation*}
\Xi_{\settp} = \forall u.(\bigsqcup_{\contp\in \settp}\contp) \sqcap \bigsqcap_{\contp\in \settp} \exists u.\contp. 
\end{equation*}
Let $\Smc_\Omc$ denote the set of all quasistates $\settp$ for $\Omc$ that are
\emph{$\ALCOu$-satisfiable}, that is, such that $\Xi_{\settp}^{t}$ is satisfiable, where $\Xi_{\settp}^{t}$ denotes the result of replacing 
every outermost occurrence of $\Diamond_{i}C$ by $A_{\Diamond_{i}C}$, for a fresh concept name $A_{\Diamond_{i}C}$.
It should be clear that $\Smc_\Omc$ can be computed in double exponential time as satisfiabilty of $\ALCOu$-concepts under $\ALCOu$ ontologies is in \ExpTime{}.

We now reserve for any $C\in \con{\Omc}$ of the form $\{a\}$ or
$\exists r.C'$ a fresh concept name $A_{C}$. Define a mapping
$\cdot^{\sharp}$ that associates with every $\MLALCOu$-concept an
$\MLDiff$-concept by replacing outermost occurrences of concepts of
the form $\exists r.C$ or concepts of the form $\{a\}$ by the
respective fresh concept name.  Let $\Omc'$ be the extension of
$\Omc^{\sharp}$
with the following CIs:
\begin{align*}
\top & \sqsubseteq \exists^{=1} u.\{a\}^\sharp, \text{ for every } a\in\NI \text{ in } \Omc, \\
\top & \sqsubseteq \bigsqcup_{\settp\in \Smc_\Omc}\Xi_{\settp}^{\sharp}.
\end{align*}
Then we have the following:
\begin{lemma}
$A$ is $\Cmc$-satisfiable under $\Omc$ iff there is a type $\contp$ for $\Omc$ such that $A\in\contp$ and $\contp^{\sharp}$ is $\Cmc$-satisfiable under $\Omc'$. 
\end{lemma}
\begin{proof}
Suppose $A$ is satisfied at $w\in W$ in an $\MLALCOu$ interpretation $\Mmf$ based on a frame from $\Cmc$ and $\Omc$ is satisfied in $\Mmf$. Then we can define the type $\contp$ by reading it off from $e\in A^{\Imc_w}$:
\begin{equation*}
\contp = \bigl\{ C \in\con{\Omc} \mid e\in C^{\Imc_w} \bigr\}.
\end{equation*}
 Then, we define an $\MLDiff$ interpretation $\Mmf'$ from $\Mmf$ by providing the interpretation of the additional concept names $A_C$  for $\MLALCOu$  concepts $C$ of the form $\{a\}$ and $\exists r.C'$: we simply set $A_C^{\Imc'_v} = C^{\Imc_v}$, for all $v\in W$. It should be clear that $\Mmf'$ satisfies $\Omc'$ and that $e\in (\contp^{\sharp})^{\Imc'_w}$.

Conversely, suppose there is a type $\contp$ for $\Omc$ such that $A\in\contp$ and $\contp^{\sharp}$ is satisfied in an $\MLDiff$ interpretation $\Mmf'$ based on  a frame from $\Cmc$ such that  $\Omc'$ is also satisfied in~$\Mmf'$. We define $\Mmf$ by extending $\Mmf'$ by suitable interpretations of nominals and role names. Let $v\in W'$. Observe that for any nominal $a$ in $\Omc$, $(a^{\sharp})^{\mathcal{I}'_v}$ is a singleton, by definition of~$\Omc^{\sharp}$. Hence we define $a^{\Imc_v}$ as its single element. For any role name $r$, we define $r^{\Imc_v}$ as the maximal relation that is consistent with existential restrictions in the types we aim to satisfy: for $d\in (\contp^{\sharp})^{\Imc'_v}$ and $d'\in (\contp'^{\sharp})^{\Imc'_v}$, let $(d,d')\in r^{\Imc_v}$ if, for all $\exists r.D \in \con{\Omc}$, if $D\in \contp'$, then $\exists r.D\in \contp$.
It is easy to see that $\Mmf$ is as required.
\end{proof}

For the concept $\Cmc$-satisfiability problem we adapt the proof in the same way as in Proposition~\ref{lemma:redtotaltopartial}: the additional CIs of $\Omc'$ need to be replaced by conjuncts with suitable prefixes of the box operators and the universal role restrictions.

\bigskip

(2) \citeauthor{DBLP:journals/jphil/GargovG93}~(\citeyear{DBLP:journals/jphil/GargovG93}) observed that the logic of
the difference modality and nominals have the same expressive power.
Their technical result can be presented in our setting in the
following way. Let $C$
be an $\MLDiff$-concept and $\exists^{\ne} u.B$ its subconcept. Denote by $C^{\ddagger_B}$ the result of
replacing every occurrence of $\exists^{\ne} u.B$ in $D$ with
\begin{equation*}
\exists u.B  \sqcap (\{ a_B \} \Rightarrow \exists u.(\neg \{ a_B \} \sqcap B)),
\end{equation*}
where $a_B$ is a fresh nominal associated with concept $B$ (note that we replace only a single subconcept here). Let
$\Omc^{\textsf{wt}}_B$ consist of the following CI:
\begin{equation*}
B \sqsubseteq \exists u.(\{ a_B \} \sqcap B). 
\end{equation*}
Then we can extend the proof of Lemmas
4.2 and 4.3 in~\cite{DBLP:journals/jphil/GargovG93} to the modal
setting (provided that interpretations do not satisfy the RDA, and so nominals can be interpreted differently in different worlds) to obtain the following:
\begin{lemma}\label{lemma:gg}
For any interpretation $\Mmf$ satisfying $\Omc^{\textup{\textsf{wt}}}_B$, we have $C^{\Imc(w)} = (C^{\ddagger_B})^{\Imc(w)}$, for all $w$ in $\Mmf$.
\end{lemma}

This result gives us the reduction for the case of satisfiability under global ontology. Let $\Omc$ be an $\MLDiff$-ontology and $A$ a concept name. By Lemma~\ref{lemma:ontology:normal-form}, we can assume that $\Omc$ is in normal form. 
Denote by $\Omc^\ddagger$ the result of replacing every CI $C_1 \sqsubseteq C_2$ in $\Omc$ with $C_1^\ddagger \sqsubseteq C_2^\ddagger$, where $C_i^\ddagger$  is the result of applying $\cdot^{\ddagger_B}$ to $C_i$, for every subconcept of the form $\exists^{\ne} u.B$; note that the order of applications of the
$\cdot^{\ddagger_B}$
does not matter as the ontology is in normal form. Let $\Omc' = \Omc^\ddagger\cup \Omc^{\textsf{wt}}$, where $\Omc^{\textsf{wt}}$ is the union of all $\Omc^{\textsf{wt}}_{B}$, for concepts $\exists^{\ne} u.B$ in $\Omc$.  It should be clear that, if $A$ is satisfied under $\Omc$ at $w$ in an interpretation $\Mmf$ based on frame from $\Cmc$, then we can extend $\Mmf$ to $\Mmf'$ by interpreting the nominals $a_B$ for concepts of the form $\exists^{\ne} u.B$: for each such nominal, if $B^{\Imc_v}\ne\emptyset$, then we pick any $e\in B^{\Imc_v}$ and assign $a_B^{\Imc_v} = e$; otherwise, we leave $a_{B}^{\Imc_v}$ undefined.
Clearly, $\Mmf'$ satisfies $\Omc^{\textsf{wt}}$. Thus, by Lemma~\ref{lemma:gg}, $\Mmf'$ satisfies $\Omc^\ddagger$, and also $A$ is satisfied  at $w$ in $\Mmf'$. 

Conversely, if $A$ is satisfied under $\Omc'$ at $w$ in an interpretation $\Mmf'$ based on a frame from $\Cmc$, then, by Lemma~\ref{lemma:gg}, $A$ is satisfied under $\Omc$ at $w$ in $\Mmf'$. We can now take the reduct $\Mmf$ of $\Mmf'$ by removing the interpretation of nominals: as $\Omc$ does not contain nominals, $A$ is satisfied under $\Omc$ at $w$ in $\Mmf$.

Next, we consider the case of concept satisfiability. Let $D$
be an $\MLDiff$-concept. We repeatedly apply Lemma~\ref{lem:normalformconc}, starting from the innermost occurrences 
of subconcepts of the form $\exists^{\ne} u.C$, to transform $D$ into normal form $D'$, where each `elsewhere' quantifier is applied
only to concept names. We then construct $D^\ddagger$ by applying the $\cdot^{\ddagger_B}$ to $D'$, for each $B$ such that $\exists^{\ne} u.B$ occurs in $D'$; again, the order of applications does not matter as $D'$ is in normal form. Consider the conjunction of $D^\ddagger$ with all concepts of the form
\begin{equation*}
\Box^\pi\forall u.(B \Rightarrow \exists u.(\{ a_B \} \sqcap B)),
\end{equation*}
for $\pi\in\rpath(D', \exists^{\ne} u.B)$ and concepts of the form $\exists^{\ne} u.B$ in~$D'$ (note that, by Lemma~\ref{lem:normalformconc}, $\rpath(D', \exists^{\ne} u.B)$ coincides with $\rpath(D, \exists^{\ne} u.C)$ for the respective $C$ in $D$). We can easily modify the argument presented above (see the proof of Proposition~\ref{lemma:redtotaltopartial}) to show that $D$ is satisfied in an interpretation $\Mmf$ iff $D'$ is satisfied in an interpretation $\Mmf'$ that additionally interprets nominals $a_B$ as `witnesses' for concepts of the form $\exists^{\ne} u.B$.  
\end{proof}	

\section{Proofs for Section~\ref{sec:reasoning}}

\subsection{Proof of Theorem~\ref{thm:modconcdec}}

We first observe that for $\Sfive$, concept satisfiability under global ontology can be reduced to concept satisfiability without ontology by using the concept $C_{\mathcal{O}} \sqcap \Box C_{\mathcal{O}}$ with $C_{\mathcal{O}}$ the conjunction of all $C \Rightarrow  D$ with $C\sqsubseteq D\in \mathcal{O}$.

\modconcdec*
\begin{proof}
By
Proposition~\ref{lemma:redmludtomlu},
we can drop definite descriptions (hence consider $L_{\ALCOu}$) and by Proposition~\ref{lemma:redpartialtototal}, it is sufficient to consider total satisfiability. By Proposition~\ref{prop:redexptoconst}, satisfiability with expanding domains can be reduced in polytime to satisfiability with constant domain, so we consider constant domain models. We give the proof for $\Sfive^{n}_{\ALCOu}$. The proof for $\K^{n}_{\ALCOu}$ can then be derived in a straightforward manner. 

We show the exponential finite model property (every $\Sfive^{n}_{\ALCOu}$-concept that is satisfiable, is satisfiable in a model of exponential size). The \NExpTime{} upper bound then follows directly. We proceed in a number of steps. First, we argue that tree-shaped models are sufficient. Then, we introduce quasimodels as a surrogate for models that is easier to manipulate and show the exponential finite model property for tree-shaped quasimodels. Finite exponential-size models are then easily constructed.     

Recall that a frame $\mathfrak{F}=(W,R_{1},\ldots,R_{n})$ is an  $\Sfive^{n}$-frame if all $R_{i}$ are equivalence relations. Then such an $\Sfive^{n}$-frame $\mathfrak{F}=(W,R_{1},\ldots,R_{n})$ is 
called a \emph{tree-shaped $\Sfive^{n}$-frames} if there exists a $w_{0}\in W$
such that the domain $W$ of $\mathfrak{F}$ is a prefix-closed set of 
words of the form 
\begin{equation}\label{eq:world}
\avec{w}=w_{0}i_{0}w_{1}\cdots i_{m-1}w_{m},
\end{equation}
where $1\leq i_{j}\leq n$, $i_{j}\not=i_{j+1}$, and each $R_{i}$ is 
the smallest equivalence relation containing all pairs of the form $(\avec{w},\avec{w}iw)\in W\times W$. We call $w_{0}$ the root of $\mathfrak{F}$. 
We define the \emph{depth} of $\mathfrak{F}$ as the maximal $m$ such that
$W$ contains a word of the form \eqref{eq:world}.
A \emph{tree-shaped $\Sfive^{n}_{\ALCOu}$-model}
takes the form $\mathfrak{M}=(\mathfrak{F},\Delta,\mathcal{I})$ with $\mathfrak{F}$
a tree-shaped $\Sfive^{n}$-frame.

\begin{lemma}\label{lem:tree-shapedS5}
	$\Sfive^{n}_{\ALCOu}$ is determined by tree-shaped $\Sfive^{n}_{\ALCOu}$-models\textup{:} an $\MLALCOu$-concept 
    $C$ is satisfiable in an $\Sfive^{n}_{\ALCOu}$-model
    iff it is satisfiable in the root of a tree-shaped $\Sfive^{n}_{\ALCOu}$-model of depth bounded by the modal depth of $C$.
\end{lemma}
\begin{proof}
	Assume that $C$ has modal depth $K$ and assume $\mathfrak{M}=(\mathfrak{F},\Delta,\mathcal{I})$ with $\mathfrak{F}=(W,R_{1},\ldots,R_{n})$ and $w_{0}\in W$ with $C^{\mathcal{I}_{w_{0}}}\not=\emptyset$ are given. Unfold $\mathfrak{F}$  into 
	$$
	\mathfrak{F}^{\ast}=(W^{\ast},R_{1}^{\ast},\ldots,R_{n}^{\ast}) 
	$$
	where $W^{\ast}$ is the set of words $\avec{w}$ of the form \eqref{eq:world}
	with $(w_{j},w_{j+1})\in R_{i_{j}}$, $w_{j}\not=w_{j+1}$,
	$i_{j}\not=i_{j+1}$, and $m\leq K$, and $R^{\ast}_{i}$ is the smallest  equivalence relation containing all pairs $(\avec{w},\avec{w}iw)\in W^{\ast}\times W^{\ast}$. Define a model $\mathfrak{M}^{\ast}=(\mathfrak{F}^{\ast},\Delta^{\ast},\mathcal{I}^{\ast})$
	by setting $\Delta^{\ast}=\Delta$ and $\mathcal{I}^{\ast}_{\avec{w}}=\mathcal{I}_{w_{m}}$ for $\avec{w}\in W^{\ast}$ of the form \eqref{eq:world}. Clearly $\mathfrak{M}^{\ast}$ is tree-shaped.
	Moreover, for all $\MLALCOu$-concepts $D$ of modal depth not exceeding $K$, 
	one can show by induction that $D^{\mathcal{I}^{\ast}_{w_{0}}}=
	D^{\mathcal{I}_{w_{0}}}$. In particular, $C^{\mathcal{I}^{\ast}_{w_{0}}}\ne\emptyset$.
\end{proof}
Given an
$\MLALCOu$
concept $C_{0}$, let
$\con{C_{0}}$ be the closure under single negation of the set of
concepts occurring in $C_{0}$.
A \emph{type} for $C_{0}$ is a subset $\contp$ of $\con{C_{0}}$
such that
\begin{enumerate}
  [label=\textbf{C\arabic*},leftmargin=*,series=run]
\item
  $\neg C \in \contp$ iff $C\not \in \contp$, for all
  $\neg C \in \con{C_{0}}$;
  \label{ct:neg}
\item
  $C \sqcap D \in \contp$ iff $C, D \in \contp$, for all
  $C \sqcap D \in \con{C_{0}}$.
  \label{ct:con}
\end{enumerate}
Note that there are at most $2^{|\con{C_{0}}|}$ types for $C_{0}$.

A \emph{quasistate for $C_{0}$} is a non-empty set $\settp$ of types
for $C_{0}$ satisfying the following conditions:
\begin{enumerate}
  [label=\textbf{Q\arabic*},leftmargin=*,series=run]
%
\item for every $\{a\} \in \con{C_{0}}$, there exists exactly one
  $\contp \in \settp$ such that $\{ a \} \in \contp$;
  \label{qs:nom}
\item for every $\contp \in \settp$ and every
  $\exists r.C \in \contp$, there exists $\contp' \in \settp$ such
  that $\{ \lnot D \mid \lnot \exists r. D \in \contp \} \cup \{ C \}
  \subseteq \contp'$;
  \label{qs:existsr}
\item for every $\contp \in \settp$, $\exists u.C \in \contp$ iff
  there exists $\contp' \in \settp$ such that $C \in \contp'$.
  \label{qs:univ}
\end{enumerate}

A \emph{basic structure for $C_{0}$} is a pair $(\mathfrak{F}, \funcand)$, where
$\mathfrak{F}=(W,R_{1},\ldots,R_{n})$ is a tree-shaped $\Sfive^{n}$-frame with root $w_{0}$ and $\funcand$ is a function associating with
every $w \in W$ a quasistate $\funcand(w)$ for $C_{0}$, satisfying
\begin{enumerate}
  [label=\textbf{B\arabic*},leftmargin=*,series=run]
  \item there exists a type $\contp \in
    \funcand(w_{0})$ such that $C_{0} \in \contp$.
    \label{b1}
\end{enumerate}  
A \emph{run through $(\mathfrak{F}, \funcand)$} is a function $\rho$ mapping each
world $w \in W$ into a type $\rho(w) \in \funcand(w)$ and satisfying
the following condition for every $\Diamond_{i} C \in \con{C_{0}}$:
\begin{enumerate}
[label=\textbf{R\arabic*},leftmargin=*,series=run]
\item $\Diamond_{i} C \in \rho(w)$ if there exists $v \in W$ such that
  $wR_{i}v$ and $C \in \rho(v)$.
  \label{rn:modal}
\item if $\Diamond_{i} C \in \rho(w)$ then there exists $v \in W$ such that
$wR_{i}v$ and $C \in \rho(v)$.
\label{rn:modal2}
\end{enumerate}

An \emph{$\Sfive^{n}_{\ALCOu}$ quasimodel for $C_{0}$} is a triple $\quasimod =
(\mathfrak{F}, \funcand, \runs)$, where $(\mathfrak{F}, \funcand)$ is a basic structure for
$C_{0}$ and $\runs$ is a set of runs through $(\mathfrak{M}, \funcand)$ such that
the following condition holds:
\begin{enumerate}
[label=\textbf{M\arabic*},leftmargin=*,series=run]
\item for every $w \in W$ and every $\contp \in \funcand(w)$, there
  exists $\rho \in \runs$ with $\rho(w) = \contp$;
  \label{run:exists}
\item for every $w\in W$ and every
  $\contp\in \funcand(w), \text{ with } \{a\}\in \contp$, there exists
  exactly one $\rho\in \runs$ such that $\{a\} \in \rho(w)$.
    \label{run:nominal}
\end{enumerate}

The next lemma provides a link between quasimodels and standard models.

\begin{lemma}\label{th-s5nsat}
An $\MLALCOu$-concept 
$C_{0}$ is satisfiable in the root of a tree-shaped $\Sfive^{n}_{\ALCOu}$-model based on a frame $\mathfrak{F}$ iff 
there exists a $\Sfive^{n}_{\ALCOu}$-quasimodel for $C_{0}$ based on the same frame $\mathfrak{F}$.
\end{lemma}
\ifappendix
\begin{proof}
$(\Rightarrow)$
Let $\Mmf = (\Fmf, \Delta, \Imc)$, with
  $\Fmf = (W,R_{1},\ldots,R_{n})$ be a tree-shaped $\Sfive^{n}_{\ALCOu}$-model of modal depth at most md$(C_{0})$ with $C^{\Imc_{w_{0}}}\ne\emptyset$ for the root $w_{0}$ of $\Fmf$. Let
  $\contp^{\Imc_{w}}(d) = \{ C \in \con{C_0} \mid d \in
  C^{\Imc_{w}}\}$, for every $d \in \Delta$ and $w \in W$. Clearly,
  $\contp^{\Imc_{w}}(d)$ is a concept type for $C_{0}$ since it 
  satisfies \eqref{ct:neg}--\eqref{ct:con}. We now define a
  triple $\Qmf = (\Fmf, \funcand, \runs)$, where
  \begin{itemize}
  \item $\funcand$ is the function from $W$ to the set of
    quasistates for $C_{0}$ defined by setting
    $\funcand(w) = \{ \contp^{\Imc_{w}}(d) \mid d \in \Delta\}$, for
    every $w\in W$;
  \item $\runs$ is the set of functions $\rho_{d}$ from $W$ to
    the set of types for $C_{0}$ defined by setting
    $\rho_{d}(w) = \contp^{\Imc_{w}}(d)$, for every
    $d \in \Delta$ and $w\in W$.
   \end{itemize}
   It is easy to show that $\Qmf$ is a quasimodel for $C_{0}$. Indeed,
   $\funcand$ is well-defined, as $\funcand(w)$ is a set of types for
   $C_{0}$ satisfying~\eqref{qs:nom}--\eqref{qs:univ}, for every $w \in
   W$. Moreover, $(\Fmf, \funcand)$ is a basic structure for $C_{0}$ since
   $\Mmf$ satisfies $C_{0}$ in the root $w_{0}$ of $\Fmf$ and thus $(\Fmf, \funcand)$
   satisfies~\eqref{b1}. The set of runs, $\runs$, by construction,
   satisfies~\eqref{rn:modal} and~\eqref{rn:modal2}. Finally, by definition of $\funcand(w)$ and of~$\rho$, $\Qmf$ satisfies~\eqref{run:exists}, and, since for
   every $w\in W$, if $\{a\}\in \contp$, there is exactly
   one $d\in\Delta$ such that $d=\{a\}^{\Imc_{w}}$, thus $\Qmf$
   satisfies~\eqref{run:nominal}.

   \noindent $(\Leftarrow)$ Suppose there is a quasimodel
   $\Qmf = (\mathfrak{F}, \funcand, \runs)$ for $C_{0}$. Define a model
   $\Mmf = (\Fmf, \Delta, \Imc)$, by setting $\Delta = \runs$, and, for any $A\in\NC$, $r\in\NR$ and $a\in\NI$:
   \begin{itemize}
   \item $A^{\Imc_{w}} = \{ \rho \in \Delta \mid A \in \rho(w) \}$;
   \item
     $r^{\Imc_{w}} = \{ (\rho, \rho') \in \Delta \times \Delta \mid \{
     \lnot C \mid \lnot \exists r . C \in \rho(w) \} \subseteq
     \rho'(w) \}$;
   \item $u^{\Imc_{w}} = \Delta \times \Delta$;
   \item $a^{\Imc_{w}} = \rho$, for the unique $\rho \in \runs$ such that $\{ a \} \in \rho(w)$.
   \end{itemize}
   Observe that $\Delta$ is well-defined since $W$ is a non-empty set
   and $\funcand(w)\neq\emptyset$, for all $w\in W$. Thus,
   by~\eqref{run:exists}, $\runs\neq\emptyset$. Also,
   by~\eqref{qs:nom} and~\eqref{run:nominal}, $a^{\Imc_{w}}$ is
   well-defined. We now require the following claim.
   \begin{claim}\label{cla:qmcon:s5}
     For every $C \in \con{C_{0}}$, $w \in W$ and $\rho \in \Delta$, we have
     $\rho \in C^{\Imc_{w}}$ iff $C \in \rho(w)$.
   \end{claim}
   \begin{proof}
     The proof is by induction on $C$. The base cases, $C = A$ and $C =
     \{a\}$, follow immediately from the definition of $\Mmf$. We then
     consider the inductive cases.\\
     Let $C = \lnot D$. $\lnot D\in \rho(w)$ iff, by~\eqref{ct:neg},
     $D\not\in \rho(w)$. By induction, $D\not\in \rho(w)$ iff
     $\rho\not\in D^{\Imc_{w}}$ iff $\rho\in(\lnot D)^{\Imc_{w}}$.\\
     Let $C = D\sqcap E$. Similar to the previous case, now by
     using~\eqref{ct:con}.\\
     Let $C = \exists u.D$. $\rho \in (\exists u.D)^{\Imc_{w}}$ iff
     there exists $\rho' \in D^{\Imc_{w}}$. By inductive hypothesis,
     $\rho' \in D^{\Imc_{w}}$ iff $D \in \rho'(w)$. By~\eqref{qs:univ},
     the previous step holds iff $\exists u.D \in \rho(w)$.\\
     Let $C = \exists r.D$. $(\Rightarrow)$ Suppose that
     $\rho\in (\exists r.D)^{\Imc_{w}}$. Then, there exists
     $\rho' \in \Delta$ such that $(\rho, \rho') \in r^{\Imc_{w}}$ and
     $\rho' \in D^{\Imc_{w}}$. By inductive hypothesis,
     $\rho' \in D^{\Imc_{w}}$ iff $D\in\rho'(w)$. By contradiction,
     assume that $\exists r.D\not\in\rho(w)$, then, by~\eqref{ct:neg},
     $\lnot \exists r.D\in\rho(w)$. By definition of $r^{\Imc_{w}}$,
     since $(\rho,\rho')\in r^{\Imc_{w}} $, then,
     $\lnot D\in \rho'(w)$, thus contradicting, by~\eqref{ct:neg}, that
     $D\in \rho'(w)$.
     $(\Leftarrow)$ Conversely, suppose that
     $\exists r.D \in \rho(w)$. By~\eqref{qs:existsr}
     and~\eqref{run:exists}, there exists a $\rho' \in \Delta$ such that
     $\{ \lnot E \mid \lnot \exists r . E \in \rho(w) \} \cup \{ D \}
     \subseteq \rho'(w)$. By inductive hypothesis and the definition
     of $r^{\Imc_{w}}$, $\rho' \in  D^{\Imc_{w}}$ and $(\rho, \rho')
     \in r^{\Imc_{w}}$. Thus, $\rho \in (\exists r.D)^{\Imc_{w}}$.\\
     Let $C = \Diamond_{i} D$.  $\rho \in (\Diamond_{i} D)^{\Imc_{w}}$ iff
     there exists $v \in W$ such that $(w,v)\in R_{i}$ and $\rho\in D ^{\Imc_{v}}$.  By
     inductive hypothesis, $\rho \in D^{\Imc_{v}}$ iff $D \in \rho(v)$. Hence, 
     by~\eqref{rn:modal} and~\eqref{rn:modal2}, such a $v$ exists iff $\Diamond_{i} D\in\rho(w)$.
\end{proof}
Now we can easily finish the proof of Lemma~\ref{th-s5nsat} by observing
that, by~\eqref{b1}, there exists a world $w'\in W$ and a type
$\contp \in \funcand(w')$ such that $C_{0}\in \contp$. Thus,
by~\eqref{run:exists}, there exists $\rho \in \runs$ with
$\rho(w') = \contp$ and, by Claim~\ref{cla:qmcon:s5},
$\rho\in C_{0}^{\Imc_{w'}}$.
\end{proof}

We now show the exponential finite model property in terms of 
quasimodels.
%

\begin{lemma}\label{lem:s5exp}
  There exists an $\Sfive^{n}_{\ALCOu}$ quasimodel for $C_{0}$ iff there exists
  an $\Sfive^{n}_{\ALCOu}$ quasimodel for $C_{0}$ of exponential size in the
  length of $C_{0}$.
\end{lemma}
\begin{proof}
Assume $\Qmf = (\mathfrak{F},\funcand,\runs)$ with $\mathfrak{F}=(W,R_{1},\ldots,R_{n})$ is given such that $W$ is a prefix-closed set of words of the form \eqref{eq:world}, where $1\leq i_{j}\leq m$, $i_{j}\not=i_{j+1}$, and $m$ is bounded by md$(C_{0})$.

	For each $\avec{w}\in W$ and $\contp\in \funcand(\avec{w})$ we fix a \emph{proto-run, $\rho_{\avec{w},\contp}\in \runs$, for $\avec{w}$ and $\contp$} with $\rho_{\avec{w},\contp}(\avec{w})=\contp$. Let $\mathcal{F}(\avec{w})$ denote the set of selected proto-runs through $\funcand(\avec{w})$, that is $\mathcal{F}(\avec{w})=\{\rho_{\avec{w},\contp} \mid \contp\in \funcand(\avec{w})\}$.
	We also fix for any $\Diamond_{i}C\in \contp\in \funcand(\avec{w})$ a \emph{proto-witness} $\avec{w}'=v_{\avec{w},i,C,\contp}\in W$ such that
	$\avec{w}'$ takes the form $\avec{w}iv$ (and so $(\avec{w},\avec{w}')\in R_{i}$) and $C\in \rho_{\avec{w},\contp}(\avec{w}')$,
	whenever such a $\avec{w}'$ exists. Note that such a $\avec{w}'$ might not exist since by the definition of $R_{i}$ as equivalence closure in tree-shaped models the witness for $\Diamond_{i}C$ might also be $\contp$ itself if $C\in \contp$ or $\avec{w}$ could take the form $\avec{w}_{0}iw$ for some $\avec{w}_{0}$ and either $C\in \rho_{\avec{w},\contp}(\avec{w}_{0})$ or there is $w'$ with $\avec{w}_{0}iw'\in W$ and $C\in\rho_{\avec{w},\contp}(\avec{w}_{0}iw')$.
	
    We next define inductively sets $W_{0},W_{1},\ldots,W_{K}\subseteq W$ with $K$ the depth of $\mathfrak{F}$ by setting $W_{0}=\{w_{0}\}$ and $W_{j+1}=W_{j+1}^{1} \cup \cdots \cup W_{j+1}^{n}$, where for $1\leq i \leq n$:
    $$
    W_{j+1}^{i}= \{v_{\avec{w},i,C,\contp} \mid  \Diamond_{i}C\in \contp\in \funcand(\avec{w}), \avec{w}\in W_{j}\}.
    $$
    Let $W'=W_{0}\cup \cdots \cup W_{K}$ and let $R_{1}',\ldots,R_{n}'$ and $\funcand'$ be the restrictions of $R_{1},\ldots,R_{n}$ and $\funcand$ to $W'$:
    $$
    R_{i}'= R_{i}\cap (W'\times W'), \quad \funcand'(\avec{w})=\funcand(\avec{w}) \text{ for $\avec{w}\in W'$}
    $$  
    Also let $\Rmf'$ denote the set of restrictions of the proto-runs $\rho_{\avec{w},\contp}$ with $\avec{w}\in W'$ to $W'$. Let $\mathfrak{F}'=
    (W',R_{1}',\ldots,R_{n}')$ and $\mathfrak{Q}'=(\mathfrak{F}',\funcand',\Rmf')$. Then 
    $\mathfrak{Q}'$ is a quasimodel except that the proto-runs in $\Rmf'$ might not all satisfy~(\ref{rn:modal2}). Note that~(\ref{rn:modal2})
    holds for every proto-run $\rho_{\avec{w},\contp}$ at the world
    $\avec{w}$ for those $\Diamond_{i}C$ for which $i_{m-1}\ne i$ (assuming that $\avec{w}$ takes the form~\eqref{eq:world}) but that there might not be witnesses for $\Diamond_{i}C$ on $\avec{w}'\ne\avec{w}$ and if $i_{m-1}=i$. To ensure~(\ref{rn:modal2}) also in these cases we introduce, in a careful way, copies of existing worlds that provide the witnesses for such $\Diamond_{i}C$. 
    
    Denote the set of all bijections on $\Rmf'$ by $B(\Rmf')$. 
    The new domain is constructed using copies of the original domain elements indexed by bijections in $B(\Rmf')$. 
    
    Take for any $\avec{w}\in W'$ and $\rho\in \Rmf'$ the bijection
    $\sigma_{\avec{w},\rho}\in B(\Rmf')$ that swaps $\rho$ with
    the proto-run $\rho_{\avec{w},\contp}$ with $\contp=\rho(\avec{w})$ and maps the remaining runs to themselves. 
    The \emph{$\avec{w}$-repair set}, $\textrm{Rep}(\avec{w})$, is the set of all $\sigma_{\avec{w},\rho}$ with $\rho \in \Rmf'$. Note that the identity function, $id$, on $\Rmf'$ is an element of $\textrm{Rep}(\avec{w})$.
    
    For any $(\avec{w},\sigma)\in W'\times B(\Rmf')$, define 
    \begin{multline*}
    Suc_{i}(\avec{w},\sigma) = \{ (\avec{w}',\sigma') \in W'\times B \mid
    \avec{w}'=\avec{w}iw\in W', \\
     \sigma'=\tau \circ \sigma \text{ for some } \tau \in \textrm{Rep}(\avec{w})\}
    \end{multline*}
    and let $W''$ denote the set of all words 
    \begin{equation}\label{eqveru}
    \avec{u}=(\avec{w}_{0},\sigma_{0})i_{0}(\avec{w}_{1},\sigma_{1})i_{1} \cdots i_{m-1}(\avec{w}_{m},\sigma_{m}),
    \end{equation}
    where $(\avec{w}_{0},\sigma_{0})= (w_{0},id)$ and $(\avec{w}_{i+1},\sigma_{i+1})\in Suc_{i}(\avec{w}_{i},\sigma_{i})$ for all $0\leq i<m$. 
    
    Define $\funcand''$ by setting $\funcand''(\avec{u})= \funcand'(\avec{w}_{m})$ for any $\avec{u}\in W''$ of the form \eqref{eqveru}.
    Associate with any run $\rho\in \Rmf'$ a unique run $\rho^{\uparrow}$ on $W''$ by setting $\rho^{\uparrow}(\avec{u}) := \sigma_{m}(\rho)(\avec{w}_{m})$ for any $\avec{u}\in W''$ of the form~\eqref{eqveru} and let $\Rmf''= \{ \rho^{\uparrow} \mid \rho\in \Rmf'\}$.
         
We show that $(\mathfrak{F}'',\funcand'', \Rmf'')$ with $\mathfrak{F}''=(W'',R_{1}'',\ldots,R_{n}'')$ the tree-shaped  $\Sfive^{n}_{\ALCOu}$ frame defined by $W''$, is as required.

We show that conditions \eqref{b1}, \eqref{rn:modal}, \eqref{rn:modal2}, \eqref{run:exists}, and \eqref{run:nominal} are satisfied.

Observe first that $\rho \mapsto \rho^{\uparrow}$ is a bijection from $\Rmf'$ onto~$\Rmf''$. To show this observe that the mapping is surjective by definition. To show that it is injective, assume that  $\rho_{1}\not=\rho_{2}$. Take $\avec{w}_{m}\in W'$ of the form~\eqref{eq:world} with $\rho_{1}(\avec{w}_{m})\ne\rho_{2}(\avec{w}_{m})$. Then let
$$
\avec{v}=(\avec{w}_{0},id)i_{0}(\avec{w}_{1},id)i_{1} \cdots i_{m-1}(\avec{w}_{m},id)\in W'',
$$
where $\avec{w}_{i+1}=\avec{w}_{i}iw_{i+1}$ for $i<m$. We have $\rho_{1}^{\uparrow}(\avec{v})=id(\rho_{1})(\avec{w}_{m})\ne id(\rho_{2})(\avec{w}_{m})$, as required.

\eqref{b1} is straightforward.

We prove \eqref{run:exists} and \eqref{run:nominal} in one go by showing that, for any 
$\avec{u}\in W''$ of the form~\eqref{eqveru} and $\contp\in \funcand''(\avec{u})$, there is a bijection between
$V_{0}=\{\rho\in \Rmf''\mid \rho(\avec{u})=\contp\}$ and 
$V_{1}=\{\rho\in \Rmf'\mid \rho (\avec{w}_{m})=\contp\}$ (recall that 
$\funcand''(\avec{u}) = \funcand'(\avec{w}_{m})$).

Indeed, we claim that $F\colon V_{0}\rightarrow V_{1}$ defined by setting $F(\rho^{\uparrow})=\sigma_{m}(\rho)$ is a bijection. As $\rho\mapsto \rho^{\uparrow}$ is bijective and $\rho^{\uparrow}(\avec{u})=\sigma_{m}(\rho)(\avec{w}_{m})$, $F$ is a well defined function into $V_{1}$. It is injective since $\sigma_{m}$ is also bijective, by definition. It is surjective since for any $\rho\in \Rmf'$ with $\rho(\avec{w}_{m})=\contp$ we have $F((\sigma_{m}^{-1}(\rho))^{\uparrow})=\sigma_{m}\sigma_{m}^{-}(\rho)=\rho$
and $\sigma_{m}^{-1}(\rho)\in \Rmf'$.

We next show \eqref{rn:modal} and \eqref{rn:modal2}. Assume that $\Diamond_{i} C \in \con{C_{0}}$ and $\rho^{\uparrow}\in \Rmf''$.
Consider $\avec{u}\in W''$ of the form~\eqref{eqveru}.

For~\eqref{rn:modal}, assume there exists $\avec{v} \in W''$ such that $\avec{u}R_{i}''\avec{v}$ and $C \in \rho^{\uparrow}(\avec{v})$.
We have to show that $\Diamond_{i}C\in \rho^{\uparrow}(\avec{u})$.
Observe that the runs in $\Rmf'$ satisfy~\eqref{rn:modal} in $\Qmf'$ and that, moreover, 
\begin{enumerate}
[label=\textbf{R\arabic*},leftmargin=*,series=run]\addtocounter{enumi}{2}
\item if $\Diamond_{i}C \in \rho(\avec{w})$ with $\rho\in \Rmf'$ and $(\avec{w},\avec{w}')\in R_{i}'$, then $\Diamond_{i}C\in \rho(\avec{w}')$.
\end{enumerate}
%
We distinguish the following four cases.
\begin{enumerate}
	\item $\avec{v}=\avec{u}$. Then $\Diamond_{i}C\in \rho^{\uparrow}(\avec{u})$ since we have $\Diamond_{i}C\in \contp$ whenever $C\in \contp$ for $\Diamond_{i} C \in \con{C_{0}}$ for any $\contp$.
	\item $\avec{v}=\avec{u}i(\avec{u}i w_{m+1},\sigma_{m+1})$ for some $(w_{m+1},\sigma_{m+1})$ with $\avec{w}_{m}iw_{m+1}\in W'$ with $\sigma_{m+1}=\tau\circ \sigma_{n}$ and $\tau \in \textrm{Rep}(\avec{w}_{m})$.
	
	We have $\rho^{\uparrow}(\avec{u})=\sigma_{m}(\rho)(\avec{w}_{m})$
	and $\rho^{\uparrow}(\avec{v})=\tau\circ \sigma_{m}(\rho)(\avec{w}_{m}iw_{m+1})$. From $C\in \tau\circ \sigma_{m}(\rho)(\avec{w}_{m}iw_{m+1})$, it follows that $\Diamond_{i}C\in \tau\circ \sigma_{m}(\rho)(\avec{w}_{m} )$ since $\tau\circ \sigma_{m}(\rho)$ satisfies condition~\eqref{rn:modal} for runs in $\Qmf'$. Now, $\tau\circ \sigma_{m}(\rho)$ coincides with $\sigma_{m}(\rho)$ on $\avec{w}_{m}$. Hence, $\Diamond_{i}C\in \sigma_{m}(\rho)(\avec{w}_{m})$, as required.
	 
	\item $\avec{v}=(\avec{w}_{0},\sigma_{0})i_{0}(\avec{w}_{1},\sigma_{1})i_{1} \cdots i_{m-2}(\avec{w}_{m-1},\sigma_{m-1})$ and $i_{m-1}=i$.
	
	We have $\sigma_{m}=\tau'\circ \sigma_{m-1}$ for some $\tau' \in \textrm{Rep}(\avec{w}_{m-1})$.
	Note that  $\rho^{\uparrow}(\avec{u})=\tau'\circ\sigma_{m-1}(\rho)(\avec{w}_{m-1}iw_{m})$
	and $\rho^{\uparrow}(\avec{v})=\sigma_{m-1}(\rho)(\avec{w}_{m-1})$.
	Now $\tau'\circ \sigma_{m-1}(\rho)$ coincides with $\sigma_{m-1}(\rho)$ on $\avec{w}_{m-1}$. Hence, $C\in \tau'\circ\sigma_{m-1}(\rho)(\avec{w}_{m-1})$.
	Thus,  as required, $\Diamond_{i}C\in \tau'\circ\sigma_{m-1}(\rho)(\avec{w}_{m-1}iw_{m})$ since
	$\tau'\circ \sigma_{m-1}$ satisfies condition~\eqref{rn:modal} for runs in $\Qmf'$.
	 
	\item for $\avec{v}'=(\avec{w}_{0},\sigma_{0})i_{0}(\avec{w}_{1},\sigma_{1})i_{1} \cdots i_{m-2}(\avec{w}_{m-1},\sigma_{m-1})$ and $i_{m-1}=i$ we have 
	$\avec{v}=\avec{v}'i(\avec{w}_{m-1}i w_{m}',\sigma_{m}')$ for some $(w_{m}',\sigma_{m}')$ with $\avec{w}_{m-1}i w_{m}'\in W'$ and $\sigma_{m}'=\tau\circ \sigma_{m-1}$ for some $\tau \in \textrm{Rep}(\avec{w}_{m-1})$.	
	
	We can then first show that $\Diamond_{i}C\in \rho^{\uparrow}(\avec{v}')$ in exactly the same way as in the proof of Point~2. Now we can argue in the same way as in the proof of Point~3 and using~(\textbf{R3}) that $\Diamond_{i}C\in \rho^{\uparrow}(\avec{u})$.  
\end{enumerate}

For~\eqref{rn:modal2}, assume $\Diamond_{i}C \in \rho^{\uparrow}(\avec{u})$. We have to show that there exists $\avec{v} \in W''$ such that $\avec{u}R_{i}''\avec{v}$ and $C \in \rho^{\uparrow}(\avec{v})$.
We have $\rho^{\uparrow}(\avec{u})= \sigma_{m}(\rho)(\avec{w}_{m})$ and distinguish the following three cases.
\begin{enumerate}
	\item $C\in \sigma_{m}(\rho)(\avec{w}_{m})$. Then we are done since, for $\avec{v}=\avec{u}$, we have $C\in \rho^{\uparrow}(\avec{v})$ and $\avec{u}R_{i}''\avec{v}$.
	\item For $\contp=\sigma_{m}(\rho)(\avec{w}_{m})$ we have for the proto-run $\rho_{\avec{w}_{m},\contp}\in \Rmf$ that there is a proto-witness 
	$\avec{w}_{m+1}=v_{\avec{w}_{m},i,C,\contp}\in W'$ of the form $\avec{w}_{m}iw_{m+1}$. Then $C\in \rho_{\avec{w}_{m},\contp}(\avec{w}_{m+1})$. 
	
	Then let $\avec{v}= \avec{u}i(\avec{w}_{m+1},\tau\circ\sigma_{m})$, where
    $\tau$ swaps $\sigma_{m}(\rho)$ with $\rho_{\avec{w}_{m},\contp}$. We have
    $\avec{u}R_{i}''\avec{v}$ and $\rho^{\uparrow}(\avec{v})=\tau\circ\sigma_{m}(\rho)(\avec{w}_{m+1})=\rho_{\avec{w}_{m},\contp}(\avec{w}_{m+1})$. So, $C \in\rho^{\uparrow}(\avec{v})$, as required.
    \item $i_{m-1}=i$ and $\Diamond_{i} C \in \sigma_{m-1}(\rho)(\avec{w}_{m-1})$. 
    
    If $C\in \sigma_{m-1}(\rho)(\avec{w}_{m-1})$, then $C \in \rho^{\uparrow}(\avec{v})$
    for $\avec{v}=(\avec{w}_{0},\sigma_{0})i_{0}(\avec{w}_{1},\sigma_{1})i_{1} \cdots i_{m-2}(\avec{w}_{m-1},\sigma_{m-1})$, and we are done.
    
    Otherwise, we can proceed with $\sigma_{m-1}(\rho)(\avec{w}_{m-1})$
    in the same way as we did with $\sigma_{m}(\rho)(\avec{w}_{m})$
    in Point~2 and find for $\avec{v}'=(\avec{w}_{0},\sigma_{0})i_{0}(\avec{w}_{1},\sigma_{1})i_{1} \cdots i_{m-2}(\avec{w}_{m-1},\sigma_{m-1})$ a
    $\avec{v}=\avec{v}'i(\avec{w}_{m-1}i w_{m}',\sigma_{m}')$ for some $(w_{m}',\sigma_{m}')$ with $\avec{w}_{m-1}i w_{m}'\in W'$ and $\sigma_{m}'=\tau\circ \sigma_{m-1}$ for some $\tau \in \textrm{Rep}(\avec{w}_{m-1})$ such that $C \in \rho^{\uparrow}(\avec{v})$.
\end{enumerate}
Clearly, the quasimodel $(\mathfrak{F}'',\funcand'',\runs'')$ is of at most exponential size, as required. 
\end{proof}
The exponential finite model property now follows from 
Lemma~\ref{lem:s5exp} using Lemma~\ref{th-s5nsat}.
\end{proof}

\subsection{Proof of Theorem~\ref{thm:moddickson}}

\moddickson*	

As before, it suffices to consider total concept satisfiability under ontologies without definite descriptions. For $\Kfn_{\ALCOu}$, the decidability proof for concept satisfiability without ontologies is easily extended to a proof with ontologies, so we consider concept satisfiability without ontology. Call a finite frame $\mathfrak{F}=(W,R_{1},\ldots,R_{n},R)$ a \emph{finite tree-shaped} frame if $R$ is the transitive closure of $R_{1}\cup \cdots \cup R_{n}$ and there exists a $w_{0}\in W$ such that the domain $W$ of $\mathfrak{F}$ is a prefix closed set of words of the form 
	\begin{equation}\label{eq:worldtrans}
		\avec{w}=w_{0}i_{0}w_{1}\cdots i_{n-1}w_{m}
	\end{equation}
	where $1\leq i_{j}\leq n$, and 
	$$
	R_{i} = \{ (\avec{w},\avec{w}iw)\mid \avec{w}iw\in W\}.
	$$
	We call $w_{0}$ the root of $\mathfrak{F}$. A \emph{finite tree-shaped model}
	takes the form $\mathfrak{M}=(\mathfrak{F},\Delta,\mathcal{I})$ with $\mathfrak{F}$
	a finite tree-shaped frame. A straightforward unfolding argument shows that any $\MLALCOu$-concept that is satisfiable in a $\Kfn_{\ALCOu}$-model is also satisfiable in a finite tree-shaped model. So we start with finite tree-shaped models.  
	
	Given an $\MLALCOu$-concept $C_{0}$, let
	$\con{C_{0}}$ be the closure under single negation of the set of
	concepts occurring in $C_{0}$.
	A \emph{type} for $C_{0}$ is a subset $\contp$ of $\con{C_{0}}$
	such that
	\begin{enumerate}
		[label=\textbf{C\arabic*},leftmargin=*,series=run]
		\item
		$\neg C \in \contp$ iff $C\not \in \contp$, for all
		$\neg C \in \con{C_{0}}$;
		\label{ct:neg}
		\item
		$C \sqcap D \in \contp$ iff $C, D \in \contp$, for all
		$C \sqcap D \in \con{C_{0}}$.
		\label{ct:con}
	\end{enumerate}
	Note that there are at most $2^{|\con{C_{0}}|}$ types for $C_{0}$.
	We fix an ordering $\contp_{1},\ldots,\contp_{k}$ of the set of types for $C_{0}$. We use vectors $\avec{x}=(x_{1},\ldots,x_{k})\in \extN^{k}$ to represent that the type~$\contp_{i}$ is satisfied by $x_{i}$-many elements in a world $w$, for $1\leq i \leq n$. Let $|\avec{x}|=\sum_{i=1}^{k}x_{i}$. A \emph{quasistate} for $C_{0}$ is a vector $\avec{x}=(x_{1},\ldots,x_{k}) \in \extN^{k}$ with $|\avec{x}|>0$   
	satisfying the following conditions:
	\begin{enumerate}
		[label=\textbf{Q\arabic*},leftmargin=*,series=run]
		%
		\item for every $\{a\} \in \con{C_{0}}$, 
		$$
		\sum_{\{a\}\in \contp_{i}} x_{i}=1;
		$$  %
		\item for every $x_{i}>0$ and every
		$\exists r.C \in \contp_{i}$, there exists $x_{j}>0$ such
		that $\{ \lnot D \mid \lnot \exists r. D \in \contp_{i} \} \cup \{ C \}
		\subseteq \contp_{j}$;
		\item for every $x_{i}>0$ and $\exists u.C \in \contp_{i}$ iff
		there exists $x_{j}>0$ such that $C \in \contp_{j}$.
		\label{qs:univ}
	\end{enumerate}
	
	A \emph{basic structure for $C_{0}$} is a pair $(\mathfrak{F}, \funcand)$, where
	$\mathfrak{F}=(W,R_{1},\ldots,R_{n})$ is a finite tree-shaped frame with root $w_{0}$ and $\funcand$ is a function associating with
	every $w \in W$ a quasistate $\funcand(w)$ for $C_{0}$, satisfying
	\begin{enumerate}
		[label=\textbf{B\arabic*},leftmargin=*,series=run]
		\item there exists $i\leq k$ with $\funcand(w_{0})_{i}>0$ such that $C_{0} \in \contp_{i}$.
		\label{b1}
	\end{enumerate}  
	Call a subset $V$ of $W$ \emph{$R_{i}$-closed} if $v\in V$ whenever $w\in V$
	and $wR_{i}v$. Call $V$ \emph{$R$-closed} if $V$ is $R_{i}$-closed for $1\leq i \leq n$.
	A \emph{run through $(\mathfrak{F}, \funcand)$} is a partial function $\rho$ mapping worlds $w \in W$ to natural numbers $\{1,\dots, k\}$ with $\funcand(w)_{\rho(w)}>0$ such that the domain of $\rho$ is $R$-closed and satisfies
	the following condition for every $\Diamond_{i} C \in \con{C_{0}}$:
	\begin{enumerate}
		[label=\textbf{R\arabic*},leftmargin=*,series=run]
		\item $\Diamond_{i} C \in \contp_{\rho(w)}$ if there exists $v \in W$ such that
		$wR_{i}v$ and $C \in \contp_{\rho(v)}$;
		\label{rn:modal}
		\item if $\Diamond_{i} C \in \contp_{\rho(w)}$ then there exists $v \in W$ such that
		$wR_{i}v$ and $C \in \contp_{\rho(v)}$.
		\label{rn:modal2}
	\end{enumerate}
	
	A \emph{quasimodel for $C_{0}$} is a triple $\quasimod =
	(\mathfrak{F}, \funcand, \runs)$, where $(\mathfrak{F}, \funcand)$ is a basic structure for
	$C_{0}$, and $\runs$ is a set of runs through $(\mathfrak{M}, \funcand)$ such that
	the following condition holds:
	\begin{enumerate}
		[label=\textbf{M\arabic*},leftmargin=*,series=run]
		\item\label{run:exists} for every $w \in W$ and every $i\leq k$,
		$$
		|\{\rho\in \runs \mid \rho(w)=i\}| = \funcand(w)_{i}.
		$$
	\end{enumerate}
	The following lemma provides a link between quasimodels for $C_{0}$ and models satisfying $C_{0}$.
	\begin{lemma}\label{th-GLnsat}
		An $\MLALCOu$-concept 
		$C_{0}$ is satisfiable in the root of a finite tree-shaped model iff there exists a quasimodel for $C_{0}$ based on the same frame. 
	\end{lemma}
	\ifappendix
	\begin{proof}
		$(\Rightarrow)$
		Let $\Mmf = (\Fmf, \Delta, \Imc)$, with
		$\Fmf = (W,R_{1},\ldots,R_{n})$ be a finite tree-shaped model  with $C^{\mathcal{I}_{w_{0}}}\not=\emptyset$ for the root $w_{0}$ of $\mathfrak{F}$. Let
		$\contp^{\Imc_{w}}(d) = \{ C \in \con{C_0} \mid d \in
		C^{\Imc_{w}}\}$, for every $d \in \Delta$ and $w \in W$. Clearly,
		$\contp^{\Imc_{w}}(d)$ is a type for $C_{0}$ since it 
		satisfies \eqref{ct:neg}--\eqref{ct:con}. We now define a
		triple $\Qmf = (\mathfrak{F}, \funcand, \runs)$, where
		\begin{itemize}
			\item $\funcand$ is the function from $W$ to the set of
			quasistates for $C_{0}$ defined by setting
			$\funcand(w) = (x_{1},\ldots,x_{k})\in\extN^k$ with 
			$$
			x_{i}=|\{d \in \Delta^{w}\mid \contp^{\Imc_{w}}(d)=\contp_{i}\}|
			$$
			for $1\leq i \leq k$ and every $w\in W$;
			\item $\runs$ is the set of functions $\rho_{d}$ from $W$ to
			$\{1,\ldots,k\}$ defined by setting
			$\rho_{d}(w) = i$ if $\contp^{\Imc_{w}}(d)=\contp_{i}$,
			for every $d \in \Delta$ and $w\in W$.
		\end{itemize}
		It is easy to show that $\Qmf$ is a quasimodel for $C_{0}$. 
		
		\medskip
		
		\noindent $(\Leftarrow)$ Suppose there is a quasimodel
		$\Qmf = (\mathfrak{F}, \funcand, \runs)$ for $C_{0}$. Define a model
		$\Mmf = (\Fmf, \Delta, \Imc)$, by setting $\Delta^{w} = \{ \rho \in \runs\mid w\in \text{dom}(\rho)\}$, and, for any $A\in\NC$, $r\in\NR$ and $a\in\NI$:
		\begin{itemize}
			\item $A^{\Imc_{w}} = \{ \rho \in \Delta^{w} \mid A \in \rho(w) \}$;
			\item
			$r^{\Imc_{w}} = \{ (\rho, \rho') \in \Delta^{w} \times \Delta^{w} \mid \{
			\lnot C \mid \lnot \exists r . C \in \rho(w) \} \subseteq
			\rho'(w) \}$;
			\item $u^{\Imc_{w}} = \Delta^{w} \times \Delta^{w}$;
			\item $a^{\Imc_{w}} = \rho$, for the unique $\rho \in \Delta^{w}$ such that $\{ a \} \in \rho(w)$.
		\end{itemize}
		Observe that $\Delta$ is well-defined and expanding since $|\avec{x}|>0$ for every quasistate $\avec{x}$ and since the domain of each $\rho \in \runs$ is $R$-closed. The following claim is now straightforward from the definition of quasimodels for $C_{0}$. 
		\begin{claim}\label{cla:qmcon}
			For every $C \in \con{C_{0}}$, $w \in W$ and $\rho \in \Delta$,
			$\rho \in C^{\Imc_{w}}$ iff $C \in \contp_{\rho(w)}$.
		\end{claim}
			Now we can easily finish the proof of Lemma~\ref{th-GLnsat} by observing
			that, by~\eqref{b1}, there exists $i\leq k$ with $\funcand(w_{0})_{i}>0$ such that $C_{0} \in \contp_{i}$. Then there exists a $\rho\in \runs$ with $\rho(w_{0})=i$. Hence, by  Claim~\ref{cla:qmcon}, 
			$\rho\in C_{0}^{\Imc_{w_{0}}}$.
		\end{proof}
		The \emph{size} of a quasimodel $\Qmf=(\mathfrak{F},\funcand,\runs)$ with $\mathfrak{F}=(W,R_{1},\ldots,R_{n})$ is defined as 
		$$
		|W|\times \max\{ |\funcand(w)|\mid w\in W\}.
		$$
		Our next aim is to show that there exists a quasimodel for $C_{0}$ of size bounded by a recursive function in $|C_{0}|$ (in particular with $|\funcand(w)|<\infty$ for all $w$). We are going to apply Dickson's Lemma. For some $k>0$,
		let $(\mathbb{N}^{k},\leq)$ be the set of $k$-tuples 
		of natural numbers ordered by the natural product ordering: for $\avec{x}=(x_{1},\ldots,x_{k})$ and $\avec{y}=(y_{1},\ldots,y_{k})$ we set $\avec{x}\leq \avec{y}$ if $x_{i}\leq y_{i}$ for $1\leq i \leq k$.
		A pair $\avec{x},\avec{y}$ with $\avec{x} \leq \avec{y}$ is called an \emph{increasing pair}. Dickson's Lemma states every infinite sequence $\avec{x}_{1},\avec{x}_{2},\ldots\in \mathbb{N}^{k}$ contains an
		increasing pair $\avec{x}_{i_{1}},\avec{x}_{i_{2}}$ with $i_{1}<i_{2}$.
		In fact, assuming $|\avec{x}_{i}|\leq |\avec{x}_{i+1}|$ for all $i\geq 0$ and
		given recursive bounds on $|\avec{x}_{1}|$ and
		$|\avec{x}_{i+1}|-|\avec{x}_{i}|$ one can compute a recursive
		bound on the length of the longest sequence without any increasing pair~\cite{DBLP:conf/lics/FigueiraFSS11}.
		
		Assume a finite tree-shaped frame $\mathfrak{F}=(W,R_{1},\ldots,R_{n})$ 
		with root $w_{0}$ and worlds of the form~\eqref{eq:worldtrans} is given. Then we say that $\avec{w}$ is the \emph{predecessor} of $\avec{w}iw\in W$ in $\mathfrak{F}$ and that $\avec{w}iw\in W$ is a \emph{successor} of $\avec{w}$ in $\mathfrak{F}$. Worlds reachable along a path of predecessors from a world are called \emph{ancestors} and worlds reachable along a path of successors from a world are called \emph{descendants}. We set $\avec{w}<\avec{v}$ if $\avec{w}$ is a predecessor of $\avec{v}$ and denote $W_{\avec{w}}=\{\avec{w}\}\cup \{\avec{v} \in W\mid \avec{w}<\avec{v}\}$.
		
		We now show that one can reduce quasimodels for $C_{0}$ to quasimodels 
		for $C_{0}$ not having increasing pairs $\funcand(\avec{w}), \funcand(\avec{w}')$ with $\avec{w}<\avec{w}'$ and satisfying recursive bounds for $|\funcand(w_{0})|$,  $|\funcand(\avec{w}iw)|-|\funcand(\avec{w})|$ and the outdegree of each~$\avec{w}$. We then obtain a recursive bound on the size of such quasimodels from Dickson's Lemma.
		%

		\begin{lemma}\label{lem:recbound}
			There exists a quasimodel for $C_{0}$ iff there exists
			a quasimodel $\Qmf = (\mathfrak{F},\funcand,\runs)$ for $C_{0}$ with
			\begin{enumerate}
				\item $|\funcand(w_{0})|\leq 2^{|C_{0}|}$, for the root $w_{0}$ of $\Qmf$;
				\item if $\avec{v}$ is a successor of $\avec{w}$, then $|\funcand(\avec{v})|\leq |\funcand(\avec{w})|+2^{|C_{0}|}$;
				\item if $\avec{w}<\avec{v}$, then $\funcand(\avec{w})\not\leq \funcand(\avec{v})$;
				\item the outdgree of $\avec{w}$ in $\mathfrak{F}$ is bounded by
				$|\funcand(\avec{w})|\times |C_{0}|$. 
			\end{enumerate} 
		\end{lemma}  
		\begin{proof}
			Assume that a quasimodel $\Qmf$ for $C_{0}$ with $\Qmf = (\mathfrak{F},\funcand,\runs)$ and $\mathfrak{F}=(W,R_{1},\ldots,R_{n})$ with root $w_{0}$ is given. We may assume that every run $\rho\in \runs$ has a \emph{root} $\avec{w}\in W$ in the sense that $W_{\avec{w}}=\text{dom}(\rho)$. We let $\runs_{\avec{w}}$ denote the set of runs with root $\avec{w}$. We introduce a few rules that allow us to modify $\Qmf$.
			
			\medskip
			\noindent
			\emph{Minimize root:} Assume $\funcand(w_{0})=(x_{1},\ldots,x_{k})$. Pick, for any $x_{i}>0$,
			a single $\rho_{i}\in \runs$ with $\rho_{i}(w_{0})=i$. Now construct an updated $\Qmf'= (\mathfrak{F}',\funcand',\runs')$ by defining 
			\begin{itemize}
				\item $\mathfrak{F}'=\mathfrak{F}$, 
				\item $\funcand'(w_{0})= (\min\{x_{1},1\},\ldots,\min\{x_{k},1\})$ and $\funcand'(\avec{w})=\funcand(\avec{w})$ for all $\avec{w}\ne w_{0}$, 
				\item $\runs'$ as the set of selected runs $\rho_{i}$, all runs in $\runs\setminus\runs_{w_{0}}$, and 
				the restrictions of all remaining $\rho\in \runs_{w_{0}}$ to $W_{\avec{w}}$ for $\avec{w}$ a successor of $w_{0}$. 
			\end{itemize}
			It is easy to see that $\Qmf'$ is a quasimodel for $C_{0}$ satisfying Condition~1.
			
			\medskip
			\noindent
			\emph{Minimize non-root}: Assume $\avec{v}$ is a successor of $\avec{w}$ and $\funcand(\avec{v})=(x_{1},\ldots,x_{k})$. Call $i\leq k$ increasing at $\avec{v}$ if there exists
			$\rho_{i}\in \runs_{\avec{v}}$ such that $\rho_{i}(\avec{v})=i$. Select, for any  $i$ increasing at $\avec{v}$, such a $\rho_{i}\in \runs_{\avec{v}}$. We construct an updated $\Qmf'= (\mathfrak{F}',\funcand',\runs')$ by defining 
			\begin{itemize}
				\item $\mathfrak{F}'=\mathfrak{F}$, 
				\item $\funcand'(\avec{v})= (x_{i}',\ldots,x_{k}')$ with 
				$$
				x_{i}'=|\{\rho \in \runs \mid \avec{w}\in \text{dom}(\rho),\rho(\avec{v})=i\}| + 1
				$$
				if $i$ is increasing at $\avec{v}$ and $x_{i}'=x_{i}$ otherwise.
				Set $\funcand'(\avec{w}')=\funcand(\avec{w}')$ for all $\avec{w}'\ne \avec{v}$.
				\item $\runs'$ as the set of selected runs $\rho_{i}$, all runs in $\runs\setminus\runs_{\avec{v}}$, and 
				the restrictions of all remaining $\rho\in \runs_{\avec{v}}$ to $W_{\avec{v}'}$ for $\avec{v}'$ a successor of $\avec{v}$. 
			\end{itemize}
			It is easy to see that $\Qmf'$ is a quasimodel for $C_{0}$ satisfying Condition~2
			for $\avec{w}$ and $\avec{v}$.
			
			\medskip
			\noindent
			\emph{Drop interval 1}: Consider $\avec{w}=\avec{w}'iw\in W$ and $\avec{v}=\avec{v'}jv\in W$ with $\avec{w}<\avec{v}$ and $\funcand(\avec{w})\leq \funcand(\avec{v})$.
			Now construct an updated $\Qmf'= (\mathfrak{F}',\funcand',\runs')$ by defining 
			\begin{itemize}
				\item $\mathfrak{F}'$ is the frame determined by $W'$ with 
				$$
				W'= (W\setminus W_{\avec{w}})\cup \{ \avec{w}'i\avec{r} \mid \avec{v}'j\avec{r}\in W_{\avec{v}}\}.
				$$
				\item Set $\funcand'(\avec{u}) = \funcand(\avec{u})$ for all $\avec{u}\in W\setminus W_{\avec{w}}$ and $\funcand'(\avec{w}'i\avec{r})=\funcand( \avec{v}'j\avec{r})$ for all $\avec{v}'j\avec{r}\in W_{\avec{v}}$.
				\item Take, for any $\rho\in \runs$ with $\avec{w}'\in \text{dom}(\rho)$, an
				$f(\rho)\in \runs$ with $\avec{v}\in \text{dom}(f(\rho))$ such that
				$f(\rho)(\avec{v})=\rho(\avec{w})$. We can ensure that all $f(\rho)$ are distinct since $\funcand(\avec{w}) \leq \funcand(\avec{v})$. Then let $\rho'\in \runs'$ with $\rho'$ be defined by setting $\rho'(\avec{u})= \rho(\avec{u})$ for $\avec{u}\in W\setminus W_{\avec{w}}$ and $\rho'(\avec{w}'i\avec{r})=f(\rho)(\avec{v}'j\avec{r})$ for all $\avec{v}'j\avec{r}\in W_{\avec{v}}$.
				
				Also add all $\rho\in \runs$ with $\text{dom}(\rho)\cap W_{\avec{w}}=\emptyset$ to $\runs'$. Finally, for all $\rho$ with $\text{dom}(\rho)\subseteq W_{\avec{w}}$ which were not selected
				as an $f(\rho)$, add the $\rho'$ to $\runs'$ defined by setting $\rho'(\avec{w}'i\avec{r})=\rho(\avec{v}'j\avec{r})$ for all $\avec{v}'j\avec{r}\in W_{\avec{v}}$. 
			\end{itemize}
			It is easy to see that $\Qmf'$ is a quasimodel for $C_{0}$.
			
			\medskip
			\noindent
			\emph{Drop interval 2}: Consider the root $w_{0}$ and $\avec{v}\in W$ with $\funcand(w_{0})\leq \funcand(\avec{v})$.
			Then construct an updated $\Qmf'= (\mathfrak{F}',\funcand',\runs')$
			by restricting $\Qmf$ to $W_{\avec{v}}$ in the obvious way.
			
			\medskip
			\noindent
			\emph{Restrict outdegree}. Consider $\avec{w}\in W$ and assume $\funcand(\avec{w})=(x_{1},\ldots,x_{k})$. To bound the number of successors of $\avec{w}$ to $|\funcand(\avec{w})| \times |C_{0}|$ pick, for every $j\leq k$ with $x_{j}>0$,  every $\Diamond_{i}C\in \contp_{j}$ and every run $\rho\in \runs$ with $\rho(\avec{w})=j$, a successor $\avec{w}iw\in W$ with $C\in \contp_{\rho(\avec{w}iw)}$ or
			$\Diamond_{i}C\in \contp_{\rho(\avec{w}iw)}$. Let $V$ be the set of successors of $W$ picked in this way. Then construct an updated $\Qmf'= (\mathfrak{F}',\funcand',\runs')$
			by restricting $\Qmf$ to 
			$$
			W'= (W\setminus W_{\avec{w}}) \cup \{ \avec{w}\} \cup \bigcup_{\avec{v}\in V} W_{\avec{v}}.
			$$
			It is easy to see that $\Qmf'$ is a quasimodel for $C_{0}$ such that the outdegree of $\avec{w}$ is bounded by  $|\funcand(\avec{w})| \times |C_{0}|$.
			
			By applying the rules above exhaustively, we obtain a quasimodel satisfying the conditions of the lemma.
		\end{proof}
As argued above, it follows from Lemma~\ref{lem:recbound} that we have a recursive bound on the size of a quasimodel for $C_{0}$. By Lemma~\ref{th-GLnsat}, we then also have a recursive bound on the size of a finite tree-shaped model satisfying $C_{0}$. This shows the decidability of concept satisfiability for $\Kfn_{\ALCOu}$ in expanding domain models.

\subsection{Noetherian Frames}

We now discuss concept satisfiability for
$\GL_{\ALCOud}$ and $\Grz_{\ALCOud}$ in models with expanding domains. 
We focus on $\GL_{\ALCOud}$; the case of  $\Grz_{\ALCOud}$ can be treated similarly. Note that a straightforward selective filtration argument applied
to appropriate quasimodels shows that  every  $\GL_{\ALCOud}$-concept satisfiable in a model with expanding domains  is satisfied in a finite model with expanding domains. This observation implies that a $\MLALCOud$-concept $C$ with a single modal operator is satisfiable in an expanding domain model for $\GL_{\ALCOud}$ iff it is  satisfiable in a model for $\Kfn_{\ALCOud}$ with expanding domains, where the modal operator in $C$ is interpreted by the transitive closure of $R_{1}\cup \cdots \cup R_{n}$. Hence the decidability 
of $\GL_{\ALCOud}$-concept satisfiability in expanding domain models is a consequence of the decidability of $\Kfn_{\ALCOud}$-concept satisfiability in expanding domain models. 

The decidability of concept satisfiability for $\GL_{\ALCOud}$
in expanding domain models can also be derived from decidability results for expanding domain product modal logics.
We define the relevant modal logics using DL syntax. Consider the modal DL with a single role, $o$, interpreted as a Noetherian strict linear order. In DL syntax, concepts in $\mathcal{ML}^{n}_{lin}$ are of the form 
\begin{equation*}
	C ::= A \mid \lnot
	C \mid (C \sqcap C)
	\mid \exists o.C
	\mid \Diamond_i C,
\end{equation*}
where $i\in I$.
\begin{restatable}{theorem}{noether}
	\label{thm:noether}
	Let $\mathcal{C}$ be a class of frames.
	Then $\Cmc$-satisfiability of $\MLALCOu$-concepts can be reduced in double exponential time to $\Cmc$-satisfiability of $\mathcal{ML}^{n}_{lin}$-concepts, both with constant and expanding domains. 	
\end{restatable}
\begin{proof}    
	The proof is similar to the proof of Theorem~\ref{th:diff} (1).
	Assume an $\MLALCOu$-concept $C$ is given. Assume for simplicity that $C$ contains a single modal operator,  denoted by~$\Box$.
	Let $m$ be the modal depth of $C$.
	Let $\con{C}$ be the closure under single negation 
	of the set of concepts occurring in $C$. A \emph{type} for $C$ is a subset $\contp$ of $\con{C}$
	such that $\neg C \in \contp$ iff $C\not \in \contp$, for all $\neg C \in \con{\Omc}$. A \emph{quasistate for $C$} is a non-empty set $\settp$ of types for $C$. The \emph{description of $\settp$} is the $\MLALCOu$-concept 
	\begin{equation*}
		\Xi_{\settp} = \forall u.(\bigsqcup_{\contp\in \settp}\contp) \sqcap \bigsqcap_{\contp\in \settp} \exists u.\contp. 
	\end{equation*}
	Let $\Smc_{C}$ denote the set of all quasistates $\settp$ for $C$ that are
	\emph{$\ALCOu$-satisfiable}, that is, such that $\Xi_{\settp}^{t}$ is satisfiable, where $\Xi_{\settp}^{t}$ denotes the result of replacing 
	every outermost occurrence of $\Diamond_{i}D$ by $A_{\Diamond_{i}D}$, for a fresh concept name $A_{\Diamond_{i}D}$. We know that $\Smc_{C}$ can be computed in double exponential time.
	
	Let $\exists o^{+}.D=D \sqcup \exists o.D$ and $\forall o^{+}.D=D \sqcap \forall o.D$.
	We reserve for any $D\in \con{C}$
	of the form $\{a\}$ or $\exists r.D'$ or $\exists u.D'$ a fresh concept name $A_{D}$. Define a mapping $\cdot^{\sharp}$ that associates with every $\MLALCOu$-concept an $\mathcal{ML}^{n}_{lin}$-concept by replacing outermost occurrences of
	concepts of the form $\{a\}$ or $\exists r.D'$ or $\exists u. D'$ by the respective fresh concept name.
	
	Let $C^{\ast}$ denote the conjunction of:
	\begin{enumerate}
		\item $C^{\sharp}$;
		\item $\Box^{\leq m}\bigsqcup_{\settp\in \Smc_C}(\forall o^{+}.(\bigsqcup_{\contp\in \settp}\contp^{\sharp}) \sqcap \bigsqcap_{\contp\in \settp} \exists o^{+}.\contp^{\sharp})$;
		\item $\Box^{\leq m}\bigsqcap_{\exists u.D\in\con{C}}(A_{\exists u.D} \Leftrightarrow \exists o^{+}.D^{\sharp})$;
		\item $\Box^{\leq m} \exists o^{+}.\{a\}^{\sharp}$, for every $a$ in $C$;
		\item $\Box^{\leq m}\forall o^{+}.(\{a\}^{\sharp} \Rightarrow \forall o. \neg \{a\}^{\sharp})$, for every $a$ in $C$.
	\end{enumerate}
	The following lemma completes the proof of the reduction.
	\begin{lemma}
		$C$ is $\Cmc$-satisfiable iff $C^{\ast}$ 
		is $\Cmc$-satisfiable in a $\mathcal{ML}^{n}_{lin}$-model, with constant and expanding domains.
	\end{lemma}	\qedhere
\end{proof}
Now decidability of $\GL_{\ALCOud}$-concept satisfiability 
in expanding domain models follows from decidability of satisfiability of $\mathcal{ML}^{1}_{lin}$-concepts in models with expanding domains and transitive Noetherian frames~\cite{DBLP:journals/apal/GabelaiaKWZ06}.

%
%


\section{Proofs for Section~\ref{sec:reasontfdl}}

The
following two theorems share the same proof below.

\tempone*

\temponeexp*
\begin{proof}
\newcommand{\TLDiff}{\mathcal{TL}^{\Diamond}_{\text{Diff}}}
The proofs of these two theorems rely on the following
results: the problem of satisfiability of $\TLDiff$ concepts is
known~\cite{HamKur15} to be
\begin{enumerate}
\item\label{hk-1} $\Sigma_1^1$-complete in constant-domain interpretations over $(\mathbb{N},<)$;
\item\label{hk-2} recursively enumerable, but undecidable in constant-domain interpretations  over finite flows of time;
\item\label{hk-3} undecidable (co-r.e.) in expanding-domains interpretations  over $(\mathbb{N},<)$;
\item\label{hk-4} decidable, but Ackermann-hard in expanding-domains interpretations  over finite flows of time.
\end{enumerate}
As a consequence of Theorem~\ref{th:diff}~{(2)} and results~{1--3}  above, we obtain the following: in
constant-domain interpretations, concept satisfiability is
$\Sigma_1^1$-hard for $\LTLd_{\ALCOu}$ and undecidable for
$\LTLfd_{\ALCOu}$, while in expanding-domains interpretations, the problem is
undecidable for $\LTLd_{\ALCOu}$. As a consequence of result~{4},
concept
satisfiability in expanding-domains interpretations is Ackermann-hard for $\LTLfd_{\ALCOu}$.

To prove the results for the languages without the universal modality, we first replace the satisfiability problem for a given concept, say, $C$, with the satisfiability problem for a concept name $A$ under the global ontology $\Omc = \{ A \equiv C\}$. Next, we transform global ontology $\Omc$ into normal form; see Lemma~\ref{lemma:ontology:normal-form}.
Next, the reduction in Lemma~\ref{lemma:spy-point-reduction} allows us
to obtain the undecidability (including $\Sigma_1^1$-hardness) and
Ackermann-hardness results for the problem of concept satisfiability
under global ontology in languages without the universal role.

The membership in $\Sigma^1_1$ follows from the fact that the problem can be encoded in second-order arithmetic with existential quantifiers over binary predicate symbols for concept names and ternary predicate symbols for role names.

To show decidability of concept satisfiability
  \textup{(}under global ontology\textup{)} for $\LTLf_{\ALCOud}$, we first eliminate definite descriptions; see
  Proposition~\ref{lemma:redmludtomlu}. The result then is a 
  consequence of Theorem~\ref{thm:noether} and  the decidability results proved for concept satisfiability in the language with both `next' and `eventually'~\cite[Theorem 3]{DBLP:journals/apal/GabelaiaKWZ06}; note that for $\LTLfd_{\ALCOu}$ this follows more directly from Theorem~\ref{th:diff}~{(1)} and the result~{4} above.
\end{proof}

\temptwo*
\begin{proof}
We prove these results by reduction from the satisfiability under global ontology in the languages based on $\ALC$. 
Let $A$ be a concept name and $\Omc$ a  $\TLALCOud$ ontology.  By Lemma~\ref{lemma:ontology:normal-form}, we assume that $\Omc$ is in normal form. First, we eliminate occurrences of $\neg$: every CI of the form $B_1 \sqsubseteq \neg B_2$ is replaced by $B_1 \sqcap B_2 \sqsubseteq \bot$ and every CI of the form $\neg B_1 \sqsubseteq B_2$ is first replaced by $\top \sqsubseteq B_1 \sqcup B_2$. Then, we eliminate the disjunction from the C by using the method suggested in~\cite{AKLWZ07}: each CI of the form $\top \sqsubseteq B_1 \sqcup B_2$ is
replaced with the following:
\begin{align*}
\top & \sqsubseteq\exists q.(\Diamond X_1 \sqcap \Diamond X_2),\\
\exists q.\Diamond (X_1 \sqcap \Diamond X_2) & \sqsubseteq B_1,\\
\exists q.\Diamond (X_1 \sqcap X_2) & \sqsubseteq B_1,\\
\exists q.\Diamond (X_2 \sqcap \Diamond X_1) & \sqsubseteq B_2,
\end{align*}
where $q$ is a fresh role name, and $X_1$ and $X_2$ are fresh concept
names.
%
One can clearly see that every interpretation $\Mmf'$ satisfying the four
CI above also satisfies $\top \sqsubseteq B_1 \sqcup B_2$. Indeed, for every $t\in\mathbb{N}$ and every $d\in\Delta^t$, there is a $d'\in\Delta^t$ such that $d'\in(\Diamond X_i)^{\Imc_t}$, for $i = 1,2$. Thus, $d'\in X_i^{\Imc_{t_i}}$, for $t_i > t$ and $i = 1,2$. Depending on whether $t_1 \geq t_2$ or not, we have $d\in B_i^{\Imc_t}$ for either $i = 1$ or $i = 2$, thus satisfying    $\top \sqsubseteq B_1 \sqcup B_2$.

Conversely, if $\top \sqsubseteq B_1 \sqcup B_2$ is satisfied at all $w\in W$ in some $\Mmf$ over $(\mathbb{N}, <)$ such that the $\Delta^t$ is countably infinite, for each $t\in\mathbb{N}$, then $\Mmf$ can be extended by interpreting the fresh symbols to satisfy the 
four CIs above: for every $t\in\mathbb{N}$ and every $d\in\Delta^t$, we pick a unique element $d'\in\Delta^t$ and make it the $q^{\Imc_t}$-successor of $d$; since the domains are countably infinite, we can choose these $q^{\Imc_t}$-successor in such a way that no $d'$ is the $q^{\Imc_{t_1}}$-successor of $d_1$  and the  $q^{\Imc_{t_2}}$-successor of $d_2$ for distinct pairs $(t_1,d_1)$ and $(t_2, d_2)$.  It now remains to define the interpretation of $X_1$ and $X_2$: for every $t\in\mathbb{N}$ and every $d\in\Delta^t$, let $d'$ be the $q^{\Imc_t}$-successor of $d$; if $d\in B_1^{\Imc_t}$, then we include $d'$ in $X_1^{\Imc_{t+1}}$ and in $X_2^{\Imc_{t+2}}$; otherwise, we include $d'$ in $X_1^{\Imc_{t+2}}$ and in $X_2^{\Imc_{t+1}}$. The interpretation of $X_1$ and $X_2$ is well-defined due to our assumption  on the interpretation of $q$.

A similar transformation works for the finite flows, but we need to reserve two last instants for the interpretation of $X_1$ and $X_2$. To achieve this, we `relativise' all CIs in a given ontology $\Omc$: every CI of the form $C_1 \sqsubseteq C_2$ is replaced by $\Diamond\Diamond\top \sqcap C_1 \sqsubseteq C_2$ (alternatively, $\Next\Next \top$ can be used). It should be clear that $\Omc$ is satisfied in an interpretation based on a flow $([0,n], <)$ iff  the relativised ontology $\Omc'$  is satisfied in an interpretation based on a flow $([0, n+2], <)$; the two interpretations coincide on the first $n$ instants, but the additional two instants are simply ignored by the relativised CIs and so, all concept and role names can be assumed to be empty there. Now, we can apply the procedure for eliminating the negation described above (note that the two additional time instants can now be used to interpret $X_1$ and $X_2$ in any suitable way).

So, we now have a  $\TLALCOud$ ontology $\Omc'$ in normal form without any occurrences of the $\neg$ operator. We have also shown that $A$ is satisfied under $\Omc$ in an interpretation based on an infinite (respectively, finite) flow and having countably infinite domains iff $A$ is satisfied under $\Omc'$  in an interpretation based on an infinite (respectively, finite) flow and having countably infinite domains. It remains to eliminate $\bot$. To this end, we again use the method suggested in~\cite{AKLWZ07} and take a fresh concept name $L$ and replace all occurrences of $\bot$ with $L$. In addition, we extend $\Omc'$ with the following CIs:
\begin{align*}
\exists s.L & \sqsubseteq L, \text{ for all role names } s \text{ in } \Omc',\\
\Diamond L & \sqsubseteq L.
\end{align*}
Denote the resulting ontology by $\Omc''$. We can easily show that $A$ is satisfied under $\Omc'$ iff $A\sqsubseteq L$ is \emph{not} entailed by~$\Omc''$. 
\end{proof}


\tempthree*

\newcommand{\run}{\rho}

\begin{proof}
Let $C_0$ be an $\TL^{\textstyle\circ}_{\ALCOu}$ concept.  By Proposition~\ref{lemma:redpartialtototal}, it is sufficient to consider total satisfiability.
	Set $d=\md(C_{0})$. Observe that $C_0$ is satisfiable in an interpretation of the form $(\Delta,(\Imc_{t})_{t\in [0,n]})$ with $n\in \Nbl$ iff
	it is satisfiable in an interpretation of the form  $(\Delta,(\Imc_{t})_{t\in [0,n]})$ with $n\leq d$. Fix $m_{0}\leq d$. Below, we give an exponential-time subroutine that checks satisfiability in an interpretation of the form
	$(\Delta,(\mathcal{I}_{t})_{t\in [0,m_{0}]})$.
	By going through all $m_{0}\leq d$, we obtain an exponential-time algorithm that checks satisfiability for $\LTLfo_{\ALCOu}$ in constant domains, which, by Proposition~\ref{prop:redexptoconst}, also gives the same upper bound for the problem in expanding domains. An algorithm for  $\LTLo_{\ALCOu}$ in constant domains simply checks satisfiability in an interpretation of the form $(\Delta,(\mathcal{I}_{t})_{t\in [0,m_{0}]})$ for $m_0 = d$; the case of expanding domains is again due to  Proposition~\ref{prop:redexptoconst}.
	
	As before, a \emph{type} for $C_{0}$ is a subset $\contp$ of $\con{C_{0}}$
such that
\begin{enumerate}
  [label=\textbf{C\arabic*},leftmargin=*,series=run]
\item
  $\neg C \in \contp$ iff $C\not \in \contp$, for all
  $\neg C \in \con{C_{0}}$;
  \label{ct:neg}
\item
  $C \sqcap D \in \contp$ iff $C, D \in \contp$, for all
  $C \sqcap D \in \con{C_{0}}$.
  \label{ct:con}
\end{enumerate}
Denote by $\boldsymbol{T}$ the set of types for $C_0$. We clearly have $|\boldsymbol{T}| \leq 2^{|\con{C_0}|}$.
Given an interpretation $\Mmf=(\Delta,(\Imc_{t})_{t\in [0,m_{0}]})$, denote by $\contp^{\Imc_i}(d)$ the type realised by $d\in \Delta$ in $\Imc_i$, for $i \leq m_0$.

	The subroutine is an elimination procedure on the set  of all $(m_0 + 1)$-tuples $\run$ of types $(\contp_{0},\cdots,\contp_{m_{0}})$, called \emph{runs}. Runs can be thought of as elements in $\boldsymbol{T}^{m_0 + 1}$ and so the number of runs is bounded by $2^{|\con{C_0}| (m_0 + 1)}$. The components of a run $\run$ are denoted $\run(0), \dots, \run(m_0)$.	We write 
	$$
	\run \rightarrow_{r,i} \run'
	$$
	if $C\in \run'(i)$ implies $\exists r.C\in \run(i)$
	for all $\exists r.C\in \con{C_{0}}$.
Given an interpretation $\Mmf=(\Delta,(\Imc_{t})_{t\in [0,m_{0}]})$, denote by $\run^{\Mmf}_d=(\contp^{\Imc_0}(d),\ldots,\contp^{\Imc_{m_0}}(d))$ the run realised by $d\in\Delta$ in~$\Mmf$.

Let $a_{1},\ldots,a_{n}$ be the nominals in $C_{0}$. To construct the initial set $\runs_0$ of runs, we fix two parameters:
\begin{itemize}
\item an $(m_0+1)$-tuple $(U^{0},\ldots,U^{m_{0}})$ of
  \emph{$u$-types}, where each $U^i$ is a set containing either
  $\exists u.C$ or $\neg\exists u.C$, for every
  $\exists u.C\in \con{C_{0}}$;
\item a set $\mathfrak{N}$ of at most $n (m_0 + 1)$ \emph{nominal runs} $\run\in \boldsymbol{T}^{m_0 + 1}$ such that
\begin{itemize}
\item $U^i \subseteq \run(i)$, for all $\run\in\mathfrak{N}$ and $i \leq m_0$, and 
%
\item for every $\run\in\mathfrak{N}$, there is $i \leq m_0$ such that $\{a_j\}\in\run(i)$, for some $1 \leq j \leq n$;
\item for every $i \leq m_0$ and every $1 \leq j \leq n$, there is
  exactly one nominal run $\run\in\mathfrak{N}$ such that $\{a_j\}\in\run(i)$.
\end{itemize}
\end{itemize}	
%
	%
%
%
%
Then we take $\runs_0 = \mathfrak{N} \cup \mathfrak{O}$, where
$\mathfrak{O}$ is the set of \emph{all} runs without nominals
consistent with the fixed $u$-types, that is, all runs
$\run\in \boldsymbol{T}^{m_0 + 1}$ such that
%
\begin{itemize}
\item $U^{i} \subseteq \run(i)$, for all $i \leq m_0$, and
\item $\{ a_j \}\notin \run(i)$, for all $i \leq m_0$ and $1 \leq j \leq n$.
\end{itemize}
We now exhaustively apply the run elimination procedure and to remove
each run $\run$ from $\runs_0$ that satisfies one of the following:
\begin{itemize}
\item there is
  $i\leq m_{0}$ and $\Next C\in \run(i)$ and either $i=m_{0}$ or
  $i<m_{0}$ but $C\notin \run(i+1)$;
\item there is
  $i<m_{0}$ and $C\in \run(i+1)$ but $\Next C\notin \run(i)$, for
  $\Next C\in \con{C_{0}}$;
\item there is $i\leq m_{0}$ and $\exists u.C\in \run(i)$ but there is
  no $\run'$ with $C\in \run'(i)$;
\item there is $i\leq m_{0}$ and $\exists r.C\in \run(i)$ but there is
  no $\run'$ with $C\in \run'(i)$ and $\run\rightarrow_{r,i}\run'$.
\end{itemize}
Assume that $\runs$ is the remaining set of runs.
\begin{lemma}
  The following conditions are equivalent:
	
  (1) $C_{0}$ is satisfiable in a total interpretation
  $\Mmf=(\Delta,(\Imc_{t})_{t\in [0,m_{0}]})$ such that
  \begin{itemize}
	\item $U^i = \{ \exists u.C\in \con{C_{0}} \mid (\exists
          u.C)^{\Imc_i}=\Delta\} \cup$\\
          $\hspace*{2.5em}\{ \neg\exists u.C\in \con{C_{0}} \mid (\exists
          u.C)^{\Imc_i}=\emptyset\}$, for $i\leq m_{0}$\textup{;}
	\item $\mathfrak{N} = \{\run_{a_j^{\Imc_i}}^\Mmf \mid i \leq m_0, 1 \leq j\leq n\}$\textup{.}
\end{itemize}
(2) the set $\runs$ of runs satisfies the following conditions\textup{:}
\begin{itemize}
\item $\mathfrak{N} \subseteq \runs$\textup{;} 
\item $C_{0}\in \run(0)$ for some $\run\in \runs$.
\end{itemize}
\end{lemma}	
\begin{proof}	 
	If (1) holds, then take a witness total interpretation $\Mmf$. Let $\runs'$ be the set of runs $\run_d^\Mmf$
	for $d\in \Delta$. It is easy to see that 
$\runs\supseteq \runs'\supseteq \mathfrak{N}$. Moreover, $\run^{\Mmf}_{d_{0}}$,  for some $d_{0}\in C_{0}^{\Imc_{0}}$, provides the requited witness run for~(2).

Conversely, assume that (2) holds. Construct a total interpretation $\Mmf$ by setting:
\begin{itemize}
	\item $\Delta = \runs$;
	\item $a^{\Imc_i}=\run$ iff $\{a\} \in \run(i)$, for $i\leq m_{0}$;
	\item $\run \in A^{\Imc_{i}}$ iff $A\in \run(i)$, for $i\leq m_{0}$;
	\item $(\run,\run')\in r^{\Imc_i}$ iff $\run\rightarrow_{r,i}\run'$, for $i\leq m_{0}$.
\end{itemize}
It is now easy to show that $\Mmf$ is well-defined and as required for (1). 
\end{proof}
The exponential-time subroutine  for satisfiability in interpretations of length $m_{0}+1$ is obtained by going through all $u$-types and sets of nominal runs (of which there are only exponentially many).   
\end{proof}

\tempfour*
\begin{proof}
  The result
  is obtained by a reduction of the (undecidable) halting problem for
  (two-counter) Minsky machines starting with $0$ as initial counter
  values~\cite{DegEtAl02,BaaEtAl17}.

  A \emph{(two-counter) Minsky machine} is a pair $M = (Q, P)$, where
  $Q = \{ q_{0}, \ldots, q_{L} \}$ is a set of \emph{states} and
  $P = ( I_{0}, \ldots, I_{L\mbox{-}1} )$ is a sequence of
  \emph{instructions}.  We assume $q_{0}$ to be the \emph{initial
    state}, and $q_{L}$ to be the \emph{halting state}.  Moreover, the
  instruction $I_{i}$ is executed at state $q_{i}$, for
  $0 \leq i < L$.  Each instruction $I$ can have one of the following
  forms, where $r \in \{ r_{1}, r_{2} \}$ is a \emph{register} that
  stores non-negative integers as values,
and $q,q'$ are states:
\begin{enumerate}
	\item $I = +(r,q)$ -- add $1$ to the value of register $r$ and go to state $q$;
	\item $I = -(r,q,q')$ -- if the value of register $r$ is (strictly) greater than $0$, subtract $1$ from the value of $r$ and go to state $q$; otherwise, go to state $q'$.
\end{enumerate}
Given a Minsky machine $M$,
a \emph{configuration of $M$} is a triple $(q, v_1, v_2)$, where $q$ is a state of $M$ and $v_1, v_2 \in \mathbb{N}$ are the \emph{values} of registers $r_1, r_2$, respectively.
In the following, we denote $\overline{k} = {3 - k}$, for
$k \in \{ 1, 2 \}$.  Given $i, j \geq 0$, we write
$(q_{i}, v_1, v_2) \Rightarrow_{M} (q_{j}, v'_1,
v'_2)$ iff one of the following conditions hold:
\begin{itemize}
\item $I_{i} = +(r_k, q_{j})$, $v'_k = v_k + 1$ and
  $v'_{\overline{k}} = v_{\overline{k}}$;
\item $I_{i} = -(r_k, q_{j}, q')$, $v_k > 0$, $v'_k = v_k - 1$,
  and $v'_{\overline{k}} = v_{\overline{k}}$;
\item $I_{i} = -(r_k, q', q_{j})$, $v_k = v'_k = 0$, and
  $v'_{\overline{k}} = v_{\overline{k}}$.
\end{itemize}
Given an \emph{input} $(v_1, v_2) \in \mathbb{N} \times \mathbb{N}$, the
\emph{computation of $M$ on input $(v_1,v_2)$} is the (unique) longest
sequence of configurations
$(q^{0}, v^{0}_{1}, v^{0}_{2}) \Rightarrow_{M} (q^{1},
v^{1}_{1}, v^{1}_{2}) \Rightarrow_{M} \ldots$, such that
$q^{0} = q_{0}$, $v^{0}_{1} = v_1$ and $v^{0}_{2} = v_2$.
We say that $M$ \emph{halts} \emph{on input} $(0,0)$ if the computation of $M$ on input $(0,0)$ is finite: thus, its initial configuration takes the form $(q_{0}, 0, 0)$, while its last configuration has the form $(q_{L}, v_{1}, v_{2})$, for some $v_1,v_2\in\mathbb{N}$. 
The \emph{halting problem for Minsky machines} is the problem of
deciding, given a Minsky machine $M$, whether $M$ halts on input
$(0,0)$.
This problem is known to be undecidable; see, e.g.,~\cite{DegEtAl02,BaaEtAl17}.

To represent the computation of a Minsky machine $M$ on input $(0,0)$, we use the
temporal dimension to model successive configurations in the
computation. We introduce a concept name~$Q_{i}$, for each state $q_{i}$ of $M$. Concept names $R_{1}, R_{2}$
are used to represent registers $r_{1}, r_{2}$, respectively: the value of register $r_k$ at 
any computation step will coincide with the cardinality of the extension of $R_k$ at the corresponding moment of time.  
The incrementation of the
value of register $r_{k}$ is modelled by requiring that the extension of
concept $\lnot R_{k} \sqcap \Next R_{k}$ is included in that of nominal
$\{ a_k \}$;  dually, the decrementation of register $r_k$ is modelled by including the extension of
$R_{k} \sqcap \Next \lnot R_{k}$ in the nominal $\{ b_{k} \}$. 
Since the individual names $a_k$ and~$b_k$ are interpreted
\emph{non-rigidly}, the (unique) element added or subtracted from the
extension of $R_{k}$ will vary over time.  

We first construct a $\TLnALCOu$ ontology $\Omc_M$ such that 
\begin{itemize}
\item $Q_0$ is satisfiable under $\Omc_M$ over $(\mathbb{N},<)$ iff machine $M$ never halts on $(0,0)$; 
\item $Q_0$ is satisfiable under $\Omc_M$ over finite flows of time iff machine $M$ halts on $(0,0)$; 
\end{itemize}
Then, Lemmas~\ref{lemma:ontology:normal-form} and~\ref{lemma:spy-point-reduction} can be used to replace the ontology $\Omc_M$ with a $\TLnALCO$ ontology $\Omc_M$ with the same property, thus establishing the required results.

We use the following CIs to ensure that concepts $Q_i$ representing states are either empty or coincide with the whole domain at any moment of time:
  \begin{enumerate}
    [label=$(\textbf{S\arabic*})$,
    align=left,leftmargin=*,series=run]
  \item\label{undec-next-state-cover}
   $\exists u.Q_i \sqsubseteq Q_i$, for all $0 \leq i \leq L$;
  \item\label{undec-next-state-disjoint}
   $Q_i \sqcap Q_j \sqsubseteq \bot$, for all $0 \leq i<  j \leq L$.  
  \end{enumerate}

Next, we encode instructions. We begin with incrementation. An instruction of the form
$I_{i} = +(r_{k}, q_{j})$ is represented by the  following CIs:
\begin{enumerate}
  [label=$(\textbf{I\arabic*})$, align=left, leftmargin=*,series=run]
\item\label{undec-inc-nominal} $Q_i \sqsubseteq \exists u. (\{ a_k \} \sqcap \neg R_k)$;
\item\label{undec-inc-counter}
  $Q_i \sqsubseteq \Next R_k \Leftrightarrow (R_k \sqcup \{ a_k \})$;
\item\label{undec-inc-other}
  $Q_{i} \sqsubseteq \Next R_{\overline{k}} \Leftrightarrow  R_{\overline{k}}$;
\item\label{undec-inc-state} $Q_{i} \sqsubseteq \Next Q_{j}$.
\end{enumerate}
An instructions of the form $I_{i} = -(r_{k}, q_{j}, q_{h})$ is
represented by the following CIs:
\begin{enumerate}
  [label=$(\textbf{D\arabic*})$, align=left, leftmargin=*,series=run]
\item\label{undec-dec-nominal}
  $Q_i \sqcap \exists u.R_k \sqsubseteq \exists u. (\{b_k\}\sqcap R_k)$;
\item\label{undec-dec-counter}
  $Q_i \sqsubseteq \Next R_k \Leftrightarrow (R_k \sqcap \neg \{ b_k \})$;
\item\label{undec-dec-other}
 $Q_{i} \sqsubseteq  \Next R_{\overline{k}} \Leftrightarrow R_{\overline{k}}$;
\item\label{undec-dec-state}
  $Q_{i} \sqcap \exists u.R_k \sqsubseteq \Next Q_{j}$;
\item\label{undec-dec-state-empty}
 $Q_i \sqcap \neg \exists u.R_k \sqsubseteq \Next Q_{h}$.
\end{enumerate}

Given $\Mmf = (\Delta, (\Imc_{t})_{t \in \Tmf})$, we say that an instant $t\in\Tmf$ \emph{encodes} a configuration $(q_i,v_1,v_2)$ whenever $Q_i^{\Imc_t} = \Delta$ and $|R_k^{\Imc_t}| = v_k$, for $k = 1,2$. We show the following:
\begin{claim}\label{claim:undec:subseq:states}
For any $\Mmf = (\Delta, (\Imc_{t})_{t \in \Tmf})$ that satisfies~\ref{undec-next-state-cover}--\ref{undec-dec-state-empty}, if $t\in\Tmf$ encodes $(q_i,v_1,v_2)$ and $(q_i, v_{1}, v_{2}) \Rightarrow_{M} (q', v'_{1}, v'_{2})$, then $t+1$ encodes $(q', v'_{1}, v'_{2})$.
\end{claim}
\begin{proof}[Proof of Claim]
We consider the following cases.
\begin{itemize}
\item Suppose $I_{i} = +(r_{k}, q_{j})$. By~\ref{undec-inc-state} with~\ref{undec-next-state-cover}, we have $Q_j^{\Imc_{t+1}}= \Delta$, and thus
$q' = q_{j}$. Note that, by~\ref{undec-next-state-disjoint}, $Q_{j'}^{\Imc_{t+1}} = \emptyset$ for any $j' \ne j$. Then,
by~\ref{undec-inc-nominal}, $a_k^{\Imc_{t}}$ is defined and $a_k^{\Imc_{t}}\notin R_k^{\Imc_t}$. By~\ref{undec-inc-counter},  $R_k^{\Imc_{t+1}} = R_k^{\Imc_t} \cup \{a_k^{\Imc_{t}}\}$, whence $|R_k^{\Imc_{t+1}}| = |R_k^{\Imc_t}| + 1$, and so $v'_k = v_k + 1$. For the other counter, 
by~\ref{undec-inc-other}, $R_{\overline{k}}^{\Imc_{t + 1}} = R_{\overline{k}}^{\Imc_{t}}$, whence $v'_{\overline{k}} = v_{\overline{k}}$.
\item Suppose $I_{i} = -(r_{k}, q_{j}, q_{h})$. By~\ref{undec-dec-state} with~\ref{undec-next-state-cover}, if $R_k^{\Imc_t}\ne\emptyset$, then $Q_j^{\Imc_{t+1}}= \Delta$, and thus
$q' = q_{j}$. Otherwise, if $R_k^{\Imc_t}=\emptyset$, then, by~\ref{undec-dec-state-empty} with~\ref{undec-next-state-cover}, $Q_h^{\Imc_{t+1}}= \Delta$, and thus
$q' = q_h$. Note that in either case, by~\ref{undec-next-state-disjoint}, $Q_{j'}^{\Imc_{t+1}} = \emptyset$ for any other $j'$. We next deal with the values of counters. If $R_k^{\Imc_t}\ne\emptyset$, then, by~\ref{undec-dec-nominal}, $b_k^{\Imc_t}$ is defined and $b_k^{\Imc_t} \in R_k^{\Imc_t}$. Then, by~\ref{undec-dec-counter}, we have $R_k^{\Imc_{t+1}} = R_k^{\Imc_t} \setminus \{ b_k^{\Imc_t}\}$. It follows that $|R_k^{\Imc_{t+1}}| = |R_k^{\Imc_t}| - 1$ and so $v'_k = v_k - 1$. On the other hand, if $R_k^{\Imc_t}=\emptyset$, then, by~\ref{undec-dec-counter}, $R_k^{\Imc_{t+1}} = \emptyset$ and so $v'_k = v_k = 0$. For the other counter, 
by~\ref{undec-dec-other}, $R_{\overline{k}}^{\Imc_{t + 1}} = R_{\overline{k}}^{\Imc_{t}}$, whence $v'_{\overline{k}} = v_{\overline{k}}$.
\end{itemize}
So, in each case, $t+1$ encodes $(q', v_1', v_2')$.
\end{proof}

It remains to encode the initial configuration and the acceptance
condition. For the initial configuration, we two CIs
  \begin{enumerate}
    [label=$(\textbf{O\arabic*})$,
    align=left,leftmargin=*,series=run]
  \item\label{undec-initial-count} 
  $Q_0 \sqcap R_k \sqsubseteq \bot$, for $k = 1,2$;
  \end{enumerate}
Clearly, these two CIs ensure that the value of both counters is
$0$ whenever the state is the initial state $q_0$.
For the non-termination condition over $(\mathbb{N}, <)$ and the termination condition over finite flows of time, we use the following CIs, respectively:
  \begin{enumerate}[label=$(\textbf{A\arabic*})$,
    align=left,leftmargin=*,series=run]
\item\label{undec-does-not-halt}
    $Q_{L}\sqsubseteq \bot$.
\item\label{undec-halts}
    $\neg \Next \top \sqsubseteq Q_{L}$.
\end{enumerate}
We claim that $Q_0$ is satisfiable under $\Omc_M$ over $(\mathbb{N},<)$ iff machine $M$ never halts on $(0,0)$. Indeed, if $Q_0$ is satisfiable under $\Omc_M$, then by~\ref{undec-initial-count}, moment 0 encodes configuration~$(q_0, 0,0)$. By Claim~\ref{claim:undec:subseq:states}, moments 0, 1, 2, etc. encode configurations $(q^t,v^t_1,v_2^t)$ such that $(q^t, v_{1}^t, v_{2}^t) \Rightarrow_{M} (q^{t+1}, v^{t+1}_{1}, v^{t+1}_{2})$, for each $t$. By~\ref{undec-does-not-halt}, the halting state $q_L$ is never reached, and so machine~$M$ does not halt on $(0,0)$.

Conversely,  $(q^0, v^0_{1}, v^0_{2}) \Rightarrow_{M} \ldots \Rightarrow_{M} (q^t,
  v^t_{1}, v^t_{2})  \Rightarrow_{M} \ldots $ be the computation of $M$ on input
  $(0,0)$, where $q^{0} = q_{0}$ and $v^0_k = 0$, for $k = 1,2$.  We
  construct an interpretation 
$\Mmf = (\Delta, (\Imc_{t})_{t \in \mathbb{N}})$ that satisfies $Q_0$ under $\Omc_M$. 
%
Let $\Delta$ be a fixed countable set.  Given
$t\in \mathbb{N}$ and $q^t = q_j$, for some $0 \leq i \leq L$,
we set $Q_{i}^{\Imc_{t}} = \Delta$ and $Q_{j}^{\Imc_{t}} = \emptyset$, for all $j\neq i$.
In addition, we set $R_{1}^{\Imc_{0}} = R_{2}^{\Imc_{0}} = \emptyset$.
Then, for every instant $t\in\mathbb{N}$, we define the interpretations
$R_{1}^{\Imc_{t+1}}$, $R_{2}^{\Imc_{t+1}}$, $a_k^{\Imc_{t}}$,
and $b_k^{\Imc_{t}}$ inductively as follows.
\begin{itemize}
\item If $q^t = q_i$ and $I_{i} = +(r_{k}, q_{j})$, then we choose some $d\notin R_{k}^{\Imc_{t}}$ and set $a_k^{\Imc_{t}} = d$ and 
   $R_{k}^{\Imc_{t + 1}} = R_{k}^{\Imc_{t}} \cup \{ d\}$;
  we also choose an arbitrary $b_{k}^{\Imc_{t}}$ and set
  $R_{\overline{k}}^{\Imc_{t + 1}} = R_{\overline{k}}^{\Imc_{t}}$.
\item If $q^t = q_{i}$ and $I_{i} = -(r_{k},q_{j},q_{h})$, then we have two options:
  \begin{itemize}
  \item if $R_{k}^{\Imc_{t}} \neq \emptyset$, then we choose an arbitrary $e\in R_k^{\Imc_t}$ and set $b_k^{\Imc_{t}} = e$ and
    $R_{k}^{\Imc_{t + 1}} = R_{k}^{\Imc_{t}} \setminus \{ e \}$;
      \item if $R_{k}^{\Imc_{t}} = \emptyset$, then we choose an arbitrary $b_k^{\Imc_{t}}$ and set
        $R_k^{\Imc_{t + 1}} = R_{1}^{\Imc_{t}}$.
      \end{itemize}
      In either case, we also choose an arbitrary $a_k^{\Imc_{t}}$ and set 
	$R_{\overline{k}}^{\Imc_{t + 1}} = R_{\overline{k}}^{\Imc_{t}}$.
    \end{itemize}
It can be seen that $\Mmf$ satisfies $Q_0$ under $\Omc_M$.

\smallskip

Next, we claim that $Q_0$ is satisfiable under $\Omc_M$ over finite flows of time iff machine $M$ halts on $(0,0)$.  Indeed, if $Q_0$ is satisfiable under $\Omc_M$ over some $([0,n], <)$, then by~\ref{undec-initial-count}, moment 0 encodes configuration~$(q_0, 0,0)$. By Claim~\ref{claim:undec:subseq:states}, moments $0, 1, 2, \dots, n$ etc. encode configurations $(q^t,v^t_1,v_2^t)$ such that $(q^t, v_{1}^t, v_{2}^t) \Rightarrow_{M} (q^{t+1}, v^{t+1}_{1}, v^{t+1}_{2})$, for each $0 \leq t < n$. By~\ref{undec-halts}, the state $q^n$ is the halting state, $q_L$, as required. The converse direction is similar to the case of $(\mathbb{N}, <)$.
\end{proof}

\end{document}